\documentclass[11pt]{article}
\usepackage{amsmath, amssymb, amsthm,bbm}
\usepackage{natbib} 
\usepackage{xcolor} 
\usepackage{authblk}

\definecolor{lblue}{RGB}{40, 103, 178}
\definecolor{cred}{RGB}{177, 4, 14}
\definecolor{sgreen}{RGB}{46, 139, 87}

\usepackage{algorithm,algorithmic,mathtools}

\usepackage{calc} \usepackage{enumerate} \usepackage{parskip}
\usepackage{longtable}
\usepackage{setspace}
\usepackage{array}
\usepackage{multirow}
\usepackage{makecell}
\usepackage{comment}
\usepackage{import}
\usepackage{booktabs}
\usepackage{tabularx}
\usepackage{tikz}
\usetikzlibrary{shapes.multipart}
\usepackage[font=small]{caption}

\usetikzlibrary{arrows.meta,decorations.pathreplacing}

\usepackage{fullpage,anyfontsize}
\allowdisplaybreaks

\RequirePackage[colorlinks,linkcolor=blue,citecolor=blue,urlcolor=blue,breaklinks]{hyperref}
\usepackage{cleveref}
\newtheorem{theorem}{Theorem}
\newtheorem{lemma}{Lemma}

\newtheorem{remark}{Remark}
\newtheorem{definition}{Definition}
\newtheorem{corollary}{Corollary}

\newtheorem{assumption}{Assumption}

\makeatletter
\newtheorem*{rep@theorem}{\rep@title}
\newcommand{\newreptheorem}[2]{\newenvironment{rep#1}[1]{ \def\rep@title{#2 \ref{##1} (full version)} \begin{rep@theorem}} {\end{rep@theorem}}}
\makeatother

\newreptheorem{theorem}{Theorem}
\definecolor{lblue}{RGB}{40, 103, 178}
\definecolor{cred}{RGB}{177, 4, 14}
\definecolor{sgreen}{RGB}{46, 139, 87}

\def\E {{\mathbb{E}}}
\def\P {{\mathbb{P}}}
\def\RR {{\mathbb{R}}}

\def\independentraw#1#2{\mathrel{\rlap{$#1#2$}\mkern2mu{#1#2}}}
\newcommand\independent{\protect\mathpalette{\protect\independentraw}{\perp}}
\newcommand{\Cov}{\mathrm{Cov}}
\newcommand{\EEst}[2]{\mathbb{E}\left[{#1}\  \middle| \ {#2}\right]}
\newcommand{\EE}[1]{\mathbb{E}\left[{#1}\right]}

\newcommand{\One}[1]{{\mathbbm{1}}\left\{{#1}\right\}}
\newcommand{\PPst}[2]{\mathbb{P}\left({#1}\  \middle| \ {#2}\right)}
\newcommand{\PP}[1]{\mathbb{P}\left({#1}\right)}

\newcommand{\Var}{\mathrm{Var}}
\newcommand{\tr}{\mathrm{tr}}

\newcommand{\bX}{\mathbf{X}}
\newcommand{\bY}{\mathbf{Y}}
\newcommand{\bZ}{\mathbf{Z}}

\newcommand{\bx}{\mathbf{x}}
\newcommand{\by}{\mathbf{y}}

\newcommand{\cF}{\mathcal{F}}
\newcommand{\cS}{\mathcal{S}}

\newcommand{\dks}{\textnormal{d}_{\textnormal{KS}}}
\newcommand{\dtv}{\textnormal{d}_{\textnormal{TV}}}
\newcommand{\dw}{\textnormal{d}_{\textnormal{W}}}
\newcommand{\eqd}{\stackrel{\textnormal{d}}{=}}
\newcommand{\iidsim}{\stackrel{\textnormal{iid}}{\sim}}

\usepackage{algorithm,algorithmic,mathtools}

\title{Local permutation tests for conditional independence: an adaptive binning perspective}
\author[1]{David Chen}
\author[2]{Rohan Hore}
\author[1]{Rina Foygel Barber}
\affil[1]{Department of Statistics, University of Chicago}
\affil[2]{Department of Statistics and Data Science, Carnegie Mellon University}
\date{\today}

\begin{document}

\maketitle

\begin{abstract}
    In this work, we study the problem of testing conditional independence between random variables $X$ and $Y$ given a confounder $Z$. The local permutation test (LPT) from \cite{kim2022local} offers a principled approach to this problem by partitioning the $Z$-space into pre-specified bins, and permuting the $X$ and $Y$ data within each bin, to assess the significance of an observed test statistic. However, when the partitions are pre-fixed, the resulting partition can be poorly balanced, as some bins may contain most of the samples while others contain only a few. This motivates the use of data-adaptive binning strategies, such as equisized bins with a fixed (typically small) number of points. We study this natural and practically important extension of LPT, providing finite-sample bounds on the Type I error for an arbitrary test statistic, providing stronger validity results than previously known. We also show that LPT attains power comparable to the oracle likelihood ratio tests derived from the Neyman--Pearson lemma. Within a linear confounder model class, we further analyze the effect of bin size and demonstrate that constant bin sizes can match the performance of partitions with growing bin-size. These results, further supported by extensive numerical simulations, position the proposed data-adaptive strategy as both practically implementable and statistically efficient.
\end{abstract}

\section{Introduction}

In this work, we study the problem of testing conditional independence (CI). Formally, suppose we observe data points $(X_1,Y_1,Z_1),\ldots,(X_n,Y_n,Z_n) \in \mathcal{X} \times \mathcal{Y} \times \mathcal{Z}$, drawn i.i.d.\ from an unknown joint distribution $P$. Our goal is to test the null hypothesis
\[
H_0^{\mathrm{CI}}: \; X \independent Y \mid Z .
\]
Here, $X$ and $Y$ denote the primary variables of interest, for instance, a response variable $Y$ and associated covariate $X$, while $Z$ represents a confounder. Note that $X$, $Y$, and $Z$ may each be multi-dimensional, and in particular, it is common in practical applications to have a high-dimensional confounder $Z$. Throughout this paper, we write $P_{X \mid Z}$, $P_{Y \mid Z}$, and $P_{X,Y \mid Z}$ for the conditional distributions of $X$, $Y$, and $(X,Y)$ given $Z$, respectively.

The problem of testing conditional independence has been studied for a long time in statistics, with early developments focusing on partial correlations and contingency table methods \citep{fisher1924035,agresti2012categorical}. Over time, conditional independence became a standard concept in areas such as graphical modeling and causal inference \citep{koller2009probabilistic}. More recently, it has also appeared in a range of modern statistical problems, including variable selection and high-dimensional inference \citep{williamson2021nonparametric,dai2022significance}. 

While this problem has long been of classical interest in the statistical literature, it has also become increasingly relevant in various parts of the modern machine learning literature, including algorithmic fairness \citep{hardt2016equality,kosorok2019annual}, and invariant representation learning \citep{PogodinDLSVG23,hwa2024}.

\subsection{The hardness of CI testing}\label{sec:hardness_CI}

Despite a rich literature in statistics, constructing valid and powerful tests for $H_0^{\mathrm{CI}}$ remains a fundamentally challenging problem without imposing any additional distributional assumptions. This difficulty was first theoretically established by \citet[Theorem 2]{shah2020hardness}; stated informally, their result shows that
\begin{quote} 
    \emph{Any test for $H_0^{\mathrm{CI}}$ that achieves finite-sample level $\alpha$ control uniformly over all Lebesgue-continuous null distributions must have power no greater than $\alpha$ (i.e., no better than random) against any Lebesgue-continuous alternative.}
\end{quote}
Subsequently, \citet{kim2022local} established a similar hardness result, extending the analysis beyond purely continuous models and arriving at a similar conclusion. Taken together, these results highlight a fundamental limitation of distribution-free CI testing: one \emph{must} impose additional structural or distributional assumptions to obtain meaningful power against alternatives of interest.

As a consequence, much of the existing literature has focused on achieving Type~I error control under restricted classes of null distributions. These restrictions take various forms, including assuming a known parametric model \citep{kalisch2007estimating,colombo12}, access to a known or accurately estimated conditional distribution $P_{X \mid Z}$ \citep{candes2018panning,barber2020robust,berrett2020conditional,niu2024reconciling}, additional structural constraints such as shape restrictions \citep{hore2025testing}, or smoothness of the conditional distributions $P_{X \mid Z}$ and $P_{Y \mid Z}$ \citep{shah2020hardness,kim2022local,lundborg2022projected}. In what follows, we focus on this last approach and review how smoothness assumptions can be leveraged for CI testing.

\subsection{Local permutation tests for conditional independence}\label{sec:review_kim_paper}

We briefly review the \emph{local permutation test} procedure proposed by \citet{kim2022local}, which shows how smoothness assumptions on the conditional distributions can be leveraged to obtain approximate Type~I error control for testing $H_0^{\mathrm{CI}}$. The core idea is to discretize the conditioning variable $Z$ and reduce the problem to conditional independence testing with a discrete confounding variable.

Specifically, fix a partition of the space $\mathcal{Z}$ into bins, $\mathcal{Z} = \bigcup_{k=1}^K \mathcal{Z}_{k}$, and let $\tilde Z_i$ denote the discretized version of $Z_i$, defined as
\[
\tilde Z_i = k \iff Z_i \in \mathcal{Z}_{k}.
\]
We now consider a different hypothesis of conditional independence,
\[
\tilde{H}_0^{\mathrm{CI}} : \; X \independent Y \mid \tilde Z .
\]
Intuitively, if the bins $\mathcal{Z}_{k}$ are each relatively small (i.e., the number of bins $K$ is large), we might expect that the original null hypothesis $H_0^{\mathrm{CI}}$ is similar, in some sense, to its modified version $\tilde{H}_0^{\mathrm{CI}}$---that is, since $\tilde{Z}$ contains nearly the same information as the original un-discretized $Z$, we might expect that a joint distribution $P$ that satisfies $H_0^{\mathrm{CI}}$ will also approximately satisfy $\tilde{H}_0^{\mathrm{CI}}$, and vice versa.

Importantly, due to the discretization, testing $\tilde{H}_0^{\mathrm{CI}}$ no longer faces the same hardness result of \citet{shah2020hardness}. Since $\tilde Z$ is discrete, testing $\tilde{H}_0^{\mathrm{CI}}$ can be carried out in a finite-sample valid manner using a permutation-based procedure that permutes values of $X$ and $Y$ within each bin, giving us the local permutation test (LPT) with corresponding $p$-value 
\[
    p = \frac{1}{|\Pi|} \sum_{\sigma \in \Pi} \One{T(\mathbf X_\sigma, \mathbf Y) \geq T(\mathbf X, \mathbf Y)}.
\]
Here $\Pi$ is the group of all permutations that permute entries within bins only, and $T$ is some chosen test statistic; see Section \ref{sec:methodology} for details.
 
Note that the conditional distributions given $\tilde Z$ correspond to a local averaging of the conditional distributions given $Z$. In particular, writing  $P_{(X,Y)\mid \tilde Z}$ to denote the conditional distribution of $(X,Y)$ given $\tilde{Z}$, we observe that 
\[
\mathsf{d}P_{(X,Y)\mid \tilde Z}
= \int \mathsf{d}P_{(X,Y)\mid Z=z} \, \mathsf{d}P_{Z\mid \tilde Z}(z),
\]
with analogous representations for $P_{X\mid \tilde Z}$ and $P_{Y\mid \tilde Z}$. Under the null $H_0^{\mathrm{CI}}$ and suitable smoothness assumptions on $P_{X\mid Z}$ and $P_{Y\mid Z}$, this local averaging approximately preserves the conditional independence structure when the bins are sufficiently fine. As a result, the permutation test for $\tilde{H}_0^{\mathrm{CI}}$ yields approximate finite-sample Type~I error control for the original null $H_0^{\mathrm{CI}}$.

While this construction is natural, the LPT framework raises several practical questions.
\begin{itemize}
    \item First, given a prescribed number of bins $K$, it is not obvious how to construct an appropriate partition of the space $\mathcal Z$. This challenge is particularly pronounced in high-dimensional settings, where naive partitions can lead to a large number of empty or sparsely populated bins. This naturally raises the question of whether the partition can be chosen in a data-adaptive manner, while still retaining approximate Type~I error control.
    \item Second, the choice of $K$ governs an inherent tradeoff in the procedure. Using finer partitions (equivalently, larger values of $K$) leads to tighter control of the approximation error in replacing $Z$ by $\tilde Z$---that is, the null hypothesis $\tilde{H}_0^{\mathrm{CI}}$ being tested by the procedure is more similar to the original null hypothesis of interest, $H_0^{\mathrm{CI}}$. But, a larger value of $K$ may reduce power due to fewer observations within each bin. In \citet{kim2022local}, the authors select $K \asymp n^{2/5}$ to obtain asymptotically optimal power guarantees. However, this choice leaves open the possibility that substantially larger values of $K$ (with possibly data-adaptive partitions) may still yield satisfactory power in practice, offering significant computational and statistical advantages.
\end{itemize}

From a practical standpoint, one can always perform sample splitting: the first split may be used to construct a partition of $\mathcal{Z}$, while the second split is used to implement the LPT framework. This extension would still retain the existing Type~I error control guarantees, but may be inefficient in several respects. If the partition is chosen to satisfy $K=\mathrm{o}(n)$, then optimizing over all feasible partitions to identify a ``good'' one may itself be computationally expensive. On the other hand, if one wishes to ensure $K=\Theta(n)$, then the second split will typically contain many bins with very few or no observations, substantially reducing the effective sample size available for inference and thereby leading to a corresponding loss in power.

\subsection{Our contributions}
In this work, we provide concrete answers to the questions raised above by proposing a natural generalization of the LPT framework that allows for data-adaptive bins, which we call the adaptive-LPT. We provide detailed theoretical analysis to demonstrate its practicality over the fixed bin LPT framework.
Our main contributions are summarized as follows:
\begin{itemize}
    \item 
We first reexamine the Type I error of the local permutation tests, conditional on observed $Z$ samples, and give two bounds on the Type I error rate: the first, in Section \ref{sec:validity_general}, applies to any bin based test statistic, generalizing the validity result of \cite{kim2022local} to data-adaptive binning procedures and thereby sharpening their finite sample bound based on total variation distances. The second guarantee, in Section \ref{sec:validity_linear}, considers the widely adopted family of statistics which are linearly decomposable into a sum of individual bin level statistics; this design choice, in turn, enables a stronger bound on excess Type I error. 

\item Our second main contribution is a detailed power analysis of LPTs, benchmarked against oracle likelihood ratio tests based on the Neyman--Pearson lemma. We show in Section \ref{sec:power_NNPT} that restricting to LPT tests (in fact, even restricting further to bins of size $2$) incurs at most a constant loss of power whenever conditional distribution of $(X,Y)$ given $Z$ is smooth in $Z$. Moreover, Section \ref{sec:linear_model}, by restricting to a linear confounding model, gives guidance on the choice of powerful test statistics and binning strategies, and formally establishes that these choices are indeed near-optimal.
\end{itemize}

\paragraph{Organization of the paper.}
The rest of the paper is organized as follows. Section \ref{sec:methodology} formally defines the adaptive-LPT testing procedure, and in Section \ref{sec:validity}, we develop corresponding bounds on Type I error. Then, in Section \ref{sec:power}, we give our power guarantees, and  Section \ref{sec:sims} includes multiple numerical simulations to illustrate our theoretical results. Most proofs are deferred to the appendix. In Section \ref{sec:power}, we also defer complicated technical conditions (and corresponding discussion on those conditions) to the appendix for expository reasons. Additional discussion, including some generalization of our results, are also available in Appendix \ref{sec:additional}.

\section{Methodology}\label{sec:methodology}
In this section, we formally introduce the extension of LPT, from Section~\ref{sec:review_kim_paper}, which allows data-adaptive partitions, i.e., one may look at the observed $Z$ values to construct an efficient and balanced partition of the samples. 

Before describing our method, we introduce some notation.
We adopt the following asymptotic shorthand: for sequences $a_n$ and $b_n$, we write $a_n \ll b_n$ and $a_n = \mathrm o(b_n)$ interchangeably, $a_n \lesssim b_n$ and $a_n = \mathrm O(b_n)$ interchangeably, and $a_n \asymp b_n$ to mean $a_n \lesssim b_n$ and $b_n \lesssim a_n$ hold interchangeably.
Now let $B_1,\dots,B_K\subseteq[n]$ be a collection of disjoint bins.\footnote{Often these bins form a partition (i.e., $B_1\cup\dots\cup B_K=[n]$), but it may sometimes be the case that $B_1\cup\dots\cup B_K$ is a strict subset of $[n]$---for example, if we wish to have bins of equal size $|B_k|=m$ then the union of the bins contains only $m\cdot \lfloor n/m\rfloor$ data points. Since in practice, the bins cover all or nearly all of the data indices, we will typically refer to $B_1,\dots,B_K$ as a `partition' of $[n]$ even though this may not be strictly the case.} We will use $m_k = |B_k|$ to denote the size of the bin $B_k$, and when $m_1=\cdots=m_K$, we will simply write $m$ for the shared bin size.
We write $\mathbf X=(X_1, X_2, \dots, X_n)$, $\mathbf Y=(Y_1,\dots,Y_n)$, and $\mathbf Z=(Z_1,\dots,Z_n)$ to denote the full vectors of data points. Given a partition $\{B_1,\ldots,B_K\}$, we write $\mathbf X_k$ and $\mathbf Y_k$ to be the subvectors of $\mathbf X$ and $\mathbf Y$ corresponding to the indices in $B_k = \{i_1, \dots, i_{m_k}\}$, (and similar for any deterministic vectors $\mathbf x$ or $\mathbf y$). We write $\mathcal S_m$ for the group of permutations on $[m]=\{1,\dots,m\}$, and for any permutation $\sigma \in \mathcal S_m$ and vector $\mathbf x = (x_1, \dots, x_m) \in \mathbb R^m$, $\mathbf x_\sigma$ is the vector $(x_{\sigma(1)}, \dots, x_{\sigma(m)})$. Note that this convention satisfies $(\mathbf x_\sigma)_{\sigma'} = \mathbf x_{\sigma \circ \sigma'}$ for repeated permutations. Finally, for a subvector $\mathbf x_k$, we write $(\mathbf x_k)_\sigma$ for the vector given by permuting the indices of ${\mathbf x}_k$ according to $\sigma$ restricted to the indices in $B_k$ (we will only work with $\sigma$ such that $\sigma(B_k) = B_k$).

\subsection{Local permutation tests with data-adaptive binning}
Now, we generalize the existing LPT procedure, and demonstrate how we can use it for testing the conditional independence null hypothesis. We will refer to this procedure as the \emph{adaptive-LPT} from now onwards. The adaptive-LPT is comprised of two ingredients: a partition into (data-adaptive) bins, and a choice of a test statistic.

\paragraph{(i) Data-adaptive bins.} First, with access to $\mathbf Z$ but not $\mathbf X$ or $\mathbf Y$, the analyst chooses the bins $B_1,\dots,B_K\subseteq[n]$ with the goal that for any $k\in [K]$,
  \[
    P_{X \mid Z_i} \approx P_{X \mid Z_j} \text{ and }
    P_{Y \mid Z_i} \approx P_{Y \mid Z_j} \quad \textnormal{for}~ i,j\in B_k.
  \]
  One particularly natural way of doing this is to bin together $Z$ values which are close, i.e., $d(Z_i,Z_j)\approx 0$ for some measure of distance or dissimilarity $d$ on $\mathcal{Z}$. This is a natural extension of the conditional permutation test with discretely supported $Z$, where we would usually group the samples with the same $Z$ value, and then apply permutation within each of the groups independently.
  
Given such a partition, we can now define the permutations that we will consider in the test:
  \begin{definition}[Bin-preserving permutation]
  Given bins $B_1, \dots, B_K$, we say that a permutation $\sigma \in \mathcal S_n$ is bin-preserving if $\sigma(B_k) = B_k$ for all $k \in [K]$. We further write $\Pi = \Pi(B_1,\dots,B_K)\subseteq\mathcal{S}_n$ to denote the subgroup of all bin-preserving permutations.
\end{definition}

We remark that a prespecified (non-data-adaptive) binning scheme, as in the work of \citet{kim2022local} (recall Section~\ref{sec:review_kim_paper}), is a special case: given a prespecified partition $\mathcal{Z} = \bigcup_{k=1}^K \mathcal{Z}_{(k)}$, we can simply define $B_k = \{i : Z_i \in\mathcal{Z}_{(k)}\}$. However, in adaptive-LPT, we have substantially more flexibility, since the bins may depend on the observed values $Z_1,\dots,Z_n$. For instance, we may choose to construct bins $B_k$ of a fixed size---e.g., as we will explore below, we might choose to group the observed values $Z_1,\dots,Z_n$ into $\lfloor n/2\rfloor$ pairs, in which case we have $m_k=2$ for each $k$.
  
\paragraph{(ii) A test statistic.}
Next, we choose a test statistic $T(\mathbf X,\mathbf Y,\mathbf Z)$, which is a function measuring some notion of evidence against the null $H_0^{\mathrm{CI}}$ (we will explore concrete examples shortly). We require that $T$ respects a symmetry assumption within each partition:

\begin{definition}[Bin-symmetric statistic]
  We say a function $T: \mathcal X^n \times \mathcal Y^n \times \mathcal Z^n \to \mathbb R$ is bin-symmetric, if for any bin-preserving permutation $\sigma \in \Pi$,
  \begin{equation}\label{eq:symmetry}
    T(\mathbf x,\mathbf y,\mathbf z) = T(\mathbf x_\sigma,\mathbf y_\sigma, \mathbf z).
  \end{equation}
\end{definition}
Being bin-symmetric is equivalent to enforcing that $T$ depends only on the collection of $(X, Y)$ pairs within each bin and is fully agnostic towards their original ordering. On the other hand, one can use the full vector $\mathbf Z$ however one likes to construct the statistic, as long as it does not violate the necessary symmetry condition.

\paragraph{Computing a p-value.} Finally, given these two components, we define the $p$-value by
  \begin{equation}\label{eq:pval_adaptive_LPT}
       p = \frac{1}{|\Pi|}\sum_{\sigma \in \Pi} \One{T(\bX_{\sigma},\bY,\bZ) \geq T(\mathbf X, \mathbf Y, \mathbf Z)}.
  \end{equation}
   Note that, since $T$ is bin-symmetric, we can equivalently define this p-value as
  \[
    p = \frac{1}{|\Pi|}\sum_{\sigma \in \Pi} \One{T(\bX, \bY_{\sigma},\bZ) \geq T(\mathbf X, \mathbf Y, \mathbf Z)}.
  \]
  In other words, $\mathbf X$ and $\mathbf Y$ are being treated symmetrically in this testing procedure. We (arbitrarily) choose to use the first notation (i.e., permuting $\mathbf X$ rather than $\mathbf Y$) from this point on.

\subsection{Concrete examples of adaptive-LPT}\label{sec:examples_Sec2}
Now, we provide two simple examples of our general framework, demonstrating how one might define the bins and the test statistic to implement the adaptive-LPT test. We will study the approximate validity and the power of these examples later on, but here we provide only the constructions, for intuition on how the adaptive-LPT test might be run in practice.

\subsubsection{Example: conditional covariance}\label{sec:examples_Sec2_covariance}
For the first example, suppose that both $X$ and $Y$ are real-valued, i.e., $\mathcal X, \mathcal Y \subseteq \mathbb R$. In this setting, a natural measure of conditional dependence would be (the absolute value of) the conditional covariance, $|\mathrm{Cov}(X,Y\mid Z)|$, under the joint distribution $P$. 

To begin, we construct a partition $B_1\cup \dots \cup B_K$ such that the values $Z_i$ within each bin are nearby, and $|B_k|=m_k\geq 2$ for each bin; for example, we may use $K$-means clustering to form a partition of $Z_1,\dots,Z_n$. We can then choose the test statistic $T$ to provide an approximation to the average conditional covariance, by defining
\[T(\mathbf X,\mathbf Y, \mathbf Z)
 = \frac{1}{K}\sum_{k=1}^K \left|\frac{1}{m_k} \sum_{i\in B_k} (X_i - \overline{\mathbf X}_k)(Y_i - \overline{\mathbf Y}_k)\right|.\]
Here $\overline{\mathbf X}_k$ denotes the sample mean of the observations $\mathbf X_k$ within bin $B_k$, and same for $\overline{\mathbf Y}_k$. This test statistic is bin-symmetric, by construction, and is likely to be large under the alternative if the joint distribution $P$ exhibits high conditional covariance.

\subsubsection{Example: betting on pairs}\label{ex:betting}
In our next example, we now implement a betting-based strategy, where $X$ and $Y$ can now take values in any space. Suppose that, given two values $x,x'\in\mathcal X$ and two values $y,y'\in\mathcal Y$, we can place a ``bet'' on which pairing is more likely---$(x,y)$ and $(x',y')$, or, $(x,y')$ and $(x',y)$. For instance, if we believe that there is positive dependence between $X$ and $Y$ after conditioning on $Z$ (in the setting $\mathcal X = \mathcal Y= \mathbb R$), and $x>x'$ and $y>y'$, we would bet that the first pairing is more likely than the second. In general, define a function
\[T_{\mathrm{bet}}(x,x',y,y') \in \{-1, 0 , +1 \},\]
with a $+1$ indicating that we believe $(x,y),(x',y')$ is a more likely pairing than $(x,y'),(x',y)$, and a $-1$ indicating the opposite. (We must allow a value $0$ to accommodate settings such as ties, e.g., $x=x'$.) It is natural to require a symmetry condition,
\[T_{\mathrm{bet}}(x,x',y,y') = T_{\mathrm{bet}}(x',x,y',y) = -T_{\mathrm{bet}}(x,x',y',y) = -T_{\mathrm{bet}}(x',x,y,y').\]

Next, suppose that we choose bins of size $m_k=2$: we have $B_k = \{i_k,j_k\}$ for each $k=1,\dots,K$ where $K=\lfloor n/2\rfloor$, with the bins chosen so that $Z_{i_k}\approx Z_{j_k}$ for each $k$. For instance, if $\mathcal Z\subseteq \mathbb R$, we can simply sort the $Z$ values and use the ranking to determine the bins: if $Z_{(1)}\leq \dots \leq Z_{(n)}$ denote the order statistics, then we form a bin containing $Z_{(1)},Z_{(2)}$, another bin containing $Z_{(3)},Z_{(4)}$, etc.

We can then define our test statistic by aggregating all our ``bets'':
\[T(\mathbf X, \mathbf Y, \mathbf Z)
=\sum_{k=1}^K T_{\mathrm{bet}}(X_{i_k},X_{j_k},Y_{i_k},Y_{j_k}).\]
If in fact the null hypothesis $H_0^{\mathrm{CI}}$ holds, then we might expect that each bet $T_{\mathrm{bet}}(X_{i_k},X_{j_k},Y_{i_k},Y_{j_k})$ is approximately equally likely to take the value $+1$ or $-1$, for each $k$, leading to a test statistic $T(\mathbf X, \mathbf Y,\mathbf Z)$ that has mean $\approx 0$. On the other hand, under the alternative, if we are able to estimate the nature of the (conditional) dependence between $X$ and $Y$ then we may be able to place bets whose performance is better than random, leading to a large value of $T(\mathbf X,\mathbf Y, \mathbf Z)$.

Such use of nearest neighbors in $\mathcal Z$ to approximate the null distribution of test statistics under conditional independence is well documented in the literature \citep{sen2017model,runge2018conditional}. Here, we view such methods through the lens of the adaptive-LPT framework and strengthen existing asymptotic Type I guarantees by establishing finite-sample Type I error bounds.

\section{Validity of adaptive-LPT} \label{sec:validity}

In this section, we study the validity of the adaptive-LPT test; that is, we analyze the Type I error rate of our procedure under the null hypothesis $H_0^{\mathrm{CI}}$. Ideally, we would wish to establish bounds of the form
\[\textnormal{If $H_0^{\mathrm{CI}}$ holds then\ }\mathbb P(p\leq \alpha) \lessapprox \alpha.\]
Of course, the discussion on hardness result from Section~\ref{sec:hardness_CI} implies that this cannot be achieved by any test with nontrivial power. Instead, we will establish that this type of approximate Type I error control holds under mild conditions: namely, in settings where the conditional distributions $P_{X\mid Z}$ and/or $P_{Y\mid Z}$ are approximately constant within bins-- that is, for two confounder values $Z_i,Z_j$ in the same bin $B_k$, we have $P_{X\mid Z_i}\approx P_{X\mid Z_j}$ and/or $P_{Y\mid Z_i}\approx P_{Y\mid Z_j}$. 
In particular, these results will demonstrate that under mild regularity conditions on the conditional distributions, our test is asymptotically valid for testing the null hypothesis $H^{\mathrm{CI}}_0$: informally,
\[\textnormal{If $H^{\mathrm{CI}}_0$ holds, along with some regularity conditions, then }\mathbb P(p\leq \alpha) \lessapprox \alpha.\]

We will present two different bounds on Type I error, with the first result holding generally for any implementation of the adaptive-LPT test, while the second specializes to the practical implementation with test statistics that can be decomposed as a sum over bins.

\subsection{Type~I error control with general statistics}\label{sec:validity_general}
We now give our first validity result, which holds for any choice of the test statistic $T$.

\begin{theorem}\label{thm:validity_main}
Let $B_1\cup\dots\cup B_K$ be any data-adaptive partition, and let $T(\mathbf X,\mathbf Y,\mathbf Z)$ be any bin-symmetric test statistic. Then, for any $\alpha\in (0,1)$ and any distribution $P$ satisfying the null $H_0^\mathrm{CI}$, the p-value in~\eqref{eq:pval_adaptive_LPT} satisfies
\[
\mathbb P \left( p\leq \alpha \mid \mathbf Z \right) \leq \alpha + \delta_n,
\]
  where
  \[
  \delta_n:=4\sum_{k=1}^K  (m_k-1) \left( \max_{i, j \in B_k} \mathrm{d_{TV}}(P_{X \mid Z_i}, P_{X \mid Z_j}) \right) \left( \max_{i, j \in B_k} \mathrm{d_{TV}}(P_{Y \mid Z_i}, P_{Y \mid Z_j}) \right).
  \]
\end{theorem}
Here $\dtv$ refers to the total-variation (TV) distance.
We emphasize that this result holds for \textit{any} binning scheme which is a function of $\mathbf Z$, covering both the original LPT framework with pre-defined bins (e.g., by partitioning $\mathcal Z$) and our proposed extension, where one may instead use $\mathbf Z$ to choose the binning scheme so as to accommodate salient features such as underlying structure or observable heterogeneity across different regions of $\mathcal Z$.

For intuition, consider the special case of bins of size $m_k=2$, with $K=\lfloor n/2\rfloor$ bins. In this case, writing $B_k = \{i_k,j_k\}$ for each bin $k$, the offset term $\delta_n$ simplifies to
\begin{equation}\label{eq:delta_n_nnpt_case}
   \delta_n = 4\sum_{k=1}^{\lfloor n/2\rfloor}\mathrm{d_{TV}}(P_{X\mid Z_{i_k}},P_{X\mid Z_{j_k}})\cdot \mathrm{d_{TV}}(P_{Y\mid Z_{i_k}},P_{Y\mid Z_{j_k}}).
\end{equation}

Note that the product structure in $\delta_n$ enables a double-robustness result: approximate Type~I error control does not require both conditional distributions $P_{X\mid Z}$ and $P_{Y\mid Z}$ to be smooth within bins. Rather, it suffices that at least one of the two conditional laws varies slowly in $Z$ (i.e., either $\mathrm{d_{TV}}(P_{X\mid z},P_{X\mid z'})$ or $\mathrm{d_{TV}}(P_{Y\mid z},P_{Y\mid z'})$ is small, when $z\approx z'$), since the excess Type~I error depends only on the product of the corresponding within-bin discrepancies.

\begin{remark}
In the statement of Theorem~\ref{thm:validity_main}, the upper bound is presented in terms of TV distances to facilitate an interpretable version of the result. In fact, it is possible to prove a similar result for generalized Hellinger distances and R\'enyi divergences,  and our bounds recover as a special case of Theorems~2 and~3 of \cite{kim2022local}. See Appendix \ref{app:other_metrics} for full details.
\end{remark}

\subsubsection{Implications: Type~I error control under smoothness assumptions} \label{sec:validity_smooth}
To interpret our Type~I error control in Theorem~\ref{thm:validity_main},
we next discuss how smoothness conditions enable bounds on $\delta_n$, and thus on the resulting Type~I error control. In this part, we suppose that $\mathcal Z$ is a metric space with metric $d$, and that the conditional distributions are Lipschitz smooth in the TV distance:

\begin{definition}\label{defn:TV_ball}
  For any positive constant $L$, let $\mathcal P_{\mathrm{TV}}(L)$ denote the set of distributions $P$ on $(X,Y,Z)$ satisfying 
  \[
    \mathrm{d_{TV}}(P_{X \mid Z = z}, P_{X \mid Z = z'}) \leq L\,d(z, z'),
    \textnormal{ and }
    \mathrm{d_{TV}}(P_{Y \mid Z = z}, P_{Y \mid Z = z'}) \leq L\,d(z, z').
  \]
\end{definition}
This Lipschitz assumption on the conditional distributions allows us to bound the term $\delta_n$ appearing in Theorem~\ref{thm:validity_main} above, and thereby give Type~I error control, conditional on $\bZ$ and uniformly over the class $\mathrm P_{\mathrm{TV}}(L)$.

\begin{corollary}\label{cor:thm2_smooth}
  Under the setting and notation of Theorem~\ref{thm:validity_main}, suppose that there exists $h_n>0$ such that $\max_k \max_{i,j\in B_k} d(Z_i,Z_j)\leq h_n$ almost surely.
  Then, it holds that
  \[
    \sup_{P\in \mathcal P_{\mathrm{TV}}(L)\cap H_0^{\mathrm{CI}}}\P_{(\bX,\bY,\bZ)\sim P^n}(p \leq \alpha \mid \mathbf Z) \leq \alpha + 4nL^2 h_n^2\quad \text{a.e.}
  \]
\end{corollary}
In particular, this immediately yields an asymptotic Type~I error guarantee of adaptive-LPT: for any $L>0$, the adaptive-LPT test gives an asymptotic valid test for $H_0^{\mathrm{CI}}$ as long as the diameter of bins shrinks at a rate $\mathrm{o}_P(n^{-1/2})$. In comparison, the excess Type I error bound of \cite{kim2022local} (Theorem 2, specialized to TV distance) is $n^{1/2}Lh_n$ which, while asymptotically equivalent in terms of validity, vanishes at a much slower rate.

\subsubsection{A proof sketch of Theorem~\ref{thm:validity_main}}\label{sec:proof_sketch}
 Though the full proof is deferred to the appendix, here we give a brief overview of the key steps in our argument.
 Recall that the p-value $p$ is constructed by comparing the value of the test statistic, $T(\mathbf X,\mathbf Y,\mathbf Z)$, against its permuted copies, $T(\mathbf X_\sigma, \mathbf Y, \mathbf Z)$. In order for $p$ to be (approximately) valid, we therefore need to verify that under the null, $T(\mathbf X,\mathbf Y,\mathbf Z)$ has (approximately) the same distribution as $T(\mathbf X_\sigma, \mathbf Y, \mathbf Z)$, where $\sigma\sim\textnormal{Unif}(\Pi)$.

Let $\sigma,\sigma'\overset{iid}{\sim}\text{Unif}(\Pi)$ denote bin-preserving permutations sampled uniformly at random, and note that the test statistic satisfies
\[T(\bX, \bY, \bZ) = T(\bX_\sigma, \bY_\sigma, \bZ),\]
since $T$ is required to be bin-symmetric. On the other hand, the permuted test statistic satisfies
\[T(\bX_\sigma, \bY, \bZ) =T(\bX_{\sigma \circ \sigma'}, \bY_{\sigma'}, \bZ)\overset{D}=T(\bX_\sigma, \bY_{\sigma'}, \bZ),\]
where the second equality holds since $(\sigma \circ \sigma', \sigma' )\overset{D}{=}(\sigma,\sigma')$, while the first equality follows by bin-symmetry of $T$.
Now, conditional on $\bZ$, let $P_k(\bZ)$ denote the conditional distribution of $((\bX_k)_{\sigma_k},(\bY_k)_{\sigma_k})$ with $\sigma_k\sim\textnormal{Unif}(\cS_{m_k})$, and let $P^*_k(\bZ)$ denote the conditional distribution of $((\bX_k)_{\sigma_k},(\bY_k)_{\sigma'_k})$ with $\sigma_k,\sigma'_k\iidsim\textnormal{Unif}(\cS_{m_k})$. 
 Then by construction, we can verify that
 \[\big(\bX_\sigma, \bY_\sigma\big) \mid \bZ \sim \otimes_{k=1}^K P_k(\bZ),\quad\quad\big(\bX_\sigma, \bY_{\sigma'}\big) \mid \bZ \sim \otimes_{k=1}^K P^*_k(\bZ).\]
 Comparing this to our calculations above, we see that the approximate validity of our test relies on bounding the TV distance between the distributions $\otimes_{k=1}^K P_k(\bZ)$ and $\otimes_{k=1}^K P^*_k(\bZ)$. This is formalized in the following lemma:

\begin{lemma}\label{lem:initial_tv_bound}
  Under the above setting, the Type I error of adaptive-LPT, conditional on $\bZ$, is bounded as
  \[
  \mathbb P \left( p\leq \alpha \mid \mathbf Z \right) \leq \alpha + \sum_{k=1}^K \mathrm{d_{TV}}(P_k(\bZ), P_k^*(\bZ)).
  \]
\end{lemma}

The task now is to control each of the TV terms. The following lemma connects such TV terms to the product of TV distances between the conditional distributions of $X$ given $Z$ and of $Y$ given $Z$, and is the technical crux of the proof, requiring a new analysis of the effect of permutations in terms of TV distance.

\begin{lemma}{\label{lem:bin_tv_bound}}
  Under the above setting, it holds that for each $k\in[K]$,
  \[
    \mathrm{d_{TV}}(P_k(\bZ), P_k^*(\bZ)) \leq 4(m_k-1)\cdot \left( \max_{i, j \in B_k} \mathrm{d_{TV}}(P_{X \mid Z_i}, P_{X \mid Z_j}) \right) \left( \max_{i, j \in B_k} \mathrm{d_{TV}}(P_{Y \mid Z_i}, P_{Y \mid Z_j}) \right).
  \]
\end{lemma}

Note that the bound depends on the \textit{product} of TV distances of the $X$ and $Y$ conditional distributions, rather than the sum, which one might naively expect, leading to much stronger control on $\mathrm{d_{TV}}(P_k(\mathbf Z), P_k^*(\mathbf Z))$. In short, the bound arises from decomposing permutations into at most $m_k - 1$ many transpositions and then showing that any transposition---which swaps, say, indices $i$ and $j$---causes the distribution of the permuted vector to move at most $4\,\mathrm{d_{TV}}(P_{X \mid Z_i}, P_{X \mid Z_j}) \cdot \mathrm{d_{TV}}(P_{Y \mid Z_i}, P_{Y \mid Z_j})$ in TV distance.

Combining these two lemmas yields Theorem \ref{thm:validity_main}.

\subsection{Type~I error control with linearly decomposable statistics}\label{sec:validity_linear}

While the Type~I error control in Theorem~\ref{thm:validity_main} holds for any statistics $T$, many practical choices of $T$ often admit additional structures that further enable a simplification of the Type~I error guarantee. In particular, both examples considered in Section~\ref{sec:examples_Sec2} admit the following decomposition:
\begin{definition}\label{def:lin_decomp}
  We say that a test statistic $T$ is linearly decomposable if it can be written as a sum over bins,
  \[
    T(\mathbf x,\mathbf y,\mathbf z) = \sum_{k=1}^K T_k(\mathbf x_k,\mathbf y_k ,\mathbf z).
  \]
\end{definition}
In this section, we will see that under mild regularity assumptions on the summands $T_k$, we can obtain a tighter bound on the Type~I error of adaptive-LPT. In particular, this structure enables us to approximate the p-value in~\eqref{eq:pval_adaptive_LPT} by treating the permutation as acting independently across summands, thereby allowing the use of Berry--Esseen bounds to obtain a sharper control.
We define
  \[\textnormal{Range}_k = \max_{\sigma_k,\sigma'_k} T_k((\bX_k)_{\sigma_k},(\bY_k)_{\sigma'_k},\bZ) - \min_{\sigma_k,\sigma'_k} T_k((\bX_k)_{\sigma_k},(\bY_k)_{\sigma'_k},\bZ),\quad \text{and}\]\[\textnormal{Var}_k = \textnormal{Var}\big(T_k((\bX_k)_{\sigma_k},(\bY_k)_{\sigma'_k},\bZ)\mid \bX,\bY,\bZ\big),\]
  where the variance is computed with respect to the randomness of  $\sigma_k,\sigma'_k\overset{iid}{\sim} \text{Unif}(\cS_{m_k})$. 
With these definitions in place, we are now ready to state 
the next theorem that establishes a Type I error bound of adaptive-LPT, specifically for linearly decomposable statistics.

\begin{theorem}\label{thm:validity_linear}
 Under the setting of Theorem~\ref{thm:validity_main}, assume also that the test statistic $T$ is linearly decomposable. Then, the p-value in~\eqref{eq:pval_adaptive_LPT} satisfies
  \[
    \mathbb P \left( p\leq \alpha \mid \mathbf Z \right) \leq \alpha + 2\sqrt{\epsilon_n\delta_n},
  \]
  where we define
  \[
  \epsilon_n:=\EEst{\frac{\max_{k=1, \dots, K} \textnormal{Range}_k}{\sqrt{\sum_{k=1}^K \textnormal{Var}_k}}}{\bZ}.\] 
\end{theorem}
Combining Theorems \ref{thm:validity_main} and \ref{thm:validity_linear}, we get that for any linearly decomposable statistics, the Type I error of adaptive-LPT, conditional on $\bZ$ satisfies:
\[
  \mathbb P(p \leq \alpha \mid \mathbf Z) \leq \alpha + \min \{\delta_n, 2\sqrt{ \epsilon_n \delta_n} \}.
\]

This updated bound greatly strengthens the result of Theorem~\ref{thm:validity_main}:  the p-value $p$ is asymptotically valid as long as \emph{either} $\delta_n\to 0$ or $\epsilon_n\delta_n\to 0$. Furthermore, for most practical choices of $T$, the new error term $\epsilon_n$ is expected to scale as $\mathrm{O}_P(K^{-1/2})$. (For instance, in the betting example of Section~\ref{sec:examples_Sec2}, suppose  we assume $T_k$ takes values in $\{\pm 1\}$, i.e., no ``ties''. We then have $\textnormal{Range}_k=2$ and  $\textnormal{Var}_k = 1$, so that $\epsilon_n = \frac{2}{\sqrt{K}} = \frac{2}{\sqrt{\lfloor n/2\rfloor}} = \mathrm{O}_P(n^{-1/2})$.) In such settings, therefore, asymptotic validity holds as long as $\delta_n = \mathrm{o}_P(K^{1/2})$, rather than the stricter requirement $\delta_n = \mathrm{o}_P(1)$ if we rely only on Theorem~\ref{thm:validity_main}. 

However, since $\epsilon_n$ may not necessarily be small (e.g., if $T_k(\bX_k,\bY_k,\bZ)$ is heavy-tailed, and its range is large relative to its variance), the result of Theorem~\ref{thm:validity_main} may nonetheless be more favorable in certain scenarios, even for linearly decomposable test statistics.

\subsubsection{Implications for bin size}\label{sec:Thm3_bin_size}
To illustrate the benefit of this refined guarantee, we now consider its implications for bin size: how large of a bin can we afford to use, without losing asymptotic validity? 

In this section, for simplicity, we consider the setting $\mathcal Z = [0,1]$, with bins $B_1,\dots,B_K$ defined via the prespecified partition $\mathcal Z = [0,\frac{1}{K}]\cup (\frac{1}{K},\frac{2}{K}]\cup\dots\cup(\frac{K-1}{K},1]$. Recalling the notation of Section~\ref{sec:validity_smooth}, the diameter of the bins is therefore bounded by $h_n = \frac{1}{K}$. We will assume that $\epsilon_n = O(K^{-1/2})$, as discussed above.

Consider choosing the number of bins as $K\propto n^\nu$. What exponent $\nu$ should we choose to maintain asymptotic validity? Under the Lipschitz smoothness condition of Definition~\ref{defn:TV_ball}, following the same calculations as in Section~\ref{sec:validity_smooth}, we have 
\[\delta_n =O\left( \frac{nL^2}{K^2}\right).\]
Therefore, Theorem~\ref{thm:validity_main} yields asymptotic validity as long as we choose $K\propto n^\nu$ for some $\nu>\frac{1}{2}$. In particular, setting $K\propto n^{1/2}$ (a common choice in practice) does not lead to any guarantee with this theorem.

On the other hand, by Theorem~\ref{thm:validity_linear}, the excess Type I error can also be bounded by the term
\[\sqrt{\epsilon_n\delta_n} = \sqrt{O(K^{-1/2}) \cdot O\left(\frac{nL^2}{K^2}\right)} = O\left( \frac{n^{1/2}L}{K^{5/4}}\right).\]
This means that asymptotic validity is ensured when $K\propto n^\nu$ for $\nu>\frac{2}{5}$, which allows for wider bins; in particular, $K\propto n^{1/2}$ yields asymptotic validity. 
Comparing the two bounds, we see that Theorem~\ref{thm:validity_linear} leads to tighter control than Theorem~\ref{thm:validity_main} (that is, $\sqrt{\epsilon_n\delta_n}\ll \delta_n$) whenever $\nu< \frac{2}{3}$, while when $\nu>\frac{2}{3}$ the original theorem gives the better bound.
These observations are illustrated in Figure~\ref{fig:bin_width}.

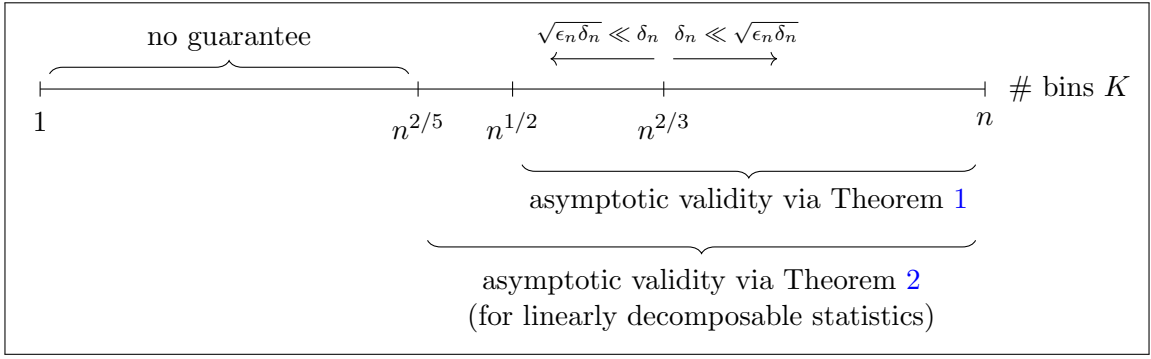
\begin{figure}[t]
  \centering
  \fbox{
  \begin{tikzpicture}[x=12.5cm]
    \draw[-]  (0,0) -- (1.0,0) node[right=5pt] {\# bins $K$};
    
    \draw[->] (0.67,0.4) -- (0.78,0.4) node[midway,above=2pt] {\scriptsize \quad $
    \delta_n\ll\sqrt{\epsilon_n \delta_n} $};

    \draw[->] (0.65,0.4) -- (0.54,0.4) node[midway,above=2pt] {\scriptsize $\sqrt{\epsilon_n \delta_n} \ll \delta_n$ \quad };

    \draw[decorate,decoration={brace,amplitude=5pt}]
    (0.01,.2) -- (0.39,.2) node[midway,above=5pt] {no guarantee};

    \draw[decorate,decoration={brace,amplitude=5pt}]
    (0.99,-1) -- (0.51,-1) node[midway,below=5pt] {asymptotic validity via Theorem~\ref{thm:validity_main}};
    
    \draw[decorate,decoration={brace,amplitude=5pt}]
    (0.99,-2) -- (0.41,-2) node[midway,below=5pt] {\begin{tabular}{c} asymptotic validity via Theorem~\ref{thm:validity_linear} \\ (for linearly decomposable statistics) \end{tabular}};

    \draw (0.4,0.1) -- (0.4,-0.1) node[below=2pt] {$n^{2/5}$};
    \draw (0.5,0.1) -- (0.5,-0.1) node[below=2pt] {$n^{1/2}$};
    \draw (0.66,0.1) -- (0.66,-0.1) node[below=2pt] {$n^{2/3}$};
    \draw (1,0.1) -- (1,-0.1) node[below=2pt] {$n$};
    \draw (0,0.1) -- (0,-0.1) node[below=2pt] {$1$};
  \end{tikzpicture}}
  \caption{An illustration of the discussion of Section~\ref{sec:Thm3_bin_size}, illustrating the range of $K$ (the number of bins) for which our theorems guarantee asymptotic validity, in the regime $\epsilon_n=O(K^{-1/2})$.}
\label{fig:bin_width}
\end{figure}

\subsubsection{Implications for smoothness}\label{sec:Thm3_smoothness}

In the last section, we studied the Type I error control for adaptive-LPT under
Lipschitz smoothness of the conditional distributions, which may be too strong in practice. Thus, we now
ask the converse question: for a fixed bin size (i.e., $K \propto n$), how do different smoothness regimes affect Type~I error control? In other words, how does
\[
\sup_{P\in \mathcal{P}\cap H_0^{\mathrm{CI}}}\P_{(\bX,\bY,\bZ)\sim P^n}(p \leq \alpha \mid \mathbf Z)
\]
vary across different choices of classes $\mathcal{P}$ encoding different levels of smoothness?

In this section, we restrict attention to $\mathcal{Z}=[0,1]$ and study the instance of adaptive-LPT with $K=\lfloor n/2\rfloor$ bins, where each bin has size $m_k=2$ and is of the form $B_k=\{i_k,j_k\}$, with $i_k$ and $j_k$ denoting neighboring points in the $\mathcal{Z}$ space. In this case, $\delta_n$ simplifies as in~\eqref{eq:delta_n_nnpt_case} and
\[
  \zeta_n := \frac{\delta_n}{K} = \frac{4}{K} \sum_{k=1}^K \mathrm{d_{TV}}(P_{X \mid Z_{i_k}}, P_{X \mid Z_{j_k}}) \mathrm{d_{TV}}(P_{Y \mid Z_{i_k}}, P_{Y \mid Z_{j_k}})
\]
can be interpreted as the deviation in a typical bin from the idealized setting in which the conditional distributions $P_{X \mid Z}$ and $P_{Y \mid Z}$ within each bin are identical; we work with $\zeta_n$ rather than $\delta_n$ to emphasize that we only require bins to be well-behaved ``on average''. Now any bound on $\zeta_n$ induces a class of distributions $\mathcal{P}$, with larger admissible values of $\zeta_n$ corresponding to broader classes $\mathcal{P}$ farther from the ideal.

Note that, under Lipschitz smoothness of the conditional distributions as in the previous section, if $\mathcal Z = [0,1]$ then we would expect to have $\zeta_n =O(L^2/n^2)$ for Lipschitz constant $L$ (since the diameter of each bin should be $|Z_{i_k}-Z_{j_k}| \propto n^{-1}$, on average over all bins). But under more degenerate distributions, we might expect a substantially larger $\zeta_n$.

The results of Theorem~\ref{thm:validity_main} then ensure asymptotic validity as long as $\zeta_n =o(n^{-1})$. However, in the case of a linearly decomposable test statistic, if we again assume $\epsilon_n=O(K^{-1/2})=O(n^{-1/2})$ as before, then the guarantee of Theorem~\ref{thm:validity_linear} yields asymptotic validity of the test in a wider regime, $\zeta_n = o(n^{-1/2})$, since $\sqrt{\epsilon_n\delta_n} = O(n^{1/4}\zeta_n^{1/2})$. More generally, Theorem~\ref{thm:validity_linear} yields a better bound than Theorem~\ref{thm:validity_main} when $\zeta_n \gg n^{-3/2}$. These observations are illustrated in Figure~\ref{fig:smoothness}.

We emphasize that in this example (and in general), no smoothness condition (such as Lipschitzness) need hold uniformly over the entire space of conditional distributions $P_{X \mid Z}$ and $P_{Y \mid Z}$; we merely need the corresponding distributions to be sufficiently well-behaved \textit{on average}, since we only need to have control over the average total variation distance term $\zeta_n$.

\begin{figure}[t]
  \centering
  \fbox{
  \begin{tikzpicture}[x=6cm]
    \draw[{Latex}-{Latex}]  (-2.0,0) -- (0.0,0) node[right] {$\zeta_n$};
    
    \draw[->] (-1.54,0.4) -- (-1.7,0.4) node[midway,above=2pt] {\scriptsize $
    \delta_n\ll\sqrt{\epsilon_n \delta_n} $ \quad\quad\quad };
    \draw[->] (-1.46,0.4) -- (-1.3,0.4) node[midway,above=2pt] {\scriptsize \quad\quad $\sqrt{\epsilon_n \delta_n} \ll \delta_n$ };

    \draw[decorate,decoration={brace,amplitude=5pt}]
    (-0.49,.2) -- (-0.01,.2) node[midway,above=5pt] {no guarantee};

    \draw[decorate,decoration={brace,amplitude=5pt}]
    (-1.01,-1) -- (-1.99,-1) node[midway,below=5pt] {asymptotic validity via Theorem~\ref{thm:validity_main}};
    \draw[decorate,decoration={brace,amplitude=5pt}]
    (-0.51,-2) -- (-1.99,-2) node[midway,below=5pt] {\begin{tabular}{c} asymptotic validity via Theorem~\ref{thm:validity_linear} \\ (for linearly decomposable statistics) \end{tabular}};

    \draw (-0.5,0.1) -- (-0.5,-0.1) node[below=2pt] {$n^{-\frac{1}{2}}$};
    \draw (-1.0,0.1) -- (-1.0,-0.1) node[below=2pt] {$n^{-1}$};
    \draw (-1.5,0.1) -- (-1.5,-0.1) node[below=2pt] {$n^{-\frac{3}{2}}$};
  \end{tikzpicture}}
  \caption{An illustration of the discussion of Section~\ref{sec:Thm3_smoothness}, illustrating the range of $\zeta_n$ (capturing the smoothness of the conditional distributions) for which our theorems guarantee asymptotic validity, in the regime $\epsilon_n=O(K^{-1/2})$. On the above axis, moving to the right expands the class of distributions under consideration away from the within-bin homogeneity assumption, whereas moving to the left shrinks that class towards distributions which have exactly identical $P_{X \mid Z}$ and $P_{Y \mid Z}$ within bins. In particular, a Lipschitz smoothness assumption would guarantee $\zeta_n = \mathrm{O}(n^{-2})$.}
\label{fig:smoothness}\end{figure}
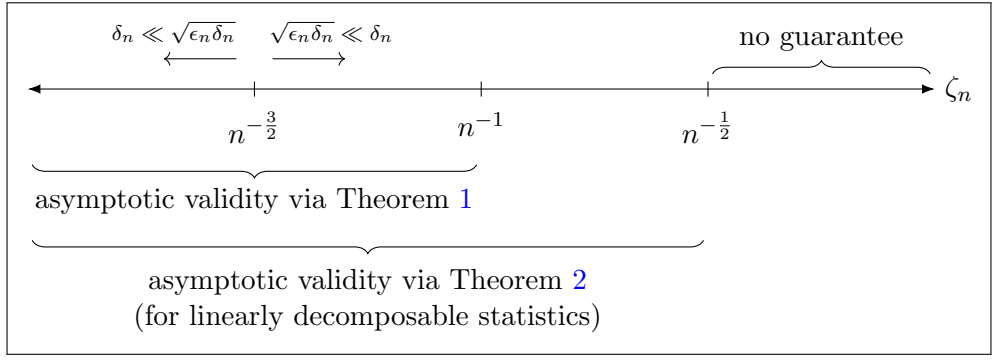

\subsection{Revisiting hardness of conditional independence}\label{sec:negative}
In the previous sections, we demonstrated how different orders of smoothness assumptions on the conditional distributions $P_{X\mid Z}$ and $P_{Y\mid Z}$ yields asymptotic validity of the adaptive-LPT test via Theorems~\ref{thm:validity_main} and~\ref{thm:validity_linear}. The core task there reduces to controlling the TV distance $\dtv\!\left(P_{X\mid Z=z},\,P_{X\mid Z=z'}\right)$ (and similarly with $X$ replaced by $Y$) for pairs $(z,z')$ that are close. While smoothness assumptions enable control of such TV terms, a natural question is how far this can be pushed: can one obtain sufficiently general control of these distances to ensure asymptotic validity under only mild regularity conditions on the joint model? 

Given the hardness result of \citet{shah2020hardness} (as reviewed in Section~\ref{sec:hardness_CI}), we know that this cannot be the case: if we could bound $\delta_n=o(1)$ uniformly over all joint distributions $P\in H_0^{\mathrm{CI}}$, then this would contradict the impossibility of distribution-free conditional independence testing. In this section, we explore this tension further, to develop a better understanding of how our results align with the hardness result.

As in earlier sections, we consider bins of size $m_k=2$, and suppose $\mathcal Z =\mathbb R$, so that our bins are constructed by taking neighboring pairs of $Z$ values. In this setting, $\delta_n$ is four times a sum of $\lfloor n/2\rfloor$ terms of the form
\[\textnormal{d}_{\textnormal{TV}}(P_{X \mid Z_i}, P_{X \mid Z_{N(i), n}}) \cdot \textnormal{d}_{\textnormal{TV}}(P_{Y \mid Z_i}, P_{Y \mid Z_{N(i), n}}),\]
where $Z_{N(i), n}$ denotes the right-nearest-neighbor to $Z_i$ within the observed set of $Z$ values (i.e., $\{Z_j\}_{j\in[n]\setminus \{i\}}$). Consequently, with large sample size, we expect $\delta_n\approx \mathbb{E}[\delta_n]$, and that 
\[\mathbb E[\delta_n] \propto  n \cdot \mathbb{E}\left[\textnormal{d}_{\textnormal{TV}}(P_{X \mid Z_1}, P_{X \mid Z_{N(1), n}}) \cdot \textnormal{d}_{\textnormal{TV}}(P_{Y \mid Z_1}, P_{Y \mid Z_{N(1), n}})\right].\]

The following result shows that the expectation on the right hand side must be vanishing for any joint distribution $P$; however, this convergence may be at an arbitrarily slow rate.

\begin{theorem} \label{thm:vanishing_tv}
  Let $P$ be any joint distribution on $\mathcal X \times\mathcal Y\times \mathbb R$, such that the conditional distributions $P_{X \mid Z = z}$ admit densities with respect to some common $\sigma$-finite measure on $\mathcal X$, and similarly for $P_{Y\mid Z=z}$. Let $(X_1,Y_1,Z_1),\dots,(X_n,Y_n,Z_n)\overset{iid}\sim P$. Then, if $Z_{N(1), n}$ denotes the right nearest neighbor to $Z_1$ among $Z_2,\dots,Z_n$,
  \[
    \lim_{n \to \infty }\mathbb E [ \textnormal{d}_{\textnormal{TV}}(P_{X \mid Z_1}, P_{X \mid Z_{N(1), n}}) \cdot \textnormal{d}_{\textnormal{TV}}(P_{Y \mid Z_1}, P_{Y \mid Z_{N(1), n}}) ] = 0.
  \]
  On the other hand, for any sequence $(a_n)_{n\ge 1}$ such that $a_n \to 0$ as $n\to\infty$, there exists a joint distribution $P$ such that
    \[
    \mathbb E [ \textnormal{d}_{\textnormal{TV}}(P_{X \mid Z_1}, P_{X \mid Z_{N(1), n}}) \cdot  \textnormal{d}_{\textnormal{TV}}(P_{Y\mid Z_1}, P_{Y \mid Z_{N(1), n}}) ]\geq a_n
  \]
  for sufficiently large $n$.
\end{theorem}

Note that the assumption of a common density for all conditionals $P_{X\mid Z}$ includes both the case where $X \mid Z$ is continuously distributed, or when $X \mid Z$ is discrete with countable support (shared over all values of $Z$). 

The first part of this result implies that we must have $\mathbb E[\delta_n] = o(n)$ for any joint distribution $P$. However, the hardness result of \citet{shah2020hardness} is not contradicted: the second part of the Theorem~\ref{thm:vanishing_tv} cautions us that any faster rate of convergence cannot hold universally. In particular, without further knowledge or conditions on the joint distribution $P$, we cannot assume that $\delta_n=o(1)$ (which is needed for the result of Theorem~\ref{thm:validity_main} to ensure asymptotic validity, as discussed in Section~\ref{sec:validity_general}), or $\delta_n=o(n^{1/2})$ (as required in Theorem~\ref{thm:validity_linear}). Consequently, no truly assumption-free test of conditional independence is possible within the LPT framework.

\section{Power}\label{sec:power}

This section develops a detailed power analysis of LPT tests. Along the way, our analysis will guide a prospective analyst in the procedure of choosing a test statistic to maximize the power of the LPT. The core contribution of this section can be split into three parts:
\begin{enumerate}
\item In Section \ref{sec:oracle_power_orc} we investigate the \textbf{oracle power} for the hypothesis test of interest, via a log-likelihood ratio test that achieves optimal power while controlling the Type~I error at pre-specified level $\alpha$.
\item Next in Sections \ref{sec:power_NNPT} and \ref{sec:linear_model}, we establish that \textbf{LPT can achieve near-oracle power} with a well chosen statistic. In particular, even with bins containing only $m=2$ points, the power of LPT is within a constant factor (in terms of effective signal strength) of the oracle.
\item Finally, in Section~\ref{sec:linear_model} we take a closer look at the role of bin size in the specific setting of a linear confounder, and give a precise \textbf{characterization of power for different choices of bin size}. In particular, as expected, we see that power is maximized by choosing bins that contain equal numbers of data points, and by choosing larger bins (although, as we have seen in the previous section, there is a tradeoff---choosing larger bins may come at the cost of losing Type I error control).
\end{enumerate}

\subsection{Oracle Power}\label{sec:oracle_power_orc}

As a benchmark, in the present section we first construct and analyze the oracle likelihood–ratio test that has access to the joint conditional model \(P_{X,Y\mid Z}\) and to the corresponding marginals \(P_{X\mid Z}\) and  \(P_{Y\mid Z}\); the resulting oracle power serves as a reference to our power throughout this section, and will serve as a point of comparison choosing bin statistics for the LPT.

In particular, for fixed $z\in \mathcal{Z}$, we denote the log-likelihood ratio function between $P_{X,Y\mid Z=z}$ and $P_{X\mid Z=z}\times P_{Y\mid Z=z}$ as
\[
\textnormal{LLR}(x,y\mid z):=\log \biggl(\frac{\mathsf{d}P_{X,Y\mid Z=z}(x,y)}{\mathsf{d}P_{X\mid Z=z}(x)\times \mathsf{d}P_{Y\mid Z=z}(y)}\biggr).
\]
For notational convenience, we further write $\textnormal{LLR}^{(i)}(x,y)$ to denote $\textnormal{LLR}(x,y\,|\,Z_i)$, the log-likelihood ratio corresponding to $Z_i$.

By the Neyman--Pearson lemma, any test that controls the Type~I error asymptotically at level $\alpha$ under $H_0^\prime: (X_i, Y_i) \sim P_{X \mid Z_i} \times P_{Y \mid Z_i}$ and achieves optimal power against $H_1^\prime: (X_i, Y_i) \sim P_{X, Y \mid Z_i}$ takes the form
\begin{equation}\label{eq:ORC_test}
    \phi_{\textnormal{ORC}}:=\One{\sum_{i=1}^n \textnormal{LLR}^{(i)}(X_i,Y_i)\geq \tau_{\alpha}},
\end{equation}
for an appropriate threshold $\tau_{\alpha}$. In the following theorem, we provide an asymptotic approximation to the power of this oracle test $\phi_{\textnormal{ORC}}$.

This approximation is primarily characterized by the following quantities. We first define the symmetrized KL divergence between $P_{X,Y\mid Z_i}$ and $P_{X\mid Z_i}\times P_{Y\mid Z_i}$, often referred to as the Jeffreys divergence:
\begin{align*}
    \overline{\mathrm{KL}}_{(i)}:=&~\mathrm{KL}(P_{X,Y\mid Z_i}\| P_{X\mid Z_i}\times P_{Y\mid Z_i}) + \mathrm{KL}(P_{X\mid Z_i}\times P_{Y\mid Z_i} \| P_{X,Y\mid Z_i})\\
  =&~\E_{P_{X,Y\mid Z_i}}[\textnormal{LLR}^{(i)}(X,Y)]-\E_{P_{X\mid Z_i}\times P_{Y\mid Z_i}}[\textnormal{LLR}^{(i)}(X,Y)],
\end{align*}
Secondly, we define the following measure
\begin{align*}
    \mathrm{V}_{\textnormal{KL},(i)}&:=\textnormal{Var}_{(X,Y)\sim P_{X,Y\mid Z_i}}(\textnormal{LLR}^{(i)}(X,Y)),
\end{align*}
for each $i\in [n]$. This is known as \emph{varentropy} in the information theory \citep{di2021stochastic}.
Finally, we define
\[
\mathrm{SNR_{ORC}}:=\frac{\sum_{i=1}^n \overline{\mathrm{KL}}_{(i)}}{\bigl(\sum_{i=1}^n\textnormal{V}_{\textnormal{KL},(i)}\bigr)^{1/2}},
\]
to be the \emph{signal-to-noise ratio} for the problem at hand. With the above notation, we are finally able to state the theorem.

\begin{theorem}[Informal]\label{thm:power_of_orc}
  Under suitable conditions on the log-likelihood ratio, the oracle likelihood ratio  test~\eqref{eq:ORC_test} satisfies
    \[
      \mathbb E[\phi_{\mathrm{ORC}} \mid \mathbf Z] = \Phi \Big(\Phi^{-1}(\alpha) +  \mathrm{SNR_{ORC}}\Big) + \mathrm{o}_P(1).
    \]
\end{theorem}
In other words, the quantity $\mathrm{SNR_{ORC}}$ characterizes the extent to which the power can be greater than $\alpha$ (i.e., better than random) for the best possible test.

The final condition about local alternatives is not necessary to get meaningful estimates of power through our methods and proofs; however, it does simplify the corresponding expression for power and we use it for expositional clarity.
See Appendix \ref{sec:power_of_orc_proof} for a precise statement (Theorem \ref{cor:orc_power}) and additional discussion. We also establish explicit finite-sample bounds (Theorem \ref{thm:general_orc_power}) in the appendix.

The oracle power derived above should be viewed primarily as a benchmark for the best achievable performance. The oracle likelihood-ratio test is constructed for testing a simple null against a simple alternative with complete knowledge of the underlying conditional distributions, whereas a valid conditional independence test must control Type~I error uniformly over a substantially larger, essentially infinite-dimensional, null class. Consequently, one cannot expect practically implementable procedures including LPT to exactly (or even asymptotically) attain the oracle power. Nevertheless, as we show in the following sections, suitably designed LPT procedures can achieve power that closely matches the oracle benchmark up to constant factors.

\subsection{Power of LPT}\label{sec:power_NNPT}
We now study the power of the adaptive-LPT test, and we establish that it has the potential to match oracle power up to a constant factor if the test statistic is chosen appropriately. More concretely, since the power of the oracle test is expressed in terms of a signal-to-noise ratio, $\mathrm{SNR_{ORC}}$, in the results above, our aim is to show that the LPT can achieve a comparable signal-to-noise ratio. We will restrict our attention to the linearly decomposable test statistics, as in Definition~\ref{def:lin_decomp}.

From this point on, for ease of the presentation, we will assume that the test statistic has been constructed to be centered with respect to permutations, satisfying
\[
  \sum_{\sigma \in \mathcal S_{m_k}} T_k((\mathbf X_k)_{\sigma}, \mathbf Y_k, \mathbf Z) = 0
\]
for each $k$; we may assume this without loss of generality since, by replacing $T_k$ with its centered version $T_k(\bX_k, \bY_k, \bZ) - \frac{1}{m_k!} \sum_{\sigma \in \mathcal S_{m_k}} T_k((\bX_k)_{\sigma}, \bY_k, \bZ)$, we can see that the p-value for the LPT (as defined in~\eqref{eq:pval_adaptive_LPT}) remains unchanged.

We now define a signal-to-noise ratio, analogous to the quantity $\mathrm{SNR_{ORC}}$ studied in Section~\ref{sec:oracle_power_orc}. In particular, for any linearly decomposable test statistic $T$, define
\[
  \mathrm{SNR_{LPT}} = \frac{\sum_{k=1}^K \mathbb E[T_k \mid \mathbf Z]}{\left( \sum_{k=1}^K \mathrm{Var}(T_k \mid \mathbf Z) \right)^{1/2}}.
\]
The next result allows us to characterize the power of the adaptive-LPT in terms of this signal-to-noise ratio, and is analogous to the result of Theorem~\ref{thm:power_of_orc} for the oracle test. Since this result involves some lengthy technical conditions, here we state an informal version of the theorem. 

\begin{theorem}[Informal]\label{thm:power_of_nnpt}
  Let $\phi_{\mathrm{LPT}}\in\{0,1\}$ denote the outcome of the LPT. Let $T$ be a linearly decomposable test statistic, with components $T_1, \dots, T_K$ that satisfy suitable conditions (analogous to those of Theorem \ref{thm:power_of_orc}). Then, the LPT with test statistic $T$ satisfies 
    \[
      \mathbb E[\phi_{\mathrm{LPT}} \mid \mathbf Z] = \Phi \left(\Phi^{-1}(\alpha) + \mathrm{SNR_{LPT}}\right) + \mathrm{o}_P(1).
    \]
\end{theorem}
See Theorem \ref{cor:nnpt_power} for a precise statement and additional discussion and Appendix \ref{sec:power_of_nnpt_proof} for details; in fact it holds that the $\mathrm{o}_P(1)$ error term above can be controlled by uniformly upper bounds on certain moments of the $T_k$. We establish explicit finite-sample bounds (Theorem \ref{thm:general_lpt_power}) in the appendix as well.

This result, which tells us that the power of the LPT is simply a function of the signal-to-noise ratio $\mathrm{SNR_{LPT}}$, corresponds to the intuitive understanding of how one should aim to choose the test statistic $T$: under the alternative, it should be as large as possible (thereby maximizing the numerator of $\mathrm{SNR_{LPT}}$) as often as possible (thereby minimizing the denominator).

\subsubsection{Local permutation tests with oracle information}\label{sec:oracle}

The result of Theorem~\ref{thm:power_of_nnpt} above expresses the power of LPT in terms of a signal-to-noise ratio $\mathrm{SNR_{LPT}}$, but this does not yet answer the question of how the LPT compares to the oracle power. To do so, we now need to explore the relationship between $\mathrm{SNR_{LPT}}$ and $\mathrm{SNR_{ORC}}$. It turns out that if one has the same knowledge as the oracle test, i.e., access to the log-likelihood ratio function $\mathrm{LLR}(x, y \mid z)$, the restriction to local permutation tests is not too costly: the next result shows that a local permutation test, using the most conservative binning strategy possible of just $m=2$ points per bin, is within a constant power loss of the oracle test when the test statistic is properly chosen. 

\begin{theorem}[Informal]\label{thm:orc_vs_nnpt_oracle}
  Under suitable conditions on the log-likelihood, there exists a choice of $T$, reliant on oracle knowledge of the log-likelihood ratio, such that the LPT with bin size $m=2$ satisfies
  \[
    \mathrm{SNR_{LPT}} \geq \frac{\mathrm{SNR_{ORC}}}{4} - \mathrm{o}_P(1).
  \]
\end{theorem}

See Theorem \ref{thm:orc_vs_nnpt_oracle_precise} for a precise statement and Appendix \ref{sec:orc_vs_nnpt_proof} for details.

We emphasize that this should not be interpreted as saying that one needs oracle knowledge to gain meaningful power, nor that any deviation from the likelihood given from the oracle would nullify the power of $\phi_{\mathrm{LPT}}$. Rather, the above is a demonstration that restricting from the space of all tests down to only LPT-style permutation tests (and in fact, LPT with only $m = 2$ points per bin) incurs only a constant factor cost in signal strengt as compared to oracle performance. In practice, using a data-driven test statistic $T$ (i.e., one that does not require oracle knowledge of the true model) can nonetheless achieve nontrivial power, as we will explore below.

\subsection{Special case: linear confounder}\label{sec:linear_model}

We now specialize our power guarantees to the setting of a linear confounder, to obtain more precise results on the power of the LPT and how it relates to the oracle power. Given Theorems \ref{thm:power_of_orc} and \ref{thm:power_of_nnpt}, the rest of the section focuses on computing the corresponding $\mathrm{SNR_{ORC}}$ and $\mathrm{SNR_{LPT}}$ quantities for this model, so that we may compare how the power of the LPT compares to the optimal oracle test.
In Section \ref{sec:linear_opt_power}, we establish the power of the optimal likelihood-ratio test, and in Section \ref{sec:linear_lpt_power}, we analyze a specific choice of the test statistic $T$ such that LPT matches the power of the optimal test up to a small constant. However, working in the special case of a linear confounder model allows us to attain additional insights unavailable in the general case: Theorem \ref{thm:power_lpt_confounder_model_general} shows that increasing the bin size in LPT grants mild power gains, and that adaptively picking bins of equal sizes can yield greater power than alternative binning strategies.

We now define the setting. Given a sample size $n\geq 1$, suppose
\begin{equation}\label{model:linear_gaussian_confounder}
  X = f_1(Z) + \beta_{1,n}\,U + \epsilon_1,
  \qquad
  Y = f_2(Z) + \beta_{2,n}\,U + \epsilon_2,
\end{equation}
where $Z,U,\epsilon_1,\epsilon_2$ are mutually independent, with $U,\epsilon_1,\epsilon_2\stackrel{\text{iid}}{\sim} N(0, 1)$.
In this model, we can quantify the confounding strength by the conditional correlation between $X$ and $Y$ given $Z$. Indeed, if $(X,Y,Z)$ follows \eqref{model:linear_gaussian_confounder}, then
\[
\mathrm{Cor}(X,Y\mid Z)
= \frac{\mathrm{Cov}(X,Y\mid Z)}
{\sqrt{\mathrm{Var}(X\mid Z)\,\mathrm{Var}(Y\mid Z)}} 
= \frac{\beta_{1,n}\,\beta_{2,n}}{\sqrt{(\beta_{1,n}^2+1)(\beta_{2,n}^2+1)}} 
= \rho_n.
\]
As one would expect, larger $\rho_n$ corresponds to stronger alternatives that are easier to detect, whereas sufficiently small $\rho_n$ yields negligible power. Because the LPT statistics considered below are one-sided, throughout this section and Appendix~\ref{sec:non_gaussianity} we assume that $\rho_n\geq 0$ eventually for simplicity, though the arguments may be adapted to two-sided versions as well. We assume that $\limsup_{n \to \infty} |\beta_{1, n}|  < \infty$ and $\limsup_{n \to \infty} |\beta_{2, n}|  < \infty$ to avoid trivial outcomes where $X$ and $Y$ are perfectly correlated, and we also assume that $\lim_{n \to \infty} \rho_n$ is defined to avoid edge cases where the limiting behavior of the test is undefined.

\paragraph{The Gaussianity assumption.} Before proceeding to our results, we briefly comment on the assumption of Gaussianity for the shared signal $U$ and the noise terms $\epsilon_1,\epsilon_2$. Can this assumption simply be replaced with some weaker moment conditions? In fact, Gaussianity plays a key role in our ability to calculate the power of the oracle, since analyzing the oracle method requires knowledge of the log-likelihood ratio function; if we only assume moment conditions then this function could take arbitrary form. On the other hand, our results for the power of LPT are straightforward to extend to a non-Gaussian setting (see Appendix~\ref{sec:non_gaussianity}).

\subsubsection{Optimal Power}\label{sec:linear_opt_power}

Our first result for this setting calculates the oracle power.

\begin{theorem}
\label{thm:gaussian_lrt_power}
Model class~\eqref{model:linear_gaussian_confounder} satisfies the conditions of Theorem \ref{thm:power_of_orc} and the power of the oracle likelihood-ratio test has signal-to-noise ratio
\[
    \mathrm{SNR_{ORC}} = \frac{\sqrt n\,|\rho_n|}{1 - \rho_n^2}.
\]
\end{theorem}
Thus, by Theorem~\ref{thm:power_of_orc},
the power can be computed as
\[
  \E\bigl[\phi_{\mathrm{ORC}}\mid \bZ\bigr] =
  \Phi\!\left(\Phi^{-1}(\alpha) + \frac{\sqrt{n}\,\rho_n}{1-\rho_n^2} \right)  + \mathrm{o}_P(1).
\]
From these results, we can see that the detection boundary is $\rho_n=\Theta(n^{-1/2})$; this is the regime in which the oracle test has nontrivial power.

\subsubsection{Power of LPT}\label{sec:linear_lpt_power}

Next, under the linear confounding model, we analyze the power of the LPT within a unified family of bin-level statistics. For each bin $B_k$ of size $m_k$, define
\begin{equation}\label{eq:bin-statistic-LPT}
T_k = \frac{1}{m_k} \sum_{\substack{i\in B_k}} \left( X_i - \overline{\mathbf X}_k \right) \left( Y_i - \overline{\mathbf Y}_k \right),
\end{equation}
where $\overline{\mathbf X}_k$ and $\overline{\mathbf Y}_k$ are the sample means of $\mathbf X_k$ and $\mathbf Y_k$. Importantly, note that this choice of the test statistic $T$ does not require assuming oracle knowledge of the exact likelihood ratio.

Define also ``error'' terms in each bin $B_k$:
\[
  S_{k}^{(1)} = \frac{1}{m_k} \sum_{i \in B_k} \left(f_1(Z_i) - \frac{1}{m_k}\sum_{j \in B_k} f_1(Z_j)\right)^2, \ 
  S_{k}^{(2)} = \frac{1}{m_k} \sum_{i \in B_k} \left(f_2(Z_i) - \frac{1}{m_k}\sum_{j \in B_k} f_2(Z_j)\right)^2.
\]
These quantities measure the deviation from the idealized case where points inside each bin have a single shared distribution (i.e. when $f_1(Z_i), f_2(Z_i)$ are constant within each bin).

\begin{theorem}\label{thm:power_lpt_confounder_model_general}
  Fix $\alpha\in(0,1/2)$ and consider the linear confounding model~\eqref{model:linear_gaussian_confounder}. For LPT with $K$ bins of sizes $m_1, \dots, m_K$ with $m_k \geq 2$ for all $k \in [K]$, and statistic~\eqref{eq:bin-statistic-LPT}, if the binning strategy satisfies $K \to \infty$ and
  \begin{equation}\label{eq:bin_size_assmp}
    \frac{\max_{k=1\dots, K} m_k}{\min_{k=1\dots, K} m_k} = \mathrm{O}_P(1)
  \end{equation}
  and the within-bin variation is sufficiently small:
  \begin{equation}\label{eq:linear_assmp_1}
    \frac{1}{K} \sum_{k=1}^K S_k^{(1)} = \mathrm{o}_P \left( 1 \right), \quad
    \frac{1}{K} \sum_{k=1}^K S_k^{(2)} = \mathrm{o}_P \left( 1 \right), \quad
    \frac{1}{K} \sum_{k=1}^K S_k^{(1)} S_k^{(2)} = \mathrm{o}_P \left( 1 \right), 
  \end{equation}
  and
  \begin{equation}\label{eq:linear_assmp_2}
    \left( \frac{1}{K} \sum_{k=1}^K S_k^{(1)} \right)^{1/2} \left( \frac{1}{K} \sum_{k = 1}^K S_k^{(2)} \right)^{1/2} = \mathrm{o}_P \left( \rho_n \right)
  \end{equation}
  then the conditional power of the LPT satisfies
  \[
    \mathbb E[\phi_{\mathrm{LPT}}\mid\mathbf Z]
    =\Phi\!\left(\Phi^{-1}(\alpha)+\mathrm{SNR_{LPT}}\right)+\mathrm{o}_P(1),
  \]
  where
\[
    \mathrm{SNR_{LPT}} = \frac{\sum_{k=1}^K \frac{m_k - 1}{m_k}}{\sqrt{\sum_{k=1}^K \frac{m_k-1}{m_k^2}}} \cdot \frac{\rho_n}{\sqrt{1 + \rho_n^2}} \cdot \left( 1 +  \mathrm{o}_P(1) \right).
\]
\end{theorem}

\begin{corollary}\label{cor:equal_bin_cor}
  For bins of equal size $m = m_1 = \dots = m_K$, if $K \to \infty$ and \eqref{eq:linear_assmp_1} and \eqref{eq:linear_assmp_2} hold, then LPT with statistic~\eqref{eq:bin-statistic-LPT} and model~\eqref{model:linear_gaussian_confounder} satisfies
  \[
    \mathrm{SNR_{LPT}} = \sqrt{\frac{m-1}{m}} \cdot \frac{\sqrt{n} \rho_n}{\sqrt{1 + \rho_n^2}} + \mathrm{o}_P(1).
  \]
\end{corollary}

Combining these results with Theorem~\ref{thm:power_of_nnpt} allows us to compute the power of the LPT: for instance, with equal bin sizes $m_1 = \dots = m_K=m$,  we have
\[
    \E[\phi_{\textnormal{LPT}} \mid \mathbf Z] = \Phi\!\left(
    \Phi^{-1}(\alpha) + 
       \sqrt{\frac{m-1}{m}} \cdot \frac{\sqrt{n} \rho_n}{\sqrt{1 + \rho_n^2}}
      \right) + \mathrm{o}_P(1).
  \]

\paragraph{Examining the assumptions.} To better understand these results, we briefly discuss the assumptions. First, note that requiring $m_k \geq 2$ for all $k$ is no actual restriction: if $m_k = 1$ that bin is invariant under permutation and does not change the outcome of $\phi_{\mathrm{LPT}}$, so if the binning strategy contains bins of size $1$, the above can still be applied after ``dropping'' that bin (and so $K$ should be interpreted as the number of \textit{nontrivial} bins). However, \eqref{eq:bin_size_assmp} is crucial: as our method relies on an eventual application of the central limit theorem to $\sum_{k=1}^K T_k$, it becomes necessary to ensure that the bins are of comparable size so that one term does not dominate the sum. 

The smoothness assumptions \eqref{eq:linear_assmp_1} and \eqref{eq:linear_assmp_2} on the error terms $S^{(1)}_k,S^{(2)}_k$ are quite lax. In essence, these conditions require that the values of $f_1(Z), f_2(Z)$ are not too disparate within bins. In particular, the first condition \eqref{eq:linear_assmp_1} is quite weak and can easily be derived from mild assumptions. For example, if $\mathcal{Z} = \mathbb{R}^d$ and $f_1,f_2$ are bounded functions, then measure theoretic arguments (e.g., Lusin's theorem) imply that
\[\textnormal{If $\max_k \max_{i,i'\in B_k}\|Z_i-Z_{i'}\| = \mathrm{o}_P(1)$ then~\eqref{eq:linear_assmp_1} holds.} \]
Condition \eqref{eq:linear_assmp_2} is more stringent, as it depends on $\rho_n$, which is potentially quite small in local alternatives. However, there are two features of this requirement that are quite nice. First, as it is a product bound, the binning needs to only control either $S_k^{(1)}$ or $S_k^{(2)}$, but not both for this condition to hold. Second, the actual rate at which $S_k^{(1)}$ and $S_k^{(2)}$ must vanish is quite slow. For example, take an analogous setting to Section \ref{sec:validity_smooth} and assume $f_1, f_2$ to be Lipschitz; then condition \eqref{eq:linear_assmp_2} is implied if we assume that
\[
  \max_k \max_{i, j \in B_k} d(Z_i, Z_j) = \mathrm o(n^{-1/4}),
\]
i.e., the binning strategy produces bins of diameter $\ll n^{-1/4}$.

\subsubsection{Comparing the LPT to the oracle}
To summarize our results for the linear confounder setting, comparing the LPT (with the specific choice of $T$ above) to the oracle, let us first consider the setting of a constant bin size $m_k=m$ for ease of comparison. We have seen that the signal-to-noise ratio for the oracle test, and for the LPT, are asymptotically given by
\[\mathrm{SNR_{ORC}} \approx \dfrac{\sqrt n \rho_n}{1 - \rho_n^2}, \quad
\mathrm{SNR_{LPT}} \approx \sqrt{\frac{m - 1}{m}} \cdot \dfrac{\sqrt n \rho_n}{\sqrt{1 + \rho_n^2}}.\]
This implies that the detection threshold for LPT matches that of the oracle test: in the local alternative regime $\rho_n \to 0$, the signal-to-noise ratio is
  \[
    \mathrm{SNR_{LPT}} \approx  \sqrt{n}\,\sqrt{\frac{m-1}{m}}\sqrt{\frac{\rho_n^2}{\rho_n^2+1}} \asymp \sqrt{n}\,\sqrt{\frac{m-1}{m}}\,\rho_n,
  \]
  so the LPT exhibits the same detection boundary $\rho_n=\Theta(n^{-1/2})$ as the oracle test, which has
  \[
    \mathrm{SNR_{ORC}} = \sqrt \frac{n \rho_n^2}{1 - \rho_n^2} \asymp \sqrt{n} \rho_n
  \]
  
Moreover, these calculations also show that the potential gains due to larger bin size are limited. Again assuming $\rho_n\to0$ as $n\to\infty$, we can compare the LPT to the oracle as
\[\lim_{n \to \infty} \dfrac{\mathrm{SNR_{LPT}}}{\mathrm{SNR_{ORC}}} \to \sqrt{\dfrac{m - 1}{m}}.\]
In particular, the power of LPT increases with $m$ as we might expect, but the gain is bounded by a constant factor. At the extreme, for the constant bin size $m=2$, the LPT power loss is asymptotically expressed by the constant factor $\frac{1}{\sqrt{2}} \approx 0.707$, as compared to the signal-to-noise ratio of the oracle, meaning that there is an inherent limit to what can be gained by increasing the bin size. Since larger bin size can lead to loss of validity (as explored in Section~\ref{sec:validity}), this suggests that choosing large $m$ is risky and offers limited gain.

Finally, when implementing LPT, we are free to choose bins $B_k$ of varying sizes $m_k$---but in fact Theorem~\ref{thm:power_lpt_confounder_model_general} implies that power is maximized whenever bins are all of equal size. To see this, fix any number of bins $K$, and assume $m=n/K$ is an integer. Then we have
  \[
    \mathrm{SNR_{LPT}} = \frac{\sum_{k=1}^K \frac{m_k - 1}{m_k}}{\sqrt{\sum_{k=1}^K \frac{m_k-1}{m_k^2}}} \cdot \frac{ \rho_n}{\sqrt{1 + \rho_n^2}} \leq \sqrt{\sum_{k=1}^K m_k - 1} \cdot \frac{ \rho_n}{\sqrt{1 + \rho_n^2}} = \sqrt{n - K} \cdot \frac{\rho_n}{\sqrt{1 + \rho_n^2}},
  \]
  where the inequality step holds since $(\sum_{k=1}^K \frac{m_k - 1}{m_k})^2\leq {\sum_{k=1}^K (m_k - 1)} \cdot {\sum_{k=1}^K \frac{m_k-1}{m_k^2}} $, by Cauchy--Schwarz. However, if $m_k=m$ for all $k$, then this inequality becomes an equality---and therefore, $ \mathrm{SNR_{LPT}}$ is maximized by choosing bins of equal size.

\section{Numerical Simulations}\label{sec:sims}

In this section, we present a set of simulations to demonstrate the validity and power of our method.\footnote{Code to reproduce the experiments can be found at \url{https://github.com/davlichen/cond-independence-testing}.}

\subsection{Experiment 1: validity of adaptive-LPT}\label{sec:experiment1}

In the first experiment, we evaluate the Type~I error control of the adaptive-LPT test under a suitably designed numerical setting, and therefore verify the validity results established in Section \ref{sec:validity}. Our primary focus is to analyze performance of adaptive-LPT under the null model, and particularly study the role of the conditional distributions, $P_{X\mid Z}$ and $P_{Y\mid Z}$.

We take $Z \sim \mathrm{Unif}([0, \theta])$ and conditional on $Z$, $X, Y \iidsim \mathrm{Unif}([Z, Z+1])$ so that the null hypothesis $H_0^{\mathrm{CI}}$ is true. To implement adaptive-LPT, we form bins of size $m=2$, by pairing $Z_{(2i - 1)}$ with $Z_{(2i)}$ for each $i \in [n / 2]$. As before, here $Z_{(i)}$ denotes the $i$-th order statistic of observed $Z$ samples.
For each bin $B_k=\{i,j\}$, by construction, the conditional distributions satisfy
\begin{align*}\dtv(P_{X\mid Z_i},P_{X\mid Z_j}) &= \dtv(P_{Y\mid Z_i},P_{Y\mid Z_j}) \\&= \dtv\big(\mathrm{Unif}([Z_i, Z_i+1]),\mathrm{Unif}([Z_j, Z_j+1])\big) = \min\{|Z_i - Z_j|, 1 \}.\end{align*}
Since we pair adjacent $Z$ values, we expect to have $|Z_i - Z_j|\propto \theta/n$ for a typical bin $B_k = \{i,j\}$. Recalling the notation of Theorem~\ref{thm:validity_main}, then, we expect $\delta_n=\mathrm{O}_P(\theta^2/n)$. Theorem~\ref{thm:validity_main} immediately implies that we expect Type I error to be bounded as $\alpha + \mathrm{O}_P(\theta^2/n)$, and consequently, this implies asymptotic validity for $\theta = \mathrm{o}_P(n^{1/2})$.

To run the LPT, we construct a simple covariance based statistic:
\[
  T_k(\mathbf X_k, \mathbf Y_k) = (X_i - X_j)(Y_i - Y_j) \text{ for } B_k = \{i, j\}.
\]
(Recalling Section~\ref{sec:examples_Sec2_covariance}, this is the same except that we do not take an absolute value.)
Since this is a linearly decomposable test statistic, the stronger Type I error guarantee of Theorem~\ref{thm:validity_linear} holds as well: we have $\epsilon_n=\mathrm{O}_P(n^{-0.5})$, and so Type I error is bounded as $\alpha + \mathrm{O}_P(\theta/n^{0.75})$, which implies asymptotic validity for $\theta = \mathrm{o}_P(n^{0.75})$ (a stronger result than the scaling implied by Theorem~\ref{thm:validity_main}).

We vary sample size $n\in\{50, 100, 200, 400, 800, 1600, 3200, 6400\}$, and then Figure \ref{fig:ex1validity} displays the Type~I error curves against $n$, under four regimes, characterized by $\theta = n^a$ for $a\in\{0,0.5,0.6,0.75\}$; as larger values of $a$ correspond to larger intra-bin total variation distances (scaling with $n$), they should correspond to worse excess Type I error. As predicted by the theory, we see that asymptotic validity seems to be achieved for both the $\theta = n^{0.5}$ and $\theta = n^{0.6}$ settings, but not $\theta = n^{0.75}$.

\begin{figure}
  \centering
  \includegraphics[width=0.60\textwidth]{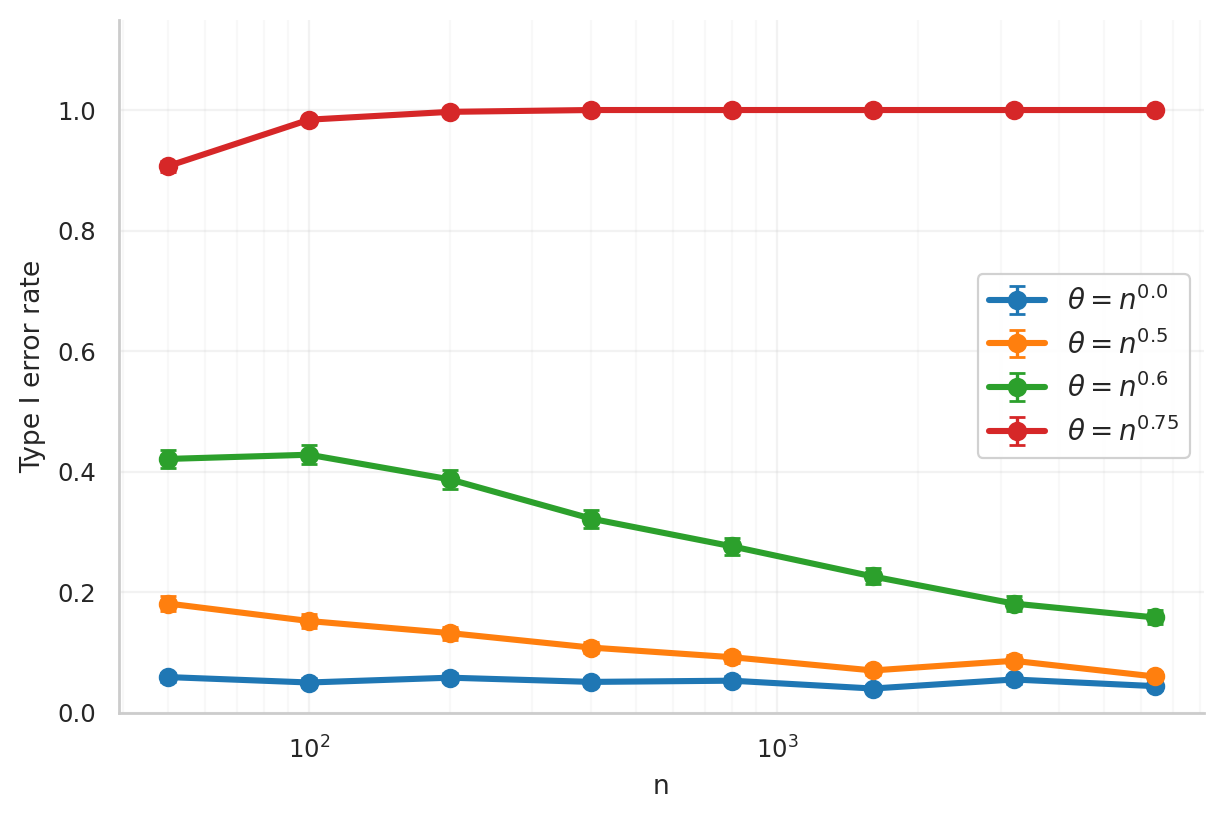}
  \caption{The plot displays Type I error rate of adaptive-LPT in Experiment 1, averaged over 1000 independent trials; standard error bars are shown, but are not easily visible as they are smaller than the points. See Section~\ref{sec:experiment1} for more details.}
  \label{fig:ex1validity}
\end{figure}

\subsection{Experiment 2: power in the linear confounder model}\label{sec:experiment2}

In the second experiment, we demonstrate the performance of adaptive-LPT in the linear confounder model~\eqref{model:linear_gaussian_confounder}, and validate the power results of Section \ref{sec:power}.
In particular, we draw $Z \sim \mathrm{Unif}([0, 1]), U \sim N(0, 1)$, and $\epsilon_1, \epsilon_2 \sim N(0, 1)$ and set
\[
  X = f(Z) + \beta U + \epsilon_1, \ Y = f(Z) + \beta U + \epsilon_2,
\]
with $f(z) = \frac{32}{3}z^{3}-16z^{2}+\frac{19}{3}z $. Note that under this model, $\beta$ alone characterizes the dependence between $X$ and $Y$, conditional on $Z$. When $\beta=0$, the null hypothesis $H_0^{\mathrm{CI}}$ holds. As $\beta$ increases, the model becomes progressively more different from the distributions that satisfy $H_0^{\mathrm{CI}}$.

For a bin $B_k$ with size $m_k$, we run LPT with the simple covariance based U-statistic, given by
\[
T_k = \frac{1}{m_k} \sum_{\substack{i\in B_k}} \left( X_i - \overline{\mathbf X}_k \right) \left( Y_i - \overline{\mathbf Y}_k \right),
\]
and form $T = \sum_{k=1}^K T_k$. This is the statistic discussed in Section~\ref{sec:examples_Sec2_covariance} except that we do not take an absolute value.
We study the power of adaptive-LPT under two binning strategies:
\begin{itemize}
    \item The first strategy is adaptive: we group points into bins of equal size by taking $B_1$ to be the indices corresponding to the first $m$ order statistics of $\mathbf Z$, $B_2$ to be the indices corresponding to the next $m$ order statistics, and so on. 
    \item The second strategy groups points by partitioning the unit interval into continuous intervals of width $w = m/n$ and grouping points which fall into the same interval together, i.e., the first bin is given by $B_1 = \{i : Z_i\in[0,\frac{m}{n}]\}$, and so on. 
\end{itemize}

Note that the expected number of points in each interval is exactly $m$ for both strategies, making them comparable. We vary $n \in \{200, 400, 800, 1600, 3200, 6400\}$, and evaluate the performance of adaptive-LPT. We validate our theoretical results from Section~\ref{sec:linear_model}, and illustrate the effect of design choices, specifically $m$, in governing power of the LPT. However, it is also important to note that, $m$ which controls the refinement of the bins, also affects the Type I error control (in the setting $\beta=0$). Therefore, in order to have a complete study of the effect of $m$ on power of adaptive-LPT, we should look at both Type~I error and power.

Firstly, the left side of Figure \ref{fig:ex2} displays the Type I error, in the setting $\beta=0$. We see that Type I error inflation is larger for binning schemes with larger bins, as expected. Moreover, we observe that for small sample sizes, the largest bin choices seem completely unreliable, having Type I error close to $1$. We remark that even though large bins (such as those of size $n^{3/5}$) satisfy the conditions for asymptotic validity in Theorem \ref{thm:validity_linear}, the convergence can be quite slow even for reasonably well-behaved distributions. The aforementioned model choice is one instance of that behavior.
  
Next, the right side of Figure \ref{fig:ex2} displays the power of the procedure with $\beta = 0.2$ 
as we vary both the binning strategy and bin size. The choice $\beta=0.2$ enables a meaningful comparison since with small sample sizes, even the oracle test fails to detect the departure from the null, while as sample size increases, the power eventually increases. Therefore, the power-curve for oracle test gives a non-trivial baseline. Next, we observe that data-adaptive bins show higher power than their fixed-partition counterparts, at each value of $m$, and for data-adaptive bins, increasing the bin size $m$ leads to higher power (as expected), but has diminishing returns: we can see that even a small constant bin size $m=6$ has power nearly matching the oracle. Note that larger bins, such as $m=n^{0.5}$, appear to show higher power than the oracle, but this is not a contradiction since the Type I error is severely inflated for such choices.

\begin{figure}
  \centering
  \includegraphics[width=0.45\textwidth]{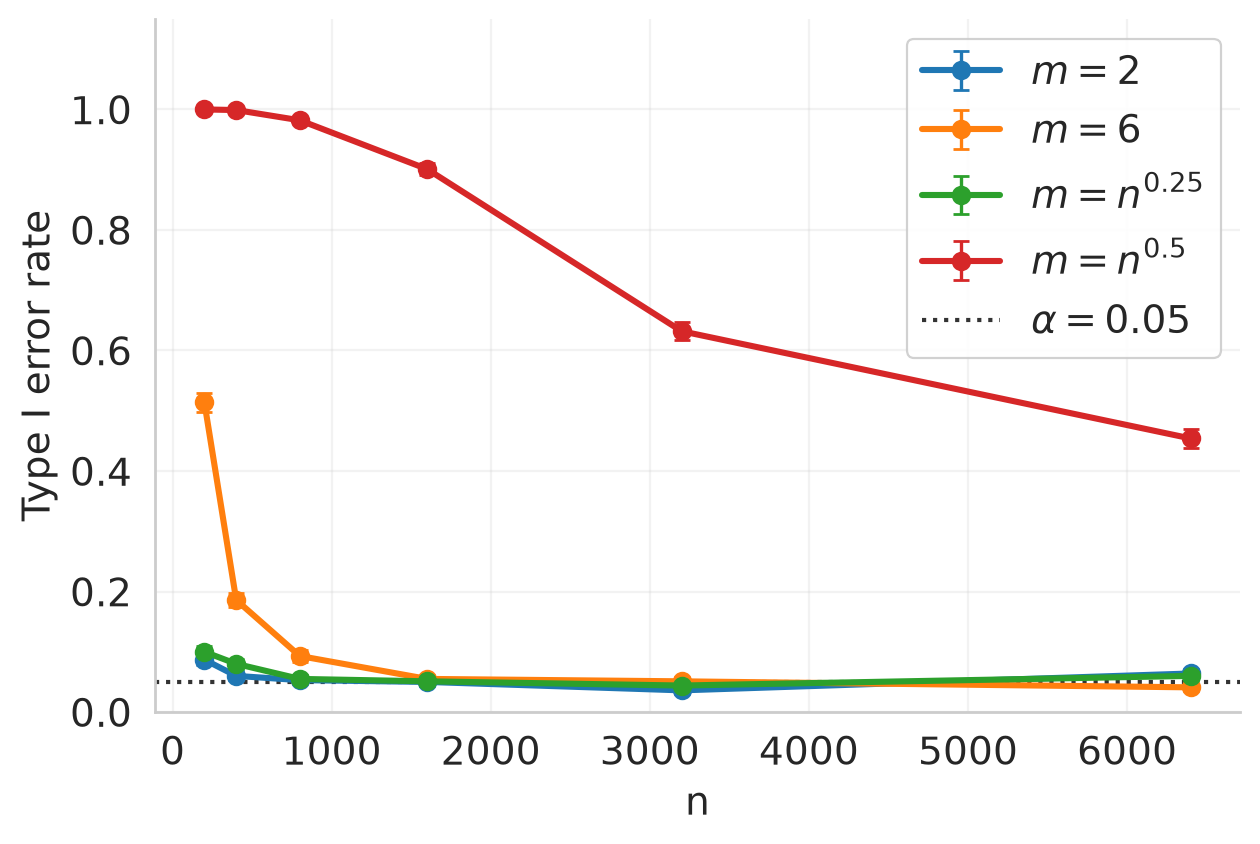}
  \includegraphics[width=0.45\textwidth]{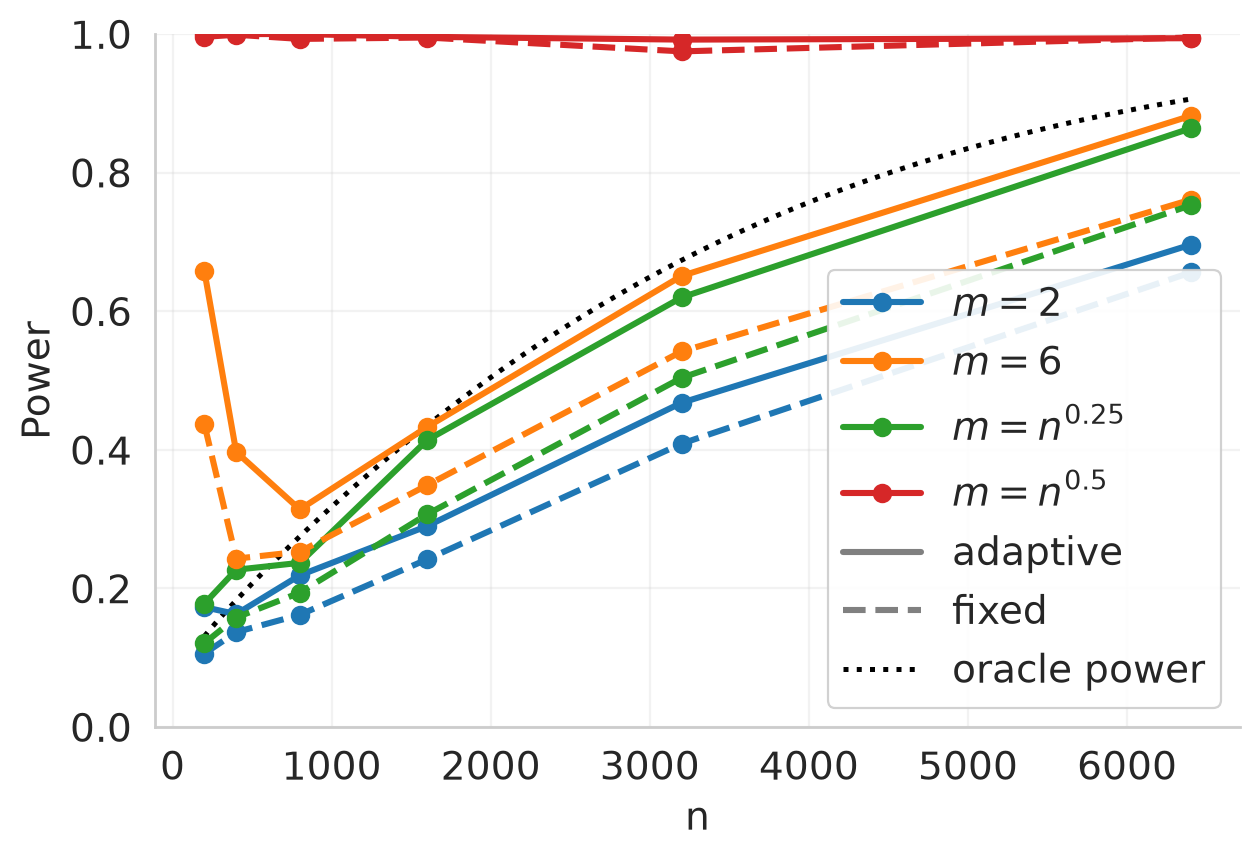}
  \caption{Type I error (left) and power (right) in Experiment 2 (Section~\ref{sec:experiment2}). Solid lines use adaptive equal-size bins of size $m$; dashed lines use fixed intervals of width $w = m/n$, so the expected bin size is $m$ for both (matching colors across panels share an expected bin size). As $n$ grows the test approaches validity, but very large bins (e.g.\ $m = n^{0.5}$) converge slowly. The black dotted line is the oracle test (Theorem~\ref{thm:power_lpt_confounder_model_general}). These are averages over 1000 independent trials; standard error bars are shown, but are not easily visible as they are smaller than the points.}
  \label{fig:ex2}
\end{figure}

\subsection{Experiment 3: adaptive binning}\label{sec:experiment3}

In the final experiment, we consider a data-generating model with strong structural properties to highlight the advantages of data-adaptive binning under the LPT framework. This supports the practicality of our proposal beyond the setting considered in Experiment $2$ and Section \ref{sec:linear_lpt_power}.

In this setting, we set $A = \{ (x_1, x_2) \mid \sqrt{x_1^2 + x_2^2} \in [1, 1.5] \}\subseteq\mathbb{R}^2$ to be the annulus with inner radius $1$ and outer radius $1.5$, and draw $Z \sim \mathrm{Unif}(A)$, that is $Z$ is sampled uniformly from this region. Then, writing $Z = (Z_1, Z_2)$, we draw $X, Y \sim N(|Z|, 1)$ with $\mathrm{Cov}(X, Y\mid Z) = 0.5 \sin(5 \arctan(Z_1 / Z_2))$.

This setup is picked specifically such that the distribution of $(X, Y, Z)$ has a strong structure which does not fit the usual Cartesian grid, meaning that a fixed-partition binning strategy (such as the standard strategy of an axis-aligned grid in $\mathbb R^2$, which seems most natural) may not suit the structure of the data.

In this case, we use a kernel-based statistic from the universal unconditional permutation test of \cite{gretton2007kernel}. Fix the Gaussian kernel $k(u,v)=\frac{\exp(-(u-v)^2)}{2}$. For each bin $B_k = \{i_1, \dots, i_m\}$, define matrices $K^{(k)}_X$ and $K^{(k)}_Y$ with entries
\[
  (K^{(k)}_X)_{ab} = k\!\left(X_{i_a}, X_{i_b}\right) = \exp\!\left(-\frac{(X_{i_a}-X_{i_b})^2}{2}\right),
\]
and
\[
  (K^{(k)}_Y)_{ab} = k\!\left(Y_{i_a}, Y_{i_b}\right) = \exp\!\left(-\frac{(Y_{i_a}-Y_{i_b})^2}{2}\right).
\]

Let the (double-)centering matrix be
\[
H_k := I_{m_k} - \frac{1}{m_k}\mathbf{1}\mathbf{1}^\top,
\]
and define the doubly-centered matrices
\[
\widetilde K^{(k)}_X := H_k K^{(k)}_X H_k, 
\qquad
\widetilde K^{(k)}_Y := H_k K^{(k)}_Y H_k.
\]

Then the per-bin statistic is
\[
T_k := \frac{1}{m_k^2}\,\mathrm{tr}\!\left(\widetilde K^{(k)}_X \,\widetilde K^{(k)}_Y\right).
\]

As in Experiment $2$, we test two different binning strategies. The first is a data-adaptive strategy, which via simulated annealing optimizes the total sum of bin diameters, i.e.\ for a fixed number of bins $K$ of fixed sizes, we (approximately) minimize the objective
\[
  \sum_{k=1}^K \max_{i, j \in B_k} |Z_i - Z_j|
\]
over all possible bins of those sizes. The second strategy uses a fixed partition, by partitioning $\mathbb R^2$ into squares of various side lengths and bins together points which have $Z$ values in the same square. See Figure~\ref{fig:ex3_bins} for an illustration. To make the two strategies comparable, 
a data-adaptive strategy with an equal number of data points $m$ in each bin is compared with an implementation of the fixed partition by choosing a grid size in $\mathbb R^2$ that results in the same average number of data points $m$ per bin. We run simulations for $n \in \{ 100, 250, 500, 750, 1000, 1250 \}$.

The results are shown in Figure~\ref{fig:ex3}. We see that both binning procedures attain approximate validity under the null. However, the data-adaptive binning strategy enjoys greater power, as compared to the fixed-partition strategy constructed by partitioning $\mathbb R^2$ into a fixed grid with the same average number of data points $m$ per bin. Moreover, note that increasing the size of the bins has diminishing returns for power, (e.g. the power gain from $m=4$ to $m=8$ is much larger than that of $m=8$ to $m=12$), and these effects are most pronounced in moderate sample sizes. 

\begin{figure}[ht]
  \centering
  \includegraphics[width=0.45\textwidth]{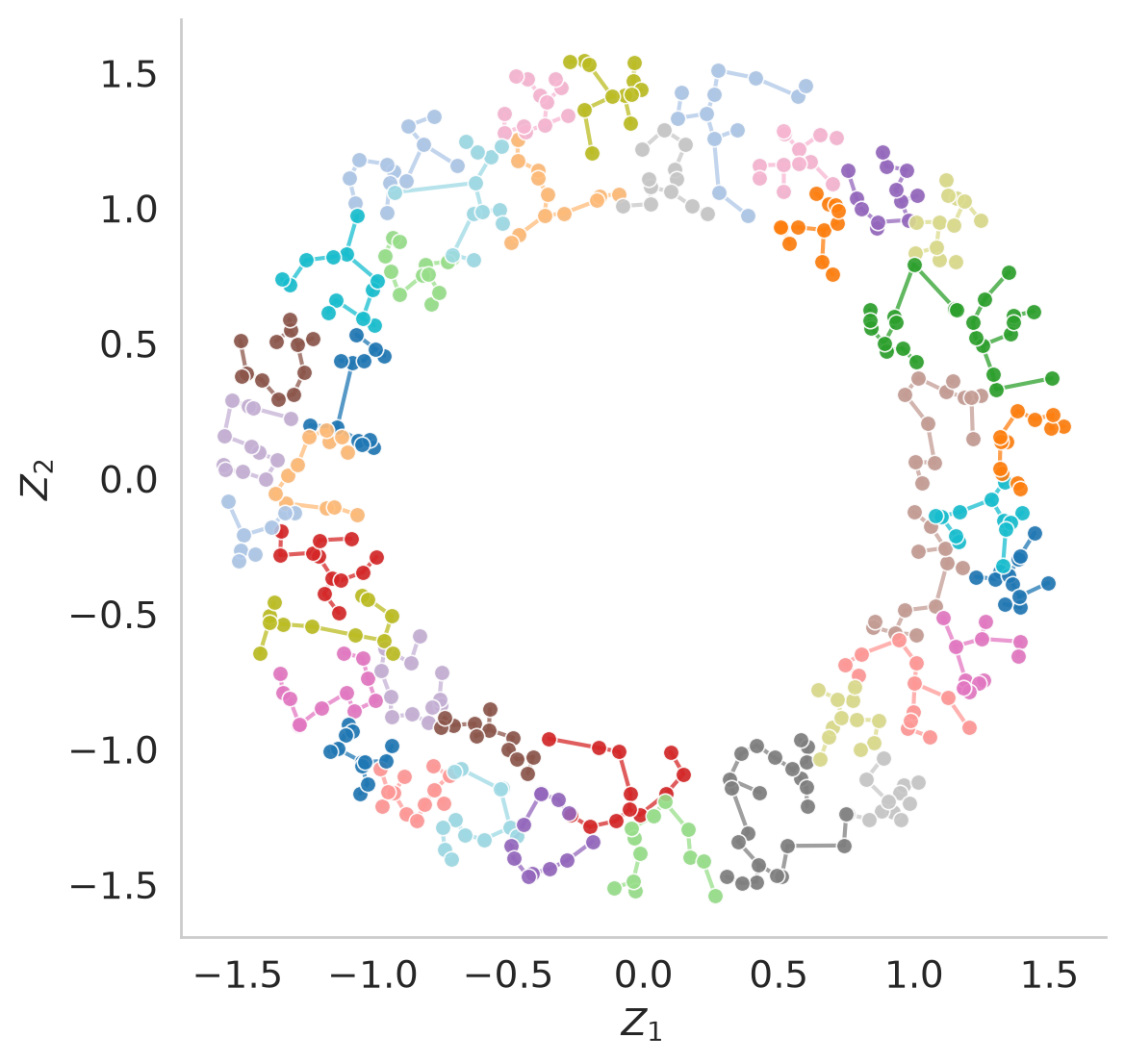}
  \includegraphics[width=0.45\textwidth]{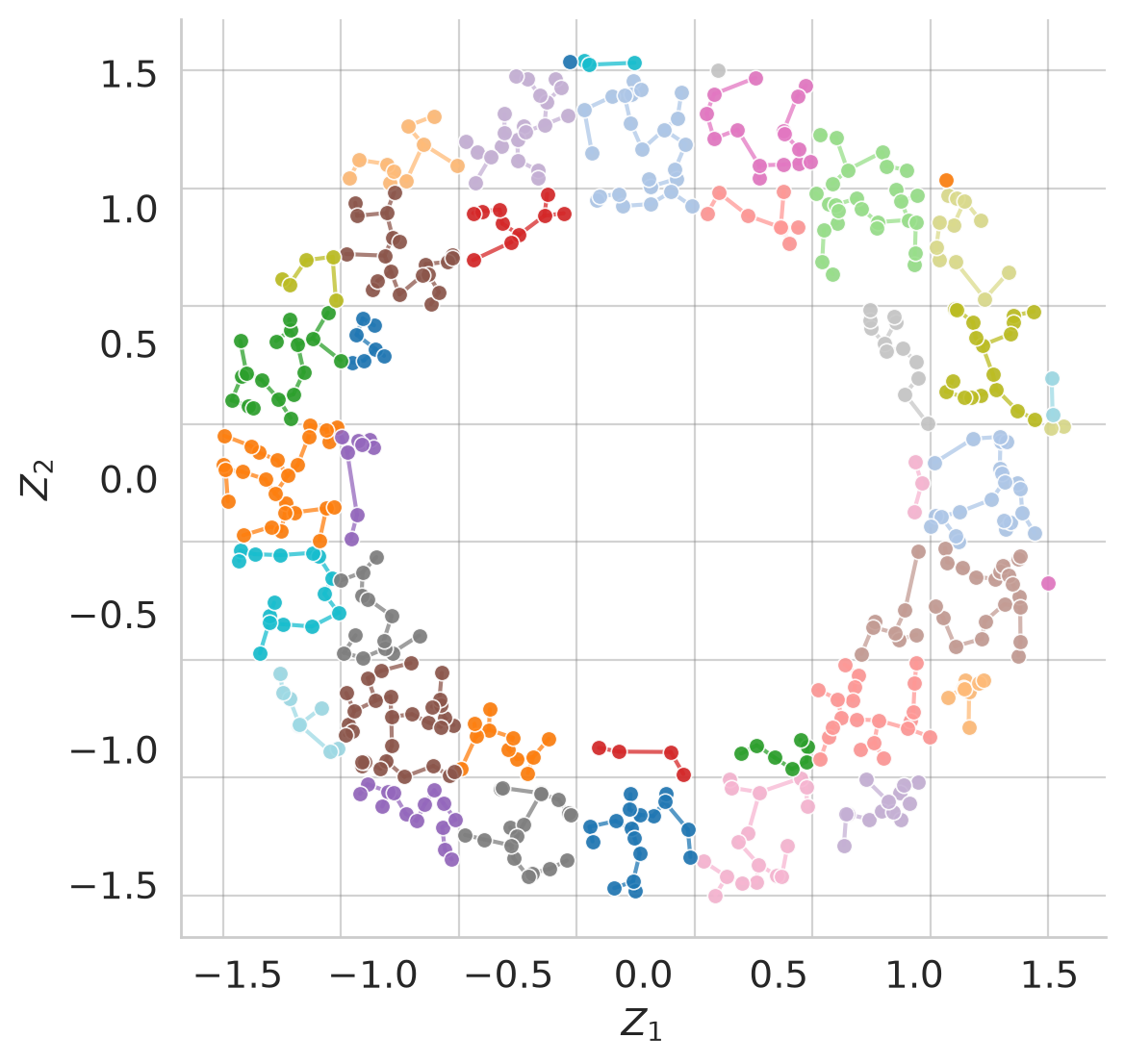} \\
  \caption{Demonstration of different binning procedures in Experiment 3 (see Section~\ref{sec:experiment3}), in a single trial (of $n=500$ points). Bins share a common color. Data-adaptive binning (left) is the result of simulated annealing to minimize the diameter of each bin of size $12$. Fixed-partition binning (right) is the result of dividing $\mathbb R^2$ into bins (shown as a grid on the plot), where the grid spacing is approximately $0.406$; this value is chosen such that the expected number of points in each nonempty bin is $12$ points and that the two methods are comparable.}
  \label{fig:ex3_bins}
\end{figure}

\begin{figure}[!htb]
  \centering
  \includegraphics[width=0.45\textwidth]{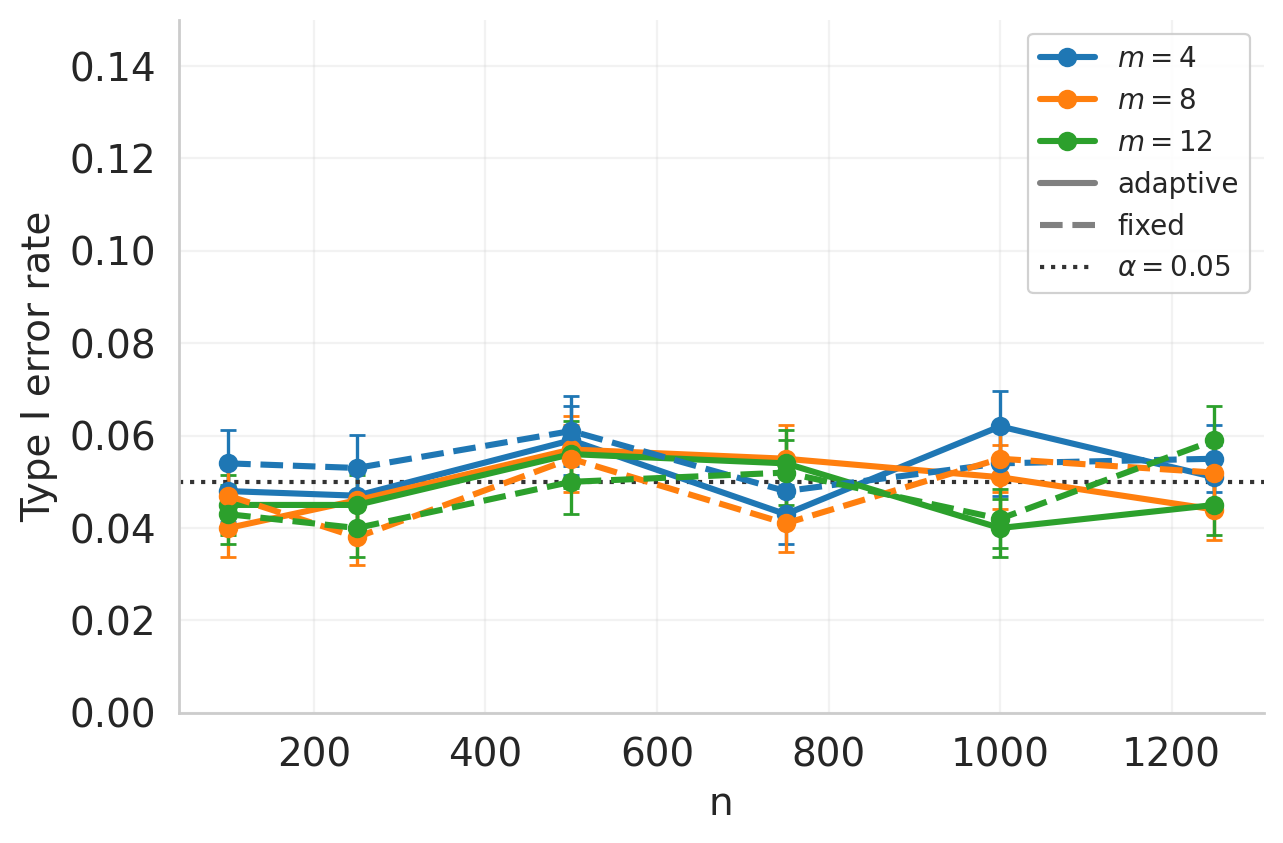}
  \includegraphics[width=0.45\textwidth]{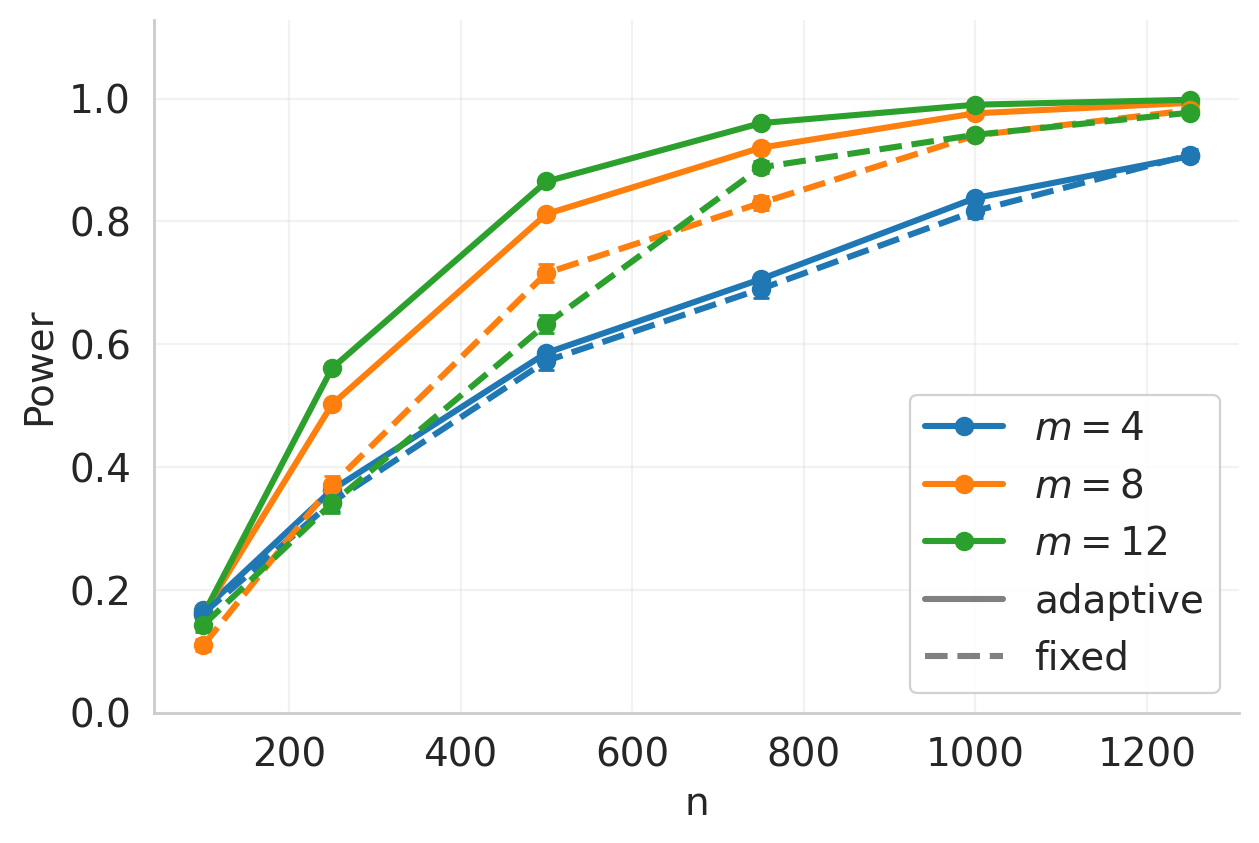} \\
  \caption{Results for Experiment 3 (see Section~\ref{sec:experiment3}), showing Type I error (left) and power (right), averaged over 1000 independent trials; standard error bars are shown, but are not easily visible as they are smaller than the points. Solid lines represent a data-adaptive binning strategy, whereas dashed lines represent a fixed-partition binning strategy.}
  \label{fig:ex3}
\end{figure}

\section{Conclusion}

In this paper, we propose data-adaptive binning strategies for local permutation tests, to construct flexible and powerful tests for conditional independence. We began by investigating such tests' validity properties and provided upper bounds on the excess Type I error based on the underlying smoothness properties of the conditional distributions, generalizing the results of \cite{kim2022local} that are limited to prefixed bins. In that process, we derived a stronger Type I guarantee under linear decomposability of the test statistic. We further derived asymptotic approximations to the power of the LPT tests and demonstrated that adaptively binned local permutation tests may closely approximate the power of oracle likelihood-ratio tests. Finally, specializing to a linear confounding model, our power guarantees closely characterize the dependence of the chosen binning strategy on the resulting power of the LPT, enabling us to further establish the optimality of equally-sized bins and a quantitative characterization of the diminishing returns of bin size to power.

We close by mentioning some potential directions for future work. It would be interesting to examine the effect of different test statistics when used for local permutation tests. Various unconditional test statistics (e.g.\ kernel and distance measures such as \cite{gretton2007kernel}, used in Section \ref{sec:sims} above) have received detailed analysis, and it would be interesting to see if their efficiency guarantees from the unconditional testing carry over to the conditional independence setting. For example, Theorem \ref{thm:power_lpt_confounder_model_general} showed that a very simple non-oracle statistic with a very simple binning strategy can achieve oracle performance in a specific model class.  the extent to which this generalizes---and to which distributional classes, with which test statistics---is an interesting open direction.

\subsection*{Acknowledgements}
D.C. and R.F.B. were partially supported by the National Science Foundation via grant DMS-2023109. R.F.B. was partially supported by 
 the Office of Naval Research via grant N00014-24-1-2544.
 
\bibliographystyle{apalike}

\bibliography{ref}

\clearpage

\appendix

\section{Additional results} \label{sec:additional}

In this section, we develop several additional results that extend the theoretical findings of this work. In Section~\ref{app:other_metrics}, we extend the validity results of Section~\ref{sec:validity}, which relied on total variation distance between conditional distributions, to allow alternative measures of distance. In Section~\ref{sec:non_gaussianity}, we extend the power analysis for the linear confounder model (Section~\ref{sec:linear_model}) to remove the Gaussianity assumption. Finally, in Section~\ref{sec:binary_power}, we develop a power analysis for the betting-based test statistic constructed in Section~\ref{ex:betting}.

\subsection{Validity bounds for other probability metrics} \label{app:other_metrics}

\subsubsection{Definitions and identities}

First, let us give some definitions. Let $P$ and $Q$ be two probability measures with densities $p$ and $q$ according to a shared dominating measure $\mu$. We thereby define the following probability metrics.

\begin{definition}
  For any $\gamma \geq 1$, the generalized $\gamma$-Hellinger distance is given by
  \[
    \mathrm d_{\mathrm H,\gamma}(P, Q) = \left( \frac{1}{2} \int |p^{1/\gamma} - q^{1/\gamma}|^\gamma \mathsf{d}\mu \right)^{1/\gamma}.
  \]
\end{definition}

$\gamma = 1$ corresponds to the total variation distance $\mathrm{d_{TV}}$, and $\gamma = 2$ to the usual Hellinger distance $\mathrm d_{\mathrm H}$.

\begin{definition}
  For any $\gamma > 0$, $\gamma \neq 1$, the R\'enyi divergence of order $\gamma$ is given by
  \[
    \mathrm{d_{R, \gamma}}(P \| Q) = \dfrac{1}{\gamma - 1} \log \left ( \int \left( \dfrac{p}{q} \right)^\gamma q\mathsf{d}\mu  \right),
  \]
  while for $\gamma=1$, the R\'enyi divergence of order $1$ is simply the KL divergence.
\end{definition}

\subsubsection{Validity results}
We first state generalized versions of
Theorem \ref{thm:validity_main}, that are based upon the Hellinger distance.

\begin{theorem}\label{thm:validity_hellinger}
  For any $\alpha \in (0, 1)$, Type I error is bounded as
  \[
    \mathbb P \left( p\leq \alpha \mid \mathbf Z \right) \leq \alpha + \left( \sum_{k=1}^K 8 m_k \left( \max_{i, j \in B_k} \mathrm d_{\mathrm H}^2(P_{X \mid Z_i}, P_{X \mid Z_j}) \right) \left( \max_{i, j \in B_k} \mathrm d_{\mathrm H}^2(P_{Y \mid Z_i}, P_{Y \mid Z_j}) \right) \right)^{1/2}.
  \]
\end{theorem}

Moreover, applying Lemmas \ref{lem:gen_hellinger} and \ref{lem:renyi} of Appendix \ref{sec:app_tech} gives the following corollaries.

\begin{corollary}\label{cor:hellinger}
  For any $\alpha \in (0, 1)$, and $1 \leq \gamma \leq 2$, Type I error is bounded as
  \[
    \mathbb P \left( p\leq \alpha \mid \mathbf Z \right) \leq \alpha +
    \left( \sum_{k=1}^K C_\gamma m_k \left( \max_{i, j \in B_k} \mathrm d_{\mathrm H,\gamma}^{\gamma}(P_{X \mid Z_i}, P_{X \mid Z_j}) \right) \left( \max_{i, j \in B_k} \mathrm d_{\mathrm H,\gamma}^{\gamma}(P_{Y \mid Z_i}, P_{Y \mid Z_j}) \right) \right)^{1/2},
  \]
  and for $\gamma > 2$, the bound is
  \[
    \mathbb P \left( p\leq \alpha \mid \mathbf Z \right) \leq \alpha +
    \left( \sum_{k=1}^K C_\gamma m_k \left( \max_{i, j \in B_k} \mathrm d_{\mathrm H,\gamma}^{2}(P_{X \mid Z_i}, P_{X \mid Z_j}) \right) \left( \max_{i, j \in B_k} \mathrm d_{\mathrm H,\gamma}^{2}(P_{Y \mid Z_i}, P_{Y \mid Z_j}) \right) \right)^{1/2}.
  \]
  where $C_\gamma$ is a constant depending only on $\gamma$.
\end{corollary}

\begin{corollary}\label{cor:renyi}
  For any $\alpha \in (0, 1)$, and $\gamma > 0$, Type I error is bounded as
  \[
    \mathbb P \left( p\leq \alpha \mid \mathbf Z \right) \leq \alpha +
    \left( \sum_{k=1}^K C'_\gamma m_k \left( \max_{i, j \in B_k} \mathrm{d_{R, \gamma}}(P_{X \mid Z_i} \| P_{X \mid Z_j}) \right) \left( \max_{i, j \in B_k} \mathrm{d_{R, \gamma}}(P_{Y \mid Z_i} \| P_{Y \mid Z_j}) \right) \right)^{1/2},
  \]
  where $C'_\gamma$ is a constant depending only on $\gamma$.
\end{corollary}

In particular, we can compare Corollary~\ref{cor:hellinger} to our original validity result, Theorem \ref{thm:validity_main}, since the generalized $\gamma$-Hellinger distance, with $\gamma=1$, corresponds to total variation distance. In fact, neither result is strictly stronger than the other. While Corollary~\ref{cor:hellinger} is stronger in the sense that it allows for any value of $\gamma$, it is weaker in terms of the result that we obtain for $\gamma=1$, i.e., for total variation distance: while Theorem~\ref{thm:validity_main} bounds Type I error as $\leq \alpha + \delta_n$, applying Corollary~\ref{cor:hellinger} with $\gamma=1$ offers a weaker bound of the form $\leq \alpha + \mathrm{O}(\sqrt{\delta_n})$. 
From a technical level, this arises because the total variation bound is proven without ever passing to the Hellinger distance, whereas the generalized Hellinger and R\'enyi bounds are proven after passing through to the Hellinger distance.
Nonetheless, both results imply asymptotic validity whenever $\delta_n\to 0$.

\subsubsection{Validity under smoothness assumptions}
We will next examine some implications of the above validity results: what guarantees can be obtained by assuming some form of smoothness on the data distribution? In particular,
we will show that the above results are a strict generalization of the validity results from \cite{kim2022local}. 

We first define the relevant Lipschitz classes.

\begin{definition}
  For any positive constant $L$, let $\mathcal P_{\mathrm{H}, \gamma}(L)$ be the collection of distributions $P_{X, Y, Z}$ satisfying that $X \independent Y \mid Z$,
  \[
    \mathrm d_{\mathrm H,\gamma}(P_{X \mid Z = z}, P_{X \mid Z = z'}) \leq L d(z, z'),
    \text{ and }
    \mathrm d_{\mathrm H,\gamma}(P_{Y \mid Z = z}, P_{Y \mid Z = z'}) \leq L d(z, z').
  \]
\end{definition}

Similarly, for R\'enyi divergences, we have the following.
\begin{definition}
  For any positive constant $L$ and $\gamma > 0$, let $\mathcal P_{\mathrm{R}, \gamma}(L)$ be the collection of distributions $P_{X, Y, Z}$ satisfying that $X \independent Y \mid Z$,
  \[
    \mathrm{d_{R, \gamma}^{1/2}}(P_{X \mid Z = z} \| P_{X \mid Z = z'}) \leq L d(z, z'),
    \text{ and }
    \mathrm{d_{R, \gamma}^{1/2}}(P_{Y \mid Z = z} \| P_{Y \mid Z = z'}) \leq L d(z, z').
  \]
\end{definition}

Then, the results in Table~\ref{table:type_I_bounds} result from simply inserting the Lipschitz bounds into the relevant theorem or corollary. In particular, we note that the last three rows of Table~\ref{table:type_I_bounds} match the results of \cite{kim2022local}.

\begin{table}[h]
  \centering
  \begin{tabular}{lcc}
    \toprule
    Class of distributions & Relevant result & Excess Type I error bound \\
    \addlinespace[0.4em]
    \hline
    \addlinespace[0.4em]
    $\mathcal P_{\mathrm{TV}}(L)$ & Theorem \ref{thm:validity_main} & $4 n L^2 w^2$ \\
    $\mathcal P_{\mathrm{H}, \gamma}(L), 1 \leq \gamma \leq 2$ & Corollary \ref{cor:hellinger} & $C_\gamma  n^{1/2} L^\gamma w^\gamma$ \\
    $\mathcal P_{\mathrm{H}, \gamma}(L), \gamma > 2$ & Corollary \ref{cor:hellinger} & $C_\gamma  n^{1/2} L^2 w^2$ \\
    $\mathcal P_{\mathrm{R}, \gamma}(L), \gamma > 0$ & Corollary \ref{cor:renyi} & $C'_\gamma  n^{1/2} L^2 w^2$ \\
    \bottomrule
  \end{tabular}
  \caption{The first row should be interpreted, for example, as: ``as an immediate implication of Theorem \ref{thm:validity_main}, for all distributions in $\mathcal P_{\mathrm{TV}}$ and $\alpha \in (0, 1)$, it holds that $\mathbb P(p \leq \alpha \mid \mathbf Z) \leq \alpha + 4nL^2w^2$.'' The other rows have the same interpretation. Note that, up to a constant factor, the bound on the first row for $\mathcal P_{\mathrm{TV}} = \mathcal P_{\mathrm{H}, 1}$ is the square of the corresponding bound for $\mathcal P_{\mathrm{TV}}$ implied by Corollary \ref{cor:hellinger} on the second row. }
  \label{table:type_I_bounds}
\end{table}

\subsection{Non-Gaussianity in the linear confounder model}\label{sec:non_gaussianity}

Now, we examine the role of the Gaussianity assumption in the power analysis of Section~\ref{sec:linear_model}. 

The only result that strictly requires Gaussianity is Theorem \ref{thm:gaussian_lrt_power} (the calculation of the oracle power). This is because Theorem \ref{thm:gaussian_lrt_power} requires an explicit likelihood to construct the oracle test, whose test statistic is given by the log-likelihood ratio. 

On the other hand, the LPT power calculation extends under the mixed moment conditions stated below, without requiring a Gaussian likelihood. We now provide a generalization of Theorem~\ref{thm:power_lpt_confounder_model_general} and Corollary~\ref{cor:equal_bin_cor} to this broader class of distributions.

Formally, consider the model:
\begin{equation}\label{model:linear_nongaussian_confounder}
  X = f_1(Z) + \beta_{1,n}\,U + \epsilon_1,
  \qquad
  Y = f_2(Z) + \beta_{2,n}\,U + \epsilon_2,
\end{equation}
where $U, \epsilon_1, \epsilon_2$ are mean zero and unit variance, and where $Z,U,\epsilon_1,\epsilon_2$ are mutually independent. Finally, write $\mu_4 = \mathbb E[U^4] < \infty$, and assume also that $\mathbb E[\epsilon_1^4],\mathbb E[\epsilon_2^4], \mathbb E[U^8] < \infty$. The following variant of Theorem \ref{thm:power_lpt_confounder_model_general} holds.

\begin{theorem}\label{thm:power_nongaussian}
  Fix $\alpha\in(0,1/2)$, and suppose the data are drawn from
  model~\eqref{model:linear_nongaussian_confounder} under the moment conditions
  stated above. Consider LPT with statistic~\eqref{eq:bin-statistic-LPT} and
  $K$ bins of sizes $m_1,\dots,m_K$, where $m_k\geq2$, $K\to\infty$, and
  \eqref{eq:bin_size_assmp} holds. If \eqref{eq:linear_assmp_1} and
  \eqref{eq:linear_assmp_2} hold, then its conditional power satisfies
  \[
    \mathbb E[\phi_{\mathrm{LPT}}\mid\mathbf Z]
    =\Phi\!\left(\Phi^{-1}(\alpha)+\mathrm{SNR_{LPT}}\right)+\mathrm{o}_P(1),
  \]
  where
  \[
  \mathrm{SNR_{LPT}} =  \frac{\sum_{k=1}^K \frac{m_k - 1}{m_k} \rho_n}{\sqrt{\sum_{k=1}^K \frac{m_k - 1}{m_k^2} \left(1 + \rho_n^2(\mu_4 - 2) - \frac{\rho_n^2(\mu_4 - 3)}{m_k}\right)}} \cdot ( 1 + \mathrm{o}_P(1)).
  \]
\end{theorem}

\begin{corollary}\label{cor:equal_bin_nongaussian}
  Under the assumptions of Theorem \ref{thm:power_nongaussian}, LPT with bins of equal size $m$ satisfies
  \[
  \mathrm{SNR_{LPT}} =  \sqrt{\frac{m-1}{m}} \cdot \frac{\sqrt{n} \rho_n}{\sqrt{1 + \rho_n^2(\mu_4 - 2) - \frac{\rho_n^2(\mu_4 - 3)}{m}}} \cdot (1 + \mathrm{o}_P(1)).
  \]
\end{corollary}

The proof is essentially the same as Theorem \ref{thm:power_lpt_confounder_model_general}, so we defer it to Appendix \ref{sec:power_linear_proof}.

Theorem \ref{thm:power_nongaussian} is quite similar to Theorem \ref{thm:power_lpt_confounder_model_general}.
In fact, several of the implications highlighted previously still hold. The detection threshold, for example, remains at $\rho_n \propto n^{-1/2}$. One still monotonically gains power from increasing the bin size, and such gains are in general mild; at $m = 2$ we have $\mathrm{SNR_{LPT}} = \frac{\sqrt{n} \rho_n}{\sqrt{2 + \rho^2(\mu_4 - 1)}}$ and as $m \to \infty$, we have $\mathrm{SNR_{LPT}} \to \frac{\sqrt n \rho_n}{\sqrt{1 + \rho^2(\mu_4 - 2)}}$, meaning that we gain at most a factor of
\[
  \frac{\frac{\sqrt n \rho_n}{\sqrt{1 + \rho^2(\mu_4 - 2)}}}{\frac{\sqrt{n} \rho_n}{\sqrt{2 + \rho^2(\mu_4 - 1)}}} = \sqrt{1 + \frac{1 + \rho_n^2}{1 + \rho_n^2(\mu_4 - 2)}}.
\]

One difference here, as compared to the Gaussian setting of Theorem~\ref{thm:power_lpt_confounder_model_general}, is that it is no longer necessarily the case that the SNR is maximized by equal size bins: this is due to the additional term $\frac{\rho_n^2(\mu_4-3)}{m_k}$ appearing in $\mathrm{SNR_{LPT}}$ (in the Gaussian case, this term is not present since the fourth moment is simply $\mu_4 =3$.) But, if the bin sizes $m_k$ are all reasonably large, the same conclusion holds, since we can then approximate 
  \[
  \mathrm{SNR_{LPT}} \approx   \frac{\sum_{k=1}^K \frac{m_k - 1}{m_k}}{\sqrt{\sum_{k=1}^K \frac{m_k - 1}{m_k^2}}} \frac{\rho_n}{\sqrt{1 + \rho_n^2(\mu_4 - 2) }}
  \]
  which is again maximized by choosing bins of equal size.

\subsection{Binary test statistics}\label{sec:binary_power}

We now give the statistic of Example \ref{ex:betting} a detailed analysis, as it trades off a little bit of power for much greater simplicity in both the requisite moment conditions and the formulation of the statistic, as well as the interpretation of the signal-to-noise ratio.

In greater detail, in this section we consider the following setting: bins of size $m_k=2$, with $B_k = \{i_k,j_k\}$ for each $k=1,\dots,K$ where $K=\lfloor n/2\rfloor$, and statistics of the form
\[
  T(\mathbf X, \mathbf Y, \mathbf Z) = \sum_{k=1}^K T_{\mathrm{bet}}(X_{i_k},X_{j_k},Y_{i_k},Y_{j_k})
\]
where each $T_{\mathrm{bet}}$ has values in $\{-1, 1\}$.\footnote{Here, we ignore ties for simplicity.} Let $p_k = \P(T_k=1)$ be the probability that the ``bet'' in the $k$-th bin was correct.

Statistics of this form almost always enjoy the bound of Theorem \ref{thm:validity_linear}, as noted in its subsequent discussion. Specifically, in the notation of Theorem \ref{thm:validity_linear}, we have
\[
  \epsilon_n = \frac{\max_{1, \dots, K}  \mathrm{Range}_k}{\sqrt{\sum_{k=1}^K \mathrm{Var}_k}} = \frac{2}{\sqrt K}
\]
as $\mathrm{Range}_k = 2$ and $\mathrm{Var}_k = 1$. 

\subsubsection{Binary test statistics with oracle information}

We now proceed to a short analysis of its power. Applying Theorem \ref{cor:nnpt_power} (in Appendix \ref{sec:power_of_nnpt_proof}, the full version of Theorem \ref{thm:power_of_nnpt}) immediately gives the following.

\begin{corollary}\label{thm:nnpt_guess_power}
   Whenever $\liminf \sum_{k=1}^K p_k > 0$ and $\liminf \sum_{k=1}^K (1 - p_k) > 0$, the power of the LPT satisfies
  \[
    \mathbb E[\phi_{\mathrm{LPT}} \mid \mathbf Z] = \Phi \left(\Phi^{-1}(\alpha) \cdot \frac{\sqrt{K}}{\left(4\sum_{k=1}^K p_k(1 - p_k)\right)^{1/2}} + \frac{\sum_{k=1}^K (p_k - 1/2)}{\left(\sum_{k=1}^K p_k(1 - p_k)\right)^{1/2}} \right) + \mathrm{o}_P(1).
  \]
\end{corollary}
Note that in local alternatives where information is low, and it is not not possible to construct a bet that is substantially better than random, $p_k \approx 1/2$ and $\frac{\sqrt K}{\left(4 \sum_{k=1}^K p_k(1 - p_k)\right)^{1/2}} \approx 1$. One immediate benefit of this regime is technical: all the moment requirements of Theorem \ref{cor:nnpt_power} are essentially trivial. Namely, they all reduce to $\liminf \sum_{k=1}^K p_k > 0$ and $\liminf \sum_{k=1}^K (1 - p_k) > 0$, which is almost trivially satisfied whenever the betting strategy is not essentially always right (or essentially always wrong). Next, we see that the signal-to-noise ratio is given as
\begin{equation}\label{eq:binary_SNR_formula}
  \mathrm{SNR_{LPT}} = \frac{\sum_{k=1}^K (p_k - 1/2)}{\left(\sum_{k=1}^K p_k(1 - p_k)\right)^{1/2}}.
\end{equation}
within which one may interpret the numerator as ``how much better is my strategy than random guessing?'' since of course random guessing will be correct $1/2$ of the time. In particular, since $K = \lfloor n / 2 \rfloor$,
\[
  \mathrm{SNR_{LPT}} \geq \sqrt{2n} \cdot \frac{1}{K} \sum_{k=1}^K\left(p_k - \frac{1}{2}\right)
\]
as long as one can employ a strategy which is on average better than random guessing (i.e. $\frac{1}{K} \sum_{k=1}^K\left(p_k - \frac{1}{2}\right) > 0$), the LPT accumulates power at a reasonable rate. One can characterize a lower bound on this rate when oracle knowledge is available, akin to Theorem \ref{thm:orc_vs_nnpt_oracle}.

\begin{lemma}\label{lem:binary_power_lb_tv}
  For any bin $B_k = \{i, j\}$, let $\delta_k = \dtv(P_{X, Y \mid Z_i}, P_{X, Y \mid Z_j})$ be the within-bin total variation distance between the conditional distributions. Then, with access to the true likelihood ratios, the oracle guessing strategy achieves,
  \[
    p_k - \frac{1}{2} \geq \frac{1}{2} \max \left\{ \mathrm{d_{TV}}(P_{X, Y \mid Z_i}, P_{X \mid Z_i} \times P_{Y \mid Z_i}), \mathrm{d_{TV}}(P_{X, Y \mid Z_j}, P_{X \mid Z_j} \times P_{Y \mid Z_j}) \right\} - \delta_k.
  \]
\end{lemma}

This first result has two critical parts paralleling Theorem \ref{thm:orc_vs_nnpt_oracle}. The first feature of importance is that in the binary statistic case, the power of a statistic based on oracle information scales with the TV distance to the null, rather than the KL divergence as in Section \ref{sec:oracle}, so that one does not achieve performance comparable to the oracle Neyman--Pearson test in general. The second is that this total variation separation must be larger than the within-bin total variation between the conditional distributions, or else that within-bin heterogeneity threatens to make the above bound trivial.

Moreover, we can still recover scaling comparable to $\mathrm{SNR_{ORC}}$ when the log-likelihood is well concentrated.

\begin{lemma}\label{lem:binary_power_lb_orc}
  Let $\delta_k$ be as the previous lemma and suppose that $n = 2K$ is even for convenience. Suppose that the log-likelihood ratios $\mathrm{LLR^{(i)}}(X_i, Y_i)$ are $\sigma$-sub-Gaussian under both the null and the alternative: that is, there exists $\sigma > 0$ such that
    \[
    \mathbb E_{P_{X, Y \mid Z_i}} \left[ \exp\left(t \left(\mathrm{LLR}^{(i)}(X, Y) - \mathbb E_{P_{X, Y \mid Z_i}}[\mathrm{LLR}^{(i)}(X, Y)] \right)\right) \right] \leq \exp \left( \frac{t^2\sigma^2}{2} \right)
    \]
    and
    \[
    \mathbb E_{P_{X \mid Z_i} \times P_{Y \mid Z_i}} \left[ \exp\left(t \left(\mathrm{LLR}^{(i)}(X, Y) - \mathbb E_{P_{X \mid Z_i} \times P_{Y \mid Z_i}}[\mathrm{LLR}^{(i)}(X, Y)] \right)\right) \right] \leq \exp \left( \frac{t^2\sigma^2}{2} \right)
    \]
    for all $i \in [n]$. Moreover, if $\sigma^2 \leq M\, \mathrm{V}_{\mathrm{KL}, (i)}$ for all $i \in [n]$ and some constant $M \geq 1$, then the oracle guessing strategy of Lemma \ref{lem:binary_power_lb_tv} also achieves
    \[
      \mathrm{SNR_{LPT}} \geq \frac{\mathrm{SNR_{ORC}}}{8 M \sqrt{\log K}}  - \frac{2 \sum_{k=1}^K \delta_k}{\sqrt K} - \mathrm{o}_P(1).
    \]
\end{lemma}

A few remarks are in order. As with Theorem \ref{thm:orc_vs_nnpt_oracle}, this result does not imply that the LPT requires oracle knowledge to be powerful; it merely gives a quantitative baseline for the optimal strategy, and a benchmark which good strategies might hope to match. Next, we discuss the constant $M$: sub-Gaussianity forces $\mathrm{V}_{\mathrm{KL}, (i)} \leq \sigma^2$, so the constant $M$ measures how tight the sub-Gaussian parameter is relative to the actual varentropy. For instance, $M = 1$ exactly when the log-likelihood ratio is Gaussian under the alternative, and for any sub-Gaussian distribution, such an $M$ must exist, just not uniformly over all sub-Gaussian distributions. Finally, there is nothing fundamental about sub-Gaussianity rather than any other type of concentration, and similar results can be proved under various concentration assumptions.

\subsubsection{Special case: linear confounder}

Following the organization of the main text, we give an explicit example and analysis in the case of the linear confounder model \eqref{model:linear_gaussian_confounder}.

We may consider the simple binary $T_k$ given by, in each bin $B_{k} = \{i, j\}$,
\begin{equation}\label{eq:betting_strat}
  T_k = 2 \cdot \mathbbm 1 \{ (X_{i} - X_{j})(Y_{i} - Y_{j}) \geq 0 \} - 1.
\end{equation}
which corresponds to guessing that the larger $X$ goes with the larger $Y$ value. With this statistic, we get the following power guarantee.
\begin{theorem}\label{thm:gaussian_binary_power}
  Consider the linear confounding model (\ref{model:linear_gaussian_confounder}), and suppose that the binning satisfies assumptions \eqref{eq:linear_assmp_1} and \eqref{eq:linear_assmp_2}. For LPT with the statistic (\ref{eq:betting_strat}), the power satisfies
  \[
  \mathrm{SNR_{LPT}} = \frac{\sqrt{2n}\, q_n}{\sqrt{1 - 4q_n^2}} \cdot (1 + \mathrm{o}_P(1))
  \]
        where
  \[
    q_n = \frac{1}{\pi} \arcsin(\rho_n).
  \]
\end{theorem}

This shows that the detection threshold still matches the oracle test in local alternatives, since for small values of $\rho_n$ we get that $q_n \approx \frac{1}{\pi} \rho_n$, so that the signal-to-noise ratio when $\rho_n \to 0$ is on the order of
\[
  \lim_{n \to \infty} \mathrm{SNR_{LPT}} = \lim_{n \to \infty} \frac{\sqrt{2n}\, q_n}{\sqrt{1 - 4q_n^2}} = \frac{\sqrt 2}{\pi} \sqrt{n} \rho_n
\]
which yields the same $\rho_n = \Theta(n^{-1/2})$ threshold as the oracle test, meaning that the restriction to binary $\{\pm 1\}$ valued statistics costs only a constant factor reduction of power. See Table \ref{table:SNR_comp} for details on the constants.

\begin{table}[h]
  \centering
  \begin{tabular}{lcc}
    \toprule
    Test & Asymptotic $\mathrm{SNR}$ & $\lim_{n \to \infty} \frac{\mathrm{SNR_{LPT}}}{\mathrm{SNR_{ORC}}}$ \\
    \addlinespace[0.4em]
    \hline
    \addlinespace[0.4em]
    Oracle likelihood-ratio test & $\dfrac{\sqrt{n} \rho_n}{1 - \rho_n^2}$ & N/A \\
    LPT (Theorem \ref{thm:power_lpt_confounder_model_general}) & $\sqrt{\dfrac{m-1}{m}} \cdot \dfrac{\sqrt n \rho_n}{\sqrt{1 + \rho_n^2}}$ & $\sqrt{\dfrac{m - 1}{m}}$ \\
    Binary LPT (Theorem \ref{thm:gaussian_binary_power}) & $\frac{\sqrt 2}{\pi} \sqrt{n} \rho_n$ & $\dfrac{\sqrt 2}{\pi}$ \\
    \bottomrule
  \end{tabular}
  \caption{Power of LPT with a binary (Theorem \ref{thm:gaussian_binary_power}) and non-binary (Theorem \ref{thm:power_lpt_confounder_model_general}) statistics (with constant bin size $m$) compared to the oracle likelihood-ratio test of Section \ref{sec:oracle_power_orc}. Numerically, $\frac{\sqrt{2}}{\pi} \approx 0.450$, and at $m = 2, \sqrt{\frac{m - 1}{m}} \approx 0.707$.}
  \label{table:SNR_comp}
\end{table}

\section{Proofs from section \ref{sec:validity}}

\subsection{Proof of Theorem \ref{thm:validity_main}}

Theorem \ref{thm:validity_main} follows immediately from substituting the bound in Lemma \ref{lem:bin_tv_bound} into the Type I error bound of Lemma \ref{lem:initial_tv_bound}.
\hfill $\square$

\subsubsection{Proof of Lemma \ref{lem:initial_tv_bound}}
We first claim that $\mathbb P^*(p \leq \alpha \mid \mathbf Z) \leq \alpha$ holds, where $\mathbb P^*$ denotes that the data is drawn from $P^*_k(\mathbf Z)$ within bins and where $P^*_k(\mathbf Z)$ is as defined in Section~\ref{sec:proof_sketch}. Analogously, we write $\mathbb E^*$ for expectation under the same law.

Now, suppose $(\bX,\bY)\sim \otimes_{k=1}^K P^*_k(\mathbf Z)$ conditional on $\bZ$. Thus, by construction, for any $\sigma\in \Pi$, 

\begin{equation}\label{eq:eq_dist_of_XY}
 (\bX, \bY)\overset{D}{=}(\bX_\sigma, \bY).
\end{equation}
Next, we define $p:\mathcal{X}^n\times \mathcal{Y}^n\mapsto [0,1]$ by 
\begin{equation*}
       p(\bx,\by) = \frac{1}{|\Pi|}\sum_{\sigma \in \Pi} \One{T(\bx_{\sigma},\by,\bZ) \geq T(\bx, \by, \bZ)},
  \end{equation*}
and note that $p=p(\bX,\bY)$. Furthermore, by~\eqref{eq:eq_dist_of_XY},
\begin{align*}
    \P^\star(p(\bX,\bY) \leq \alpha \mid \mathbf Z) &= \frac{1}{|\Pi|}\sum_{\sigma\in \Pi} \P^\star(p(\bX_\sigma,\bY) \leq \alpha \mid \mathbf Z)\\
     &= \frac{1}{|\Pi|}\sum_{\sigma\in \Pi} \E^\star\bigl[\One{p(\bX_\sigma,\bY) \leq \alpha} \bigm| \mathbf Z\bigr]\\
     &= \E^\star\left[\frac{1}{|\Pi|}\sum_{\sigma\in \Pi} \One{\frac{1}{|\Pi|}\sum_{\sigma' \in \Pi} \One{T(\bX_{\sigma'},\bY,\bZ) \geq T(\bX_\sigma, \bY, \mathbf Z)} \leq \alpha} \ \middle| \ \mathbf Z\right]\le \alpha,
\end{align*}
where the penultimate step follows by noting that for any fixed $\sigma\in \Pi$, we have $\sigma\circ \Pi=\Pi$ and the last step follows since by noting the deterministic inequality that for any collection $\{t_1,\ldots,t_m\}$ and for any $\alpha\in (0,1)$,
\[
\frac{1}{m}\sum_{j=1}^m \One{\sum_{k=1}^m \One{t_k\ge t_j}\le \alpha} \le \alpha.
\]
For a general version of this deterministic inequality see \citet[Lemma~3]{harrison2012conservative}. This proves our claim.

Finally, since $T(\bX, \bY, \bZ) = T(\bX_\sigma, \bY_\sigma, \bZ)$ by symmetry of $T$, the distribution of $T$ under $(\bX,\bY)\sim \otimes_{k=1}^K P_k(\mathbf Z)$ matches its distribution under$(\bX,\bY)\sim \otimes_{i=1}^n P_{X,Y\mid Z_i}$. Hence, by the definition of total-variation distance, we have
\[
  \mathbb P(p \leq \alpha \mid \mathbf Z) \leq \mathbb P^*(p \leq \alpha \mid \mathbf Z) + \dtv\left(\prod_{k=1}^K P_k(\mathbf Z), \prod_{k=1}^K P^*_k(\mathbf Z)\right) \leq \alpha + \sum_{k=1}^K \dtv\left(P_k(\mathbf Z), P_k^*(\mathbf Z) \right)
\]
as desired. \qed

\subsubsection{Proof of Lemma \ref{lem:bin_tv_bound}}

We break the proof into Lemmas \ref{lem:dtv_product}, \ref{lem:dtv_transposition}, and \ref{lem:dtv_product_swap} to reduce notational burden. In particular, Lemma \ref{lem:bin_tv_bound} follows immediately from applying Lemma \ref{lem:dtv_product} to $\mathbf X_k$ and $\mathbf Y_k$ conditional on $\mathbf Z$. \qed

\begin{lemma}\label{lem:dtv_product}
    Let $m\geq 1$, and define random variables $X=(X_1,\dots,X_m)$, $Y=(Y_1,\dots,Y_m)$, where
    \[X_i\sim P_{X,i},\quad Y_i\sim P_{Y,i},\]
    with $X_1,\dots,X_m,Y_1,\dots,Y_m$ mutually independent.
    Define
    \[\epsilon_X = \max_{1\leq i, j\leq m}\dtv(P_{X,i},P_{X,j}), \quad \epsilon_Y = \max_{1\leq i, j\leq m}\dtv(P_{Y,i},P_{Y,j}).\]
    For any permutation $\sigma\in\cS_m$, with
    \[X_\sigma = (X_{\sigma(1)},\dots,X_{\sigma(m)}),\quad Y_\sigma = (Y_{\sigma(1)},\dots,Y_{\sigma(m)}).\]
    we have
    \[\dtv\big((X_\sigma,Y_\sigma),(X_{\sigma_X},Y_{\sigma_Y})\big)\leq 4(m-1)\cdot \epsilon_X\epsilon_Y,\]
    where $\sigma,\sigma_X,\sigma_Y\iidsim\textnormal{Unif}(\cS_m)$ are sampled independently of $(X,Y)$. 
\end{lemma}

\begin{proof}[Proof of Lemma~\ref{lem:dtv_product}]

First, we will use the fact that
\[\textnormal{If $(\sigma,\pi)\sim\textnormal{Unif}(\cS_m)\times \textnormal{Unif}(\cS_m)$ then $(\sigma,\sigma\circ\pi)\sim\textnormal{Unif}(\cS_m)\times \textnormal{Unif}(\cS_m)$},\]
which holds by the group structure of $\cS_m$ (here $\circ$ denotes the composition of permutations). Therefore,
it is equivalent to prove that
\begin{equation}\label{eqn:random_pi_lem:dtv_product}\begin{split}\dtv\big((X_\sigma,Y_\sigma),&(X_\sigma,Y_{\sigma\circ\pi})\big)\leq 4(m-1)\epsilon_X\epsilon_Y\\&\textnormal{ where }(X,Y,\sigma,\pi)\sim \prod_{i=1}^m P_{X,i}\times \prod_{i=1}^m P_{Y,i}\times \textnormal{Unif}(\cS_m)\times \textnormal{Unif}(\cS_m).\end{split}\end{equation}
In fact, we will prove a strictly stronger statement: we will show that for any \emph{fixed} permutation $\pi\in\cS_m$, 
\begin{equation}\label{eqn:fixed_pi_lem:dtv_product}\begin{split}\dtv\big((X_\sigma,Y_\sigma),&(X_\sigma,Y_{\sigma\circ\pi})\big)\leq 4(m-1)\epsilon_X\epsilon_Y\\&\textnormal{ where }(X,Y,\sigma)\sim \prod_{i=1}^m P_{X,i}\times \prod_{i=1}^m P_{Y,i}\times \textnormal{Unif}(\cS_m).\end{split}\end{equation}
(This is strictly stronger than~\eqref{eqn:random_pi_lem:dtv_product} because, if $\dtv\big((X_\sigma,Y_\sigma),(X_\sigma,Y_{\sigma\circ\pi})\big)\leq 4(m-1)\epsilon_X\epsilon_Y$ holds for every fixed $\pi$, then it also holds if we average over the distribution $\pi\sim\textnormal{Unif}(\cS_m)$.)

From this point on, let $\pi\in\cS_m$ be fixed. If $\pi$ is the identity permutation then the claim~\eqref{eqn:fixed_pi_lem:dtv_product} is trivial. Otherwise, we can write $\pi$ as a composition of transpositions,
\[\pi = \pi_L \circ \dots \circ\pi_1,\]
where each $\pi_\ell$ is a transposition (i.e., a permutation that swaps two indices), and where the total number of transpositions satisfies $L\leq m-1$ (see e.g. Chapter 3.5 of \cite{dummit2003abstract}).
By the triangle inequality, we therefore have
\begin{equation}\label{eqn:triangle_ineq_lem:dtv_product}\begin{split}\dtv\big((X_\sigma,Y_\sigma),(X_\sigma,Y_{\sigma\circ\pi})\big)\leq {}&\\\dtv\big((X_\sigma,Y_\sigma),(X_\sigma,&Y_{\sigma\circ\pi_1})\big) + \dtv\big((X_\sigma,Y_{\sigma\circ\pi_1}),(X_\sigma,Y_{\sigma\circ\pi_2\circ\pi_1})\big) \\&{}+ \dots + \dtv\big((X_\sigma,Y_{\sigma\circ\pi_{L-1}\circ\dots\circ \pi_1}),(X_\sigma,Y_{\sigma\circ\pi_L\circ\dots\circ \pi_1})\big).\end{split}\end{equation}
To help with notation, define
\[Y^{(\ell)} = Y_{\pi^{(\ell)}}\textnormal{ where }\pi^{(\ell)}=\begin{cases} \textnormal{Id}, & \ell = 1,\\ \pi_{\ell-1}\circ\dots\circ\pi_1, & \ell\in\{2,\dots,L\}.\end{cases}\]
We can then rewrite~\eqref{eqn:triangle_ineq_lem:dtv_product} as
\[\dtv\big((X_\sigma,Y_\sigma),(X_\sigma,Y_{\sigma\circ\pi})\big)\leq  \sum_{\ell=1}^L \dtv\big((X_\sigma,Y^{(\ell)}_\sigma),(X_\sigma,Y^{(\ell)}_{\sigma\circ\pi_\ell})\big).\]
Next we will bound each term on the right-hand side. Fix any $\ell\in[L]$. We then have
\[(X,Y^{(\ell)},\sigma)\sim \prod_{i=1}^m P_{X,i} \times \prod_{i=1}^m P_{Y,\pi^{(\ell)}(i)} \times\textnormal{Unif}(\cS_m).\]
By Lemma~\ref{lem:dtv_transposition} (applied with distributions $P_{X,1},\dots,P_{X,m}$ and $P_{Y,\pi^{(\ell)}(1)},\dots,P_{Y,\pi^{(\ell)}(m)}$), we have
\[\dtv\big((X_\sigma,Y^{(\ell)}_\sigma),(X_\sigma,Y^{(\ell)}_{\sigma\circ\pi_\ell})\big)\leq 4\epsilon_X\epsilon_Y,\]
which completes the proof since $L\leq m-1$.
\end{proof}

\begin{lemma}\label{lem:dtv_transposition}
    In the setting of Lemma~\ref{lem:dtv_product}, let $\pi\in\cS_m$ be a fixed transposition.
    Then
    \[\dtv\big((X_\sigma,Y_\sigma),(X_\sigma,Y_{\sigma\circ\pi})\big)\leq 4\epsilon_X\epsilon_Y,\]
    where $\sigma\sim\textnormal{Unif}(\cS_m)$ is sampled independently of $(X,Y)$. 
\end{lemma}
\begin{proof}[Proof of Lemma~\ref{lem:dtv_transposition}]
    Let $\nu\sim\frac12\delta_{\textnormal{Id}}+
    \frac12\delta_\pi$ be independent of $\sigma$ and $(X,Y)$.
    Since $\sigma\circ\nu$ is uniform on $\cS_m$, the distance in the
    lemma equals
    \[
      \dtv\big((X_{\sigma\circ\nu},Y_{\sigma\circ\nu}),
      (X_{\sigma\circ\nu},Y_{\sigma\circ\nu\circ\pi})\big).
    \]
    By joint convexity of total variation distance, it suffices to
    bound this distance conditionally on each fixed value of $\sigma$.

    Fix $\sigma\in\cS_m$. Suppose $\pi$ exchanges positions $a$ and
    $b$, and set $i=\sigma(a)$ and $j=\sigma(b)$. Define the laws on
    the two affected $X$-positions by
    \[
      A_X=P_{X,i}\times P_{X,j},
      \qquad
      A_X'=P_{X,j}\times P_{X,i},
    \]
    and similarly define
    \[
      A_Y=P_{Y,i}\times P_{Y,j},
      \qquad
      A_Y'=P_{Y,j}\times P_{Y,i}.
    \]
    The coordinates outside positions $a$ and $b$ have the same
    product law in both distributions. After removing this common
    factor and averaging over $\nu$, the two full laws on the affected
    coordinates are
    \[
      \frac12\left(A_X\times A_Y+A_X'\times A_Y'\right)
      \quad\text{and}\quad
      \frac12\left(A_X\times A_Y'+A_X'\times A_Y\right).
    \]
    Lemma~\ref{lem:dtv_product_swap} therefore gives the upper bound
    \[
      \dtv(A_X,A_X')\,\dtv(A_Y,A_Y').
    \]
    By the triangle inequality and invariance of total variation under
    tensoring with a common probability measure,
    \begin{align*}
      \dtv(A_X,A_X')
      &\leq
      \dtv(P_{X,i}\times P_{X,j},P_{X,j}\times P_{X,j})\\
      &\quad+
      \dtv(P_{X,j}\times P_{X,j},P_{X,j}\times P_{X,i})\\
      &=2\dtv(P_{X,i},P_{X,j})
      \leq2\epsilon_X.
    \end{align*}
    The same argument gives $\dtv(A_Y,A_Y')\leq2\epsilon_Y$.
    Thus the conditional distance is at most
    $4\epsilon_X\epsilon_Y$ for every fixed $\sigma$. Averaging over
    $\sigma$ completes the proof.
\end{proof}

\begin{lemma}\label{lem:dtv_product_swap}
    Let $P,P'$ be distributions on $\mathcal{X}$ and let $Q,Q'$ be distributions on $\mathcal{Y}$. Then
    \[\dtv\left(\frac{1}{2}(P\times Q + P'\times Q'), \frac{1}{2}(P\times Q' + P'\times Q)\right) \leq \dtv(P,P')\,\dtv(Q,Q').\]
\end{lemma}
\begin{proof}[Proof of Lemma~\ref{lem:dtv_product_swap}]
    By definition of total variation distance, we can write
    \[\begin{cases} P = (1-a)\cdot P_0 + a\cdot P_1, \\ P' = (1-a)\cdot P_0 + a\cdot P_2,\end{cases} \quad \begin{cases} Q = (1-b)\cdot Q_0 + b\cdot Q_1, \\ Q' = (1-b)\cdot Q_0 + b\cdot Q_2,\end{cases}\]
    for $a=\dtv(P,P')$, $b=\dtv(Q,Q')$, and for some distributions $P_0,P_1,P_2$ on $\mathcal{X}$ and $Q_0,Q_1,Q_2$ on $\mathcal{Y}$. We can then calculate
    \begin{multline*}
        \frac{1}{2}(P\times Q + P'\times Q')
        = (1-a)(1-b) \cdot P_0\times Q_0 + a(1-b)\cdot \frac{1}{2}(P_1+P_2)\times Q_0 \\+(1-a)b\cdot P_0 \times \frac{1}{2}(Q_1 + Q_2) + ab \cdot \frac{1}{2}(P_1\times Q_1 + P_2\times Q_2)
    \end{multline*}
and similarly
\begin{multline*}
        \frac{1}{2}(P\times Q' + P'\times Q)
        = (1-a)(1-b) \cdot P_0\times Q_0 + a(1-b)\cdot \frac{1}{2}(P_1+P_2)\times Q_0 \\+(1-a)b\cdot P_0 \times \frac{1}{2}(Q_1 + Q_2) + ab \cdot \frac{1}{2}(P_1\times Q_2 + P_2\times Q_1).
    \end{multline*}
These two decompositions differ only in the last term, which then yields that
\[\dtv\left(\frac{1}{2}(P\times Q + P'\times Q'),\frac{1}{2}(P\times Q' + P'\times Q)\right)\leq ab = \dtv(P,P')\,\dtv(Q,Q'),\]
as desired.
\end{proof}

\subsection{Proof of Theorem \ref{thm:validity_linear}}

Now we turn to establishing Type~I error control of adaptive-LPT with linearly decomposable statistics. We begin with stating and proving a lemma that forms the main step of the proof. We then use this lemma to prove the theorem.

For any two random variables $U,V$ taking values in  $(\RR, \mathcal{B}(\RR))$, we write $\dks(U,V)$ to denote the Kolmogorov--Smirnov distance, given by 
\[\sup_{t\in \RR} |\mathbb P(U \leq t) - \mathbb P(V\le t)|.\]

\begin{lemma}\label{lem:KS_multinomial}
  Let $a_k = (a_{k,1},\dots,a_{k,L})\in \RR^L$ be fixed vector for each $k\in [K]$. Suppose that, for each $k\in [K]$, 
  \[
  A_k = a_{k,I_k}, \quad B_k = a_{k,J_k}\qquad \text{with}\,\, I_k\sim \sum_{\ell=1}^L p_{k,\ell}\cdot \delta_\ell, \text{ and } J_k\sim \textnormal{Unif}([L]),
  \]
  for some $(p_{k,1},\ldots, p_{k,L})\in [0,1]^L$ such that $\sum_{\ell=1}^L p_{k,\ell}=1$.
  Assume further that $A_1,\dots,A_K$ are mutually independent, and likewise $B_1,\dots,B_K$.
                Then,
    \begin{multline*}
        \dks\left(\sum_{k=1}^K A_k,\sum_{k=1}^K B_k\right) \leq \frac{1.12 \cdot \max_{k\in [K]}\textnormal{Range}(a_k)}{\sqrt{\sum_{k=1}^K\textnormal{Var}(a_k)}} \\+ \frac{2}{\sqrt[4]{2\pi}}\sqrt{\sum_{k=1}^K \dtv(A_k,B_k)\cdot \frac{\max_{k\in [K]}\textnormal{Range}(a_k)}{\sqrt{\sum_{k=1}^K\textnormal{Var}(a_k)}}},\hspace{1.4cm}
    \end{multline*}
    where we write for any $k\in [K]$ $\textnormal{Range}(a_k) = \max_\ell a_{k,\ell} - \min_\ell a_{k,\ell}$ and  $\textnormal{Var}(a_k) = \frac{1}{L}\sum_\ell (a_{k,\ell} - \overline{a}_k)^2$.
\end{lemma}
\begin{proof}[Proof of Lemma~\ref{lem:KS_multinomial}]
 
We start by noting that $B_k$'s are independent across $k\in [K]$, and we observe that
\begin{align*}
    \EE{B_k} = \overline{a}_k = \frac{1}{L}\sum_{\ell=1}^L a_{k,\ell},& \qquad \textnormal{Var}(B_k) = \frac{1}{L}\sum_{\ell=1}^L(a_{k,\ell} - \overline{a}_k)^2 = \textnormal{Var}(a_k),\quad \text{and}\\
    \EE{|B_k - \EE{B_k}|^3} &= \frac{1}{L}\sum_{\ell=1}^L|a_{k,\ell} - \overline{a}_k|^3 \leq \textnormal{Range}(a_k)\cdot \textnormal{Var}(a_k).
\end{align*}
Hence, by the Berry--Esseen theorem,
    \[\dks\left(\sum_{k=1}^K B_k, \mathcal{N}\Big(\sum_{k=1}^K \overline{a}_k, \sum_{k=1}^K \sigma^2_k\Big)\right) \leq \frac{0.56\sum_{k=1}^K\EE{|B_k - \EE{B_k}|^3}}{(\sum_{k=1}^K \textnormal{Var}(B_k))^{3/2}} \leq \frac{0.56\max_{k\in [K]}\textnormal{Range}(a_k)}{\sqrt{\sum_{k=1}^K\textnormal{Var}(a_k)}},\]    
    By Lemma~\ref{lem:normal_ks_dw}, therefore, it holds that
    \[\dks\left(\sum_{k=1}^K A_k,\sum_{k=1}^K B_k\right) \leq \frac{1.12 \cdot \max_{k\in [K]}\textnormal{Range}(a_k)}{\sqrt{\sum_{k=1}^K\textnormal{Var}(a_k)}} + 2\sqrt{\frac{\dw\left(\sum_{k=1}^K A_k,\sum_{k=1}^K B_k\right)}{\sqrt{2\pi\sum_{k=1}^K \textnormal{Var}(a_k)}}},\]
    where $\dw$ is the $1$-Wasserstein distance. Now it remains to bound the $1$-Wasserstein distance appearing on the right hand side of the above inequality. Since the $A_k$'s and the $B_k$'s are mutually independent,
    \[\dw\left(\sum_{k=1}^K A_k,\sum_{k=1}^K B_k\right) \leq \sum_{k=1}^K \dw(A_k,B_k).\]
    Further, since $A_k,B_k$ both have the common support $\{a_{k,1},\dots,a_{k,L}\}$, we have
    \[\dw(A_k,B_k) \leq \dtv(A_k,B_k) \cdot \textnormal{Range}(a_k).\]
    Consequently, it holds that
    \[\dw\left(\sum_{k=1}^K A_k,\sum_{k=1}^K B_k\right) \leq \sum_{k=1}^K \dtv(A_k,B_k)\cdot\max_{k=1,\dots,K}\textnormal{Range}(a_k).\]
    Combining all these arguments, the proof follows.
        \end{proof}

\begin{proof}[Proof of Theorem \ref{thm:validity_linear}]
We start by defining $\cF$ to be the $\sigma$-algebra generated by $\bZ$, and the empirical distributions $\bigl\{\frac{1}{m_k}\sum_{i\in B_k}\delta_{X_i}, \frac{1}{m_k}\sum_{i\in B_k}\delta_{Y_i}\bigr\}_{k=1}^K$. Thereby, let $P_k(\cF)$ and $P^*_k(\cF)$ be the conditional distributions, given $\cF$, of $((\bX_k)_{\sigma_k},(\bY_k)_{\sigma_k})$ and $((\bX_k)_{\sigma_k},(\bY_k)_{\sigma'_k})$, respectively, where $\sigma_k, \sigma_k' \sim \mathrm{Unif}(\mathcal S_{m_k})$. Note that $P_k(\cF)$ and $P^*_k(\cF)$ have the same interpretation as $P_k(\bZ)$ and $P^*_k(\bZ)$ except that we now condition on more information. Note that $P^*_k(\mathcal F)$ is uniform over all permutations of (the realized values of) $\mathbf X_k$ and $\mathbf Y_k$.

As in the proof of Theorem \ref{thm:validity_main}, let $\mathbb P^*$ denote that the data is drawn from $P^*_k(\mathbf Z)$ in each bin. 
By an argument analogous to that in the proof of Lemma \ref{lem:initial_tv_bound}, we obtain
\begin{equation}\label{eq:linear_p*_bound}
  \mathbb P^*(p \leq \alpha \mid \mathcal F) \leq \alpha.
\end{equation}

Next, note that for linearly decomposable statistics, the p-value $p$ specializes to
\[
  p =  \frac{1}{|\Pi|} \sum_{\sigma \in \Pi} \One{\sum_{k=1}^K T_k((\mathbf X_k)_{\sigma_k}, \mathbf Y, \mathbf Z) \geq \sum_{k=1}^K T_k(\mathbf X_k, \mathbf Y, \mathbf Z)}.
\]
Now, we apply Lemma~\ref{lem:KS_multinomial} with $a_k=\bigl(T_k((\mathbf X_k)_{\sigma_k}, \mathbf Y_k, \mathbf Z): \sigma_k\in S_{m_k}\bigr)$, $L=(m_k)!$ and $(p_{k,1},\ldots, p_{k,L})$ be the probability mass function of  $T_k((\mathbf X_k), \mathbf Y_k, \mathbf Z)$ given $\mathcal F$. Henceforth, under the notation of Lemma~\ref{lem:KS_multinomial},  by bin-symmetry of $T_k$, $A_k\eqd T_k((\mathbf X_k), \mathbf Y_k, \mathbf Z)$ and $B_k \eqd T_k((\mathbf X'_k), \mathbf Y'_k, \mathbf Z)$ where $(\bX'_k,\bY'_k)\sim P^*_k(\mathcal F)$.  Note $B_k$ is indeed uniformly chosen from $a_k$ as for any fixed permutations $\tau, \tau' \in \mathcal S_{m_k}$, by the symmetry of $T_k$,
\[
  T_k((\mathbf X_k)_{\tau}, \mathbf Y_k, \mathbf Z) = 
  T_k((\mathbf X_k)_{\tau' \circ \tau}, (\mathbf Y_k)_{\tau'}, \mathbf Z)
\]
and
\[
  \mathbb P^*(((\mathbf X_k)_{\sigma_k}, \mathbf Y) = ((\mathbf X_k)_{\tau' \circ \tau}, (\mathbf Y_k)_{\tau'}) \mid \mathcal F) = \frac{m_k!}{(m_k!)^2} = \frac{1}{m_k!}.
\]

Note $B_k$ is uniformly chosen from $a_k$ as for any fixed permutations $\tau, \tau' \in \mathcal S_{m_k}$, by the bin-symmetry of $T_k$,
\[
  T_k((\mathbf X_k)_{\tau}, \mathbf Y_k, \mathbf Z) = 
  T_k((\mathbf X_k)_{\tau' \circ \tau}, (\mathbf Y_k)_{\tau'}, \mathbf Z)
\]
and
\[
  \mathbb P^*(((\mathbf X_k)_{\sigma_k}, \mathbf Y) = ((\mathbf X_k)_{\tau' \circ \tau}, (\mathbf Y_k)_{\tau'}) \mid \mathcal F) = \frac{m_k!}{(m_k!)^2} = \frac{1}{m_k!}.
\]

Therefore, Lemma~\ref{lem:KS_multinomial} gives
\begin{multline*}
  \mathbb P(p\leq \alpha \mid \cF) - \mathbb P^*(p\leq \alpha \mid \cF) \\
  \leq \frac{1.12 \cdot \max_{k\in [K]}\textnormal{Range}_k}{\sqrt{\sum_{k=1}^K\textnormal{Var}_k}}
  + \frac{2}{\sqrt[4]{2 \pi}}\sqrt{\frac{\max_{k\in [K]} \textnormal{Range}_k}{\sqrt{\sum_{k=1}^K \textnormal{Var}_k}}} \cdot \sqrt{\sum_{k=1}^K \dtv(P_k(\cF),P^*_k(\cF))}.
\end{multline*}
Moreover, by Lemma \ref{lem:cond_tv_ineq}, $\dtv(P_k(\bZ),P^*_k(\bZ)) \leq \EEst{\dtv(P_k(\cF),P^*_k(\cF))}{\bZ}$ and by the Cauchy--Schwarz inequality,
\begin{align*}
  &\EEst{\frac{2}{\sqrt[4]{2 \pi}}\sqrt{\frac{\max_{k\in [K]} \textnormal{Range}_k}{\sqrt{\sum_{k=1}^K \textnormal{Var}_k}}} \cdot \sqrt{\sum_{k=1}^K \dtv(P_k(\cF),P^*_k(\cF))}}{ \mathbf Z} \\
  &\qquad \leq \frac{2}{\sqrt[4]{2 \pi}}\sqrt{\EEst{\frac{\max_{k\in [K]} \textnormal{Range}_k}{\sqrt{\sum_{k=1}^K \textnormal{Var}_k}}}{\mathbf Z}} \cdot \sqrt{\sum_{k=1}^K \dtv(P_k(\bZ),P^*_k(\bZ)) } \leq \frac{2}{\sqrt[4]{2\pi}} \sqrt{\epsilon_n \delta_n}.
\end{align*}

Finally, by \eqref{eq:linear_p*_bound} and applying the tower law, this gives
\[
  \mathbb P(p\leq \alpha \mid \bZ) \leq \alpha + 1.12 \cdot \epsilon_n + \frac{2}{\sqrt[4]{2\pi}} \sqrt{\epsilon_n \delta_n} \leq \alpha + \frac{2}{\sqrt[4]{2\pi}} \left( \sqrt{\epsilon_n} + \sqrt{\delta_n}\right) \sqrt{\epsilon_n}
\]
(where we use $2/\sqrt[4]{2\pi} \approx 1.26 > 1.12$).
Now, if $4\epsilon_n < \delta_n$, the aforementioned bound gives
\[
  \mathbb P(p \leq \alpha \mid \mathbf Z) \leq \alpha + \frac{3}{\sqrt[4]{2 \pi}} \sqrt{\epsilon_n \delta_n} \approx 1.89 \sqrt{\epsilon \delta_n}.
\]
Otherwise, if $4\epsilon_n \geq \delta_n$, Theorem \ref{thm:validity_main} gives
\[
  \mathbb P(p\leq \alpha \mid \bZ) \leq \alpha + \delta_n \leq \alpha + 2\sqrt{\epsilon_n \delta_n}.
\]
Combining both cases, the result follows.
\end{proof}

\subsection{Proof of Theorem \ref{thm:vanishing_tv}}

Throughout this proof, we denote the densities, evaluated at $x\in \mathcal{X}$, of the conditional distributions $P_{X \mid Z=z}$ as $p(x \mid z)$, the marginal density of $X $ as $p(x)$, and the common dominating measure as $\mu(x)$.

\paragraph{Proof of first part: convergence in TV.}
We can write
\[
  \dtv(p(\cdot \mid Z_1), p(\cdot \mid Z_{N(1), n})) = \int_\RR \left( p(x \mid Z_1) - p(x \mid Z_{N(1), n})\right)_+ \mathsf{d}\mu(x),
\]
where for nay $a\in \mathbb R$, we write $(a)_+$ to denote $\max\{a,0\}$. Therefore, by Fubini's theorem, interchanging the order of integration (this is where the $\sigma$-finiteness of $\mu$ is needed), it suffices to study
\[
  \lim_{n \to \infty} \int  \E \left[ (p(x \mid Z_1) - p(x \mid Z_{N(1), n}))_+ \right] \mathsf{d}\mu(x).
\]
Note that
\[
  \E \left[ \left(p(x \mid Z_1) - p(x \mid Z_{N(1), n})\right)_+ \right] \leq \E[p(x \mid Z_1)]  = p(x)
\]
so that if for all $x\in \mathcal{X}$, 
\begin{equation}\label{eq:convergence_tv_at_x}
    \lim_{n \to \infty} \E \left[ (p(x \mid Z_1) - p(x \mid Z_{N(1), n}))_+ \right]\to 0,
\end{equation}
then we may apply dominated convergence to obtain
\[
  \int \lim_{n \to \infty} \E \left[ (p(x \mid Z_1) - p(x \mid Z_{N(1), n}))_+ \right] \mathsf{d}\mu(x) = 0.
\]
To see \eqref{eq:convergence_tv_at_x}, by Lemma \ref{lem:conv_in_prob}, 
$p(x \mid Z_1) - p(x \mid Z_{N(1), n}) \to 0$ in probability. Moreover, since $\{p(x \mid Z_{N(1), n})\}_{n=1}^\infty$ is uniformly integrable by Lemma \ref{lem:ui}, we may conclude that $p(x \mid Z_{N(1), n}) \to p(x \mid Z_1)$ in $L^1$ as well, so that we have
\[
  0\le\lim_{n \to \infty} \E \left[ (p(x \mid Z_1) - p(x \mid Z_{N(1), n}))_+ \right] \leq
  \lim_{n \to \infty} \E \left[ |p(x \mid Z_1) - p(x \mid Z_{N(1), n})| \right] = 0,
\]
which establishes~\eqref{eq:convergence_tv_at_x}.

\paragraph{Proof of second part: arbitrarily slow convergence.} For the second half of the result, we will explicitly construct a distribution on $\mathcal{X}\times \mathcal{Z}$, for which $\mathbb E \left[  \dtv(p(\cdot \mid Z_1), p(\cdot \mid Z_{N(1), n})) \right]$ converges arbitrarily slowly. In particular, suppose $Z_1,\ldots, \iidsim \text{Unif}([0, 1])$, and for some set $S \subset [0, 1]$ and for each $i$, $X_i = 1_{Z_i \in S}$. In this case, the TV distance simplifies to
\[
  \dtv(p(\cdot \mid Z_1), p(\cdot \mid Z_{N(1), n})) = \begin{cases}
    0 & 1_{Z_1 \in S} = 1_{Z_{N(1), n} \in S} \\
    1 & \text{otherwise}.
  \end{cases}
\]
so that now onwards, we need to study
\[
\E\left[\dtv(p(\cdot \mid Z_1), p(\cdot \mid Z_{N(1), n}))\right]=\P(1_{Z_1 \in S} \neq 1_{Z_{N(1), n} \in S}) = \P(X_1 \neq X_{N(1)}).
\]
The rest of the proof constructs the set iteratively. First, fix a sequence $\{s_m\}_{m=1}^\infty$ (explicit choice will be stated later in the proof) such that
\[
  \sum_{m=1}^\infty 2^{m-1} s_m = \frac{1}{2}
\]
and set $S_0 = [0, 1]$; then, set $S_1 = [0, \frac{1 - s_1}{2}] \cup [\frac{1 + s_1}{2}, 1]$, which is obtained from $S_0$ by removing the middle segment of length $s_1$. Then, $S_2$ is obtained from $S_1$ by removing the middle segments of length $s_2$ from the two intervals composing $S_1$; generalizing, $S_{m}$ is attained from $S_{m - 1}$ by removing the middle intervals of length $s_{m}$ of the $2^{m-1}$ segments of $S_{m - 1}$. Then, set
\[
  S = \bigcap_{m=1}^\infty S_m.
\]
Note that $S$ is Borel and has Lebesgue measure $1/2$ by construction. See Figure \ref{fig:drawing_s} for an illustration.

\begin{figure}[ht]
  \centering
  {
    \def\svgwidth{0.8\columnwidth}
    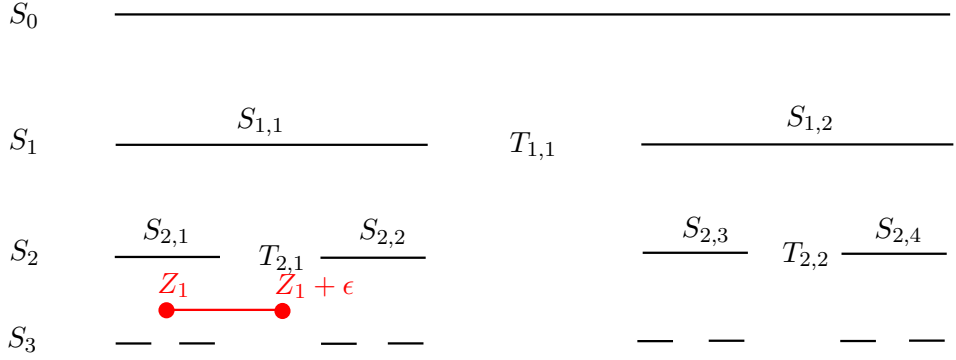
      }
  \caption{We construct a Cantor-like set $S = \bigcap_{m=1}^\infty S_m$ and consider the $Z_1$ which fall into the gaps $T_{m, \ell}$. As long as we choose the lengths $s_m$ of the gaps $T_{m, \ell}$ to decay sufficiently slowly, we will get slow convergence in expectation for the total variation distance between $X_1$ and $X_{n,1}$. To get an intuitive grasp on why the rate of removal matters, consider widening the lowest-most removals $T_{3, \bullet}$ and narrowing the earlier gaps $T_{1, \bullet}, T_{2, \bullet}$ (as shown in bottom figure). The failure to converge in expected TV happens at points which have nearby points both in and out of $S$; this action of widening/narrowing (which corresponds to slowing the rate at which $R_n \to \frac{1}{2}$) ``spreads out'' that failure state.}
  \label{fig:drawing_s}
\end{figure}

We now set some additional notation. 
Let $\lambda(E)$ be the Lebesgue measure of $E$, $p_n(z)$ be the density of $Z_{N(1), n} - Z_1$ (with respect to the Lebesgue measure) conditional on $Z_1$ for $z \in (0, 1 - Z_1]$.

We define
  \[
    f(z_1, z_2) = \begin{cases}
      0 & 1_{z_1 \in S} = 1_{z_2 \in S} \\
      1 & \text{otherwise}
    \end{cases},
  \]
  and write $S_{m, \ell}$ be the $\ell$-th interval of $S_m$, where $1 \leq \ell \leq 2^m$, $T_{m, \ell}$ be the $\ell$-th interval which is removed from $S_{m - 1}$ to get $S_m$, where $1 \leq \ell \leq 2^{m - 1}$, and finally set $R_m = \sum_{k=1}^m 2^{k-1} s_k$ the total amount of mass removed after step $m$, so that each interval $S_{m, \ell}$ is of length $\frac{1 - R_m}{2^m}$, and $R_m \to \frac{1}{2}$ as $m \to \infty$.

Then, for any $z \in (0, 1 - Z_1]$, we have
\begin{align*}
  \P(Z_{N(1), n} - Z_1 \geq z \mid Z_1) &= \P(Z_2, \dots, Z_n \notin [Z_1, Z_1 + z] \text{ and there exists } Z_i \text{ s.t. } Z_i > Z_1 \mid Z_1) \\
 &=\P(\text{there exists}~~ Z_i \text{ s.t. } Z_i > Z_1 \mid Z_2, \dots, Z_n \notin [Z_1, Z_1 + z], Z_1) \\
 &\qquad \times \P(Z_2, \dots, Z_n \notin [Z_1, Z_1 + z] \mid Z_1) \\
 &= \left( 1 - \left(\frac{Z_1}{1 - z}\right)^{n-1} \right) \cdot (1 - z)^{n - 1} = (1 - z)^{n - 1} - Z_1^{n-1}
\end{align*}
and so
\[
  p_n(z) = (n - 1)(1 - z)^{n - 2}
\]
holds for $z \in (0, 1 - Z_1]$. Then, we have that 
\[
  \P(X_1 \neq X_{N(1)} \mid Z_1) = \int_{0}^{1 - Z_1} f(z, Z_1) p_n(z)\mathsf{d}z \geq \int_0^{\epsilon} f(z, Z_1) p_n(z)\mathsf{d}z.
\]

Now take any $m \geq 3$ and consider the set $G_\epsilon$ such that any $Z_1 \in G_\epsilon$ satisfies:
\begin{enumerate}
\item $Z_1 \leq \frac{1}{2}$;
\item after choosing $\epsilon = \frac{1 - R_{m - 1}}{2^{m - 1}}$, the interval $[Z_1, Z_1 + \epsilon]$ contains $S_{m, \ell'}$ for some $\ell'$.
\end{enumerate}
We will compute the measure of $G_\epsilon$ later. Then, noting that $\epsilon \leq 1/2 \leq 1 - Z_1$ and that $p_n(z)$ is decreasing, for $Z_1 \in G_\epsilon$,
\[
  \int_0^{\epsilon} f(z, Z_1) p_n(z)\mathsf{d}z \geq p_n(\epsilon) \cdot \lambda(S_{m, \ell'} \cap S).
\]
We may compute that
\[
  \lambda(S_{m, \ell'} \cap S) = \frac{1 - R_m}{2^m} - \sum_{k=m+1}^\infty 2^{k - (m + 1)}s_k = \frac{1}{2^{m+1}}
\]
and, since $0 \leq \epsilon \leq 1/2$, we have
\[
  1 - \epsilon \geq \exp(-2\epsilon) \implies p_n(\epsilon) = (n - 1)(1 - \epsilon)^{n-2} \gtrsim n\exp(-2n \epsilon) = n \exp \left( -n \cdot \frac{1 - R_{m-1}}{2^{m-2}} \right).
\]
Now we just need to compute the measure of $G_\epsilon$. In particular, consider $Z_1 \in T_{k, \ell}$ for $k, \ell$ as follows:
take any $k > m \geq 3$, and let $1 \leq \ell \leq 2^{k - 4}$ or $2^{k - 3} + 1 \leq \ell \leq 2^{k - 3} + 2^{k - 4}$. Since $\ell \leq 2^{k-3} + 2^{k-4} < 2^{k - 2}$, $Z_1 \in T_{k, \ell} \implies Z_1 \leq \frac{1}{2}$; moreover, the condition on $\ell$ ensures that $T_{k, \ell}$ is on the left side of the bisection of the various intervals $S_{m, \ell'}$. The important part is that for $k >m$, there are $2^{k - 3}$ many $\ell$ which satisfy the condition. Then,
 \[
   \lambda \left( G_\epsilon \right) \geq \sum_{k=m+1}^\infty 2^{k-3}s_k = \frac{\frac{1}{2} - R_m}{4} \gtrsim \frac{1}{2} - R_m.
 \]
 Finally, by integrating the conditional probability over all $G_\epsilon$, we get a lower bound
 \begin{align*}
   \P(X_1 \neq X_{N(1)}) &\geq \lambda \left(G_\epsilon\right) \cdot p_n(\epsilon) \cdot \lambda(S_{m, \ell'} \cap S) \\
                         &\gtrsim  \left( \frac{1}{2} - R_m \right) \cdot n \exp\left( -n \cdot \frac{1 - R_{m - 1}}{2^{m-2}} \right) \cdot \frac{1}{2^{m+1}} \\
                         &\gtrsim  \frac{n}{2^{m+1}} \exp\left( -\frac{n}{2^{m-2}} \right) \left( \frac{1}{2} - R_m \right) \\
   \intertext{and by taking $m = \lfloor \log_2 n \rfloor$ (recall that this bound holds for any choice of $m$), we get}
                         &\gtrsim \frac{1}{2} - R_{\lfloor \log_2 n \rfloor}
 \end{align*}
 Now, since the only requirement we have imposed on the removed mass $R_m$ is that $R_{m} \to \frac{1}{2}$ as $m \to \infty$, we may pick the sequence $\{s_m\}_{m=1}^\infty$ such that $\frac{1}{2} - R_{\lfloor \log_2 n \rfloor}$ converges arbitrarily slowly, which implies the theorem.
 \qed

\section{Proofs from Section \ref{sec:power} }
\subsection{Proof of Theorem \ref{thm:power_of_orc}} \label{sec:power_of_orc_proof}
 
Everything in this section is stated conditional on $\mathbf Z$.
 
In this section, we study the power of the oracle test $\phi_{\textnormal{ORC}}$ defined in~\eqref{eq:ORC_test}. Recall that the test $\phi_{\textnormal{ORC}}$ assumes the knowledge of the log-likelihood ratios, $\mathrm{LLR}^{(i)}(X_i, Y_i)$, and then applies the Neyman--Pearson lemma to construct the uniformly most powerful test between
\[
  H_0: X, Y \mid Z \sim P_{X \mid Z_i} \times P_{Y \mid Z_i} \text{ and } H_1: X, Y \mid Z \sim P_{X, Y \mid Z_i}.
\]
In particular, we will assume that $\phi_{\textnormal{ORC}}$ admits asymptotic Type~I error control so that it rejects if and only if $$\sum_{i=1}^n \mathrm{LLR}^{(i)}(X_i, Y_i) > \tau_\alpha,$$  where $\tau_\alpha$ is an appropriate cutoff to admit asymptotic Type~I error control.
 
We define
\begin{align}\label{eqn:eta_0 and eta_1}
    \eta_0 &:= 0.56\,\frac{\sum_{i=1}^n\E_{P_{X\mid Z_i}\times P_{Y\mid Z_i}}\left[\left|\mathrm{LLR}^{(i)}(X, Y)-\E_{P_{X\mid Z_i}\times P_{Y\mid Z_i}}\left[\mathrm{LLR}^{(i)}(X, Y)\right]\right|^3\right]}{\left(\sum_{i=1}^n \mathrm V^{[0]}_{\mathrm{KL}, (i)}\right)^{3/2}}~~\text{and}\notag \\
    \eta_1 &:= 0.56\,\frac{\sum_{i=1}^n\E_{P_{X,Y\mid Z_i}}\left[\left|\mathrm{LLR}^{(i)}(X, Y)- \E_{P_{X,Y\mid Z_i}}\left[\mathrm{LLR}^{(i)}(X, Y)\right]\right|^3\right]}{\left(\sum_{i=1}^n \mathrm V_{\mathrm{KL}, (i)}\right)^{3/2}}.
\end{align}
We will shortly see that $\eta_0$ and $\eta_1$ appear from the Berry--Esseen Theorem and appears as the approximation error terms of our power guarantees. Typically, we would expect each of them to diminish with sample size $n$.
 
We start by giving a general finite-sample power guarantee of the oracle test~\eqref{eq:ORC_test}, which will later yield an asymptotic approximation.
 
\begin{theorem}\label{thm:general_orc_power}
  Suppose $0\leq\eta_0\leq \tfrac12\min\{\alpha,1-\alpha\}$ and $\eta_1<\infty$. Under the setting of Section~\ref{sec:oracle_power_orc} and the notation defined above, it holds that
  \[
    \left|\E\left[\phi_{\textnormal{ORC}}\mid \bZ\right] - \Phi\left(\Phi^{-1}(\alpha)\cdot\frac{\bigl(\sum_{i=1}^n\textnormal{V}^{[0]}_{\textnormal{KL},(i)}\bigr)^{1/2}}{\bigl(\sum_{i=1}^n\textnormal{V}_{\textnormal{KL},(i)}\bigr)^{1/2}}+\mathrm{SNR}_{\textnormal{ORC}}\right)\right|\le \eta,
  \]
  where
  \[
    \eta:=\ \frac{1}{\sqrt{2\pi}} \frac{\bigl(\sum_{i=1}^n\textnormal{V}^{[0]}_{\textnormal{KL},(i)}\bigr)^{1/2}}{\bigl(\sum_{i=1}^n\textnormal{V}_{\textnormal{KL},(i)}\bigr)^{1/2}}\left(\Phi^{-1}(\alpha + \eta_0) - \Phi^{-1}(\alpha - \eta_0)\right)+\eta_1.
  \] 
\end{theorem}
 
\begin{proof}
    We start by defining some convenient notation: for each $i\in [n]$, let
\begin{multline*}
    \mu_{0,i}=\E_{P_{X\mid Z_i}\times P_{Y\mid Z_i}}\left[\mathrm{LLR}^{(i)}(X, Y)\right],\quad \sigma_{0,i}^2=\mathrm{Var}_{P_{X\mid Z_i}\times P_{Y\mid Z_i}}\left(\mathrm{LLR}^{(i)}(X, Y)\right),\text{~and}\\ \rho_{0,i}^3=\E_{P_{X\mid Z_i}\times P_{Y\mid Z_i}}\left[\left|\mathrm{LLR}^{(i)}(X, Y)-\mu_{0,i}\right|^3\right]\hspace{4cm}
\end{multline*}
denote the mean, variance and the third central moment of the log-likelihood ratio $\mathrm{LLR}^{(i)}(X_i, Y_i)$ under $H_0$. Analogously we write
\begin{multline*}
   \mu_{1,i}=\E_{P_{X,Y\mid Z_i}}\left[\mathrm{LLR}^{(i)}(X, Y)\right],\quad \sigma_{1,i}^2=\mathrm{Var}_{P_{X,Y\mid Z_i}}\left(\mathrm{LLR}^{(i)}(X, Y)\right),\text{~and}\\
 \rho_{1,i}^3=\E_{P_{X,Y\mid Z_i}}\left[\left|\mathrm{LLR}^{(i)}(X, Y)-\mu_{1,i}\right|^3\right] \hspace{4cm}
  \end{multline*}
to denote the mean, variance and the third central moment of the log-likelihood ratio $\mathrm{LLR}^{(i)}(X_i, Y_i)$ under $H_1$. Note that $\sigma_{0, i}^2 = \mathrm V^{[0]}_{\mathrm{KL}, (i)}$ and $\sigma_{1, i}^2 = \mathrm V_{\mathrm{KL}, (i)}$.
    
Now under $H_0$, an application of the Berry--Esseen Theorem gives the following upper and lower bounds to the cutoff $\tau_\alpha$:
    \begin{equation*}
        \sum_{i=1}^n\mu_{0,i}+\left(\sum_{i=1}^n \sigma_{0,i}^2\right)^{1/2} \Phi^{-1}(1 - \alpha - \eta_0)\le \tau_\alpha
        \le \sum_{i=1}^n\mu_{0,i}+\left(\sum_{i=1}^n \sigma_{0,i}^2\right)^{1/2} \Phi^{-1}(1 - \alpha + \eta_0).
    \end{equation*}
    Write
    \[
      q_{\alpha,+}:=\sum_{i=1}^n\mu_{0,i}
      +\left(\sum_{i=1}^n \sigma_{0,i}^2\right)^{1/2}
       \Phi^{-1}(1-\alpha+\eta_0).
    \]
    By applying the upper bound on $\tau_\alpha$, we have that
    \begin{align*}
        &\P_{H_1} \left(\sum_{i=1}^n \mathrm{LLR}^{(i)}(X_i, Y_i) \ge \tau_{\alpha}\right) 
         \ge \P_{H_1} \left(\sum_{i=1}^n \mathrm{LLR}^{(i)}(X_i, Y_i)\ge q_{\alpha,+}\right)\\
        &\quad =\P_{H_1} \Biggl(\sum_{i=1}^n \frac{\left(\mathrm{LLR}^{(i)}(X_i, Y_i)-\mu_{1,i}\right)}{\left(\sum_{i=1}^n \sigma_{1,i}^2\right)^{1/2}} 
        \ge \frac{q_{\alpha,+}-\sum_{i=1}^n\mu_{1,i}}{\left(\sum_{i=1}^n \sigma_{1,i}^2\right)^{1/2}}\Biggr).
        \end{align*}
        Now, by noting that the power of oracle test is given by
        $$\E\left[\phi_{\textnormal{ORC}}\mid \bZ\right]=\P_{H_1} \left(\sum_{i=1}^n \mathrm{LLR}^{(i)}(X_i, Y_i)\ge \tau_{\alpha}\right),$$ 
        another application of the Berry--Esseen Theorem yields the following lower bound
        \[
        \overline\Phi\left(\frac{q_{\alpha,+}-\sum_{i=1}^n\mu_{1,i}}{\left(\sum_{i=1}^n \sigma_{1,i}^2\right)^{1/2}}\right)-\eta_1,
        \]
        which by noting that $|\Phi(x)-\Phi(y)|\le \frac{1}{\sqrt{2\pi}}|x-y|$ for any $x,y\in \RR$, we obtain the following lower bound:
        \begin{multline*}
        \ge \overline\Phi\left(\frac{\sum_{i=1}^n(\mu_{0,i}-\mu_{1,i})+\left(\sum_{i=1}^n \sigma_{0,i}^2\right)^{1/2} \Phi^{-1}(1-\alpha)}{\left(\sum_{i=1}^n \sigma_{1,i}^2\right)^{1/2}}\right) \\-\frac{1}{\sqrt{2\pi}}\frac{\left(\sum_{i=1}^n \sigma_{0,i}^2\right)^{1/2}}{\left(\sum_{i=1}^n \sigma_{1,i}^2\right)^{1/2}}\left(\Phi^{-1}(1-\alpha+\eta_0)-\Phi^{-1}(1-\alpha-\eta_0)\right)-\eta_1.
    \end{multline*}
    This completes the proof for lower bound, after recalling that
    \[
      \mathrm{SNR_{ORC}} = \frac{\sum_{i=1}^n (\mu_{1, i} - \mu_{0, i})}{\left( \sum_{i=1}^n \sigma_{1, i}^2 \right)^{1/2}}
    \]
    by definition. The upper bound on power follows similarly.
\end{proof}
 
While the finite-sample characterization of power holds in general, to state a more interpretable asymptotic characterization of the same, we make the following assumptions on the data generating model.
 
\begin{assumption}[Aggregate Lyapunov condition]\label{assn:bdd_skewness}
  All third absolute centered moments below are finite, and the null and
  alternative log-likelihood-ratio arrays satisfy
  \begin{align*}
    \frac{\sum_{i=1}^n
      \E_{P_{X\mid Z_i}\times P_{Y\mid Z_i}}
      \!\left[\left|\mathrm{LLR}^{(i)}(X,Y)
      -\E_{P_{X\mid Z_i}\times P_{Y\mid Z_i}}
      [\mathrm{LLR}^{(i)}(X,Y)]\right|^3\right]}
         {\left(\sum_{i=1}^n \mathrm V^{[0]}_{\mathrm{KL},(i)}\right)^{3/2}}
      &=\mathrm{o}_P(1),\\
    \frac{\sum_{i=1}^n
      \E_{P_{X,Y\mid Z_i}}
      \!\left[\left|\mathrm{LLR}^{(i)}(X,Y)
      -\E_{P_{X,Y\mid Z_i}}
      [\mathrm{LLR}^{(i)}(X,Y)]\right|^3\right]}
         {\left(\sum_{i=1}^n \mathrm V_{\mathrm{KL},(i)}\right)^{3/2}}
      &=\mathrm{o}_P(1).
  \end{align*}
  Equivalently, \(\eta_0=\mathrm{o}_P(1)\) and
  \(\eta_1=\mathrm{o}_P(1)\) in~\eqref{eqn:eta_0 and eta_1}.
\end{assumption}

A familiar sufficient condition for Assumption~\ref{assn:bdd_skewness} is a
uniform skewness bound together with variance negligibility. For example,
under the null, suppose that for every $i$,
\[
  \E_{P_{X\mid Z_i}\times P_{Y\mid Z_i}}
  \!\left[\left|\mathrm{LLR}^{(i)}(X,Y)
  -\E_{P_{X\mid Z_i}\times P_{Y\mid Z_i}}
  [\mathrm{LLR}^{(i)}(X,Y)]\right|^3\right]
  \leq M\bigl(\mathrm V^{[0]}_{\mathrm{KL},(i)}\bigr)^{3/2}.
\]
Then
\begin{align*}
  &\frac{\sum_{i=1}^n
    \E_{P_{X\mid Z_i}\times P_{Y\mid Z_i}}
    \!\left[\left|\mathrm{LLR}^{(i)}(X,Y)
    -\E_{P_{X\mid Z_i}\times P_{Y\mid Z_i}}
    [\mathrm{LLR}^{(i)}(X,Y)]\right|^3\right]}
       {\left(\sum_{i=1}^n\mathrm V^{[0]}_{\mathrm{KL},(i)}\right)^{3/2}}\\
  &\qquad\leq M\left(
    \frac{\max_i\mathrm V^{[0]}_{\mathrm{KL},(i)}}
         {\sum_{i=1}^n\mathrm V^{[0]}_{\mathrm{KL},(i)}}
    \right)^{1/2}
  =\mathrm{o}_P(1).
\end{align*}
The identical calculation under $P_{X,Y\mid Z_i}$, with
$\mathrm V_{\mathrm{KL},(i)}$ in place of
$\mathrm V^{[0]}_{\mathrm{KL},(i)}$, proves the alternative part. Thus the
previously used uniform skewness condition, together with the first part of
Assumption~\ref{assn:cond_var_nontrivial}, implies the aggregate condition
above.
 
\begin{assumption}\label{assn:cond_var_nontrivial}
    Suppose the conditional variances of the log-likelihood ratios satisfy
    \[
      \frac{\max_{i} \mathrm V^{[0]}_{\mathrm{KL}, (i)}}{\sum_{i=1}^n \mathrm V^{[0]}_{\mathrm{KL}, (i)} } = \mathrm{o}_P(1),\qquad  
      \frac{\max_{i}\mathrm V_{\mathrm{KL}, (i)}}{\sum_{i=1}^n \mathrm V_{\mathrm{KL}, (i)}} = \mathrm{o}_P(1)
    \]
 and that the ratio of null and alternative variances are on average bounded, i.e.
    \[
        \frac{\sum_{i=1}^n \mathrm V_{\mathrm{KL}, (i)}^{[0]}}{\sum_{i=1}^n \mathrm V_{\mathrm{KL}, (i)}} = \mathrm{O}_P(1)
    \]
\end{assumption}
 
\begin{assumption}\label{assn:local_alternative}
  Either $\mathrm{SNR_{ORC}} \to \infty$ in probability or 
  \[
      \frac{\sum_{i=1}^n \mathrm V_{\mathrm{KL}, (i)}^{[0]}}{\sum_{i=1}^n \mathrm V_{\mathrm{KL}, (i)}}  = 1 + \mathrm{o}_P(1).
  \]
\end{assumption}
When the aforementioned assumptions on the log-likelihood ratio hold, we can give the following precise version of Theorem \ref{thm:power_of_orc}.
 
\begin{theorem}[full version of Theorem \ref{thm:power_of_orc}]\label{cor:orc_power}
  Under the setting of Section~\ref{sec:oracle_power_orc} and assumptions~\ref{assn:bdd_skewness} and \ref{assn:cond_var_nontrivial}, it holds that
  \begin{equation}\label{eq:orc_conc_1}
      \mathbb E[\phi_{\mathrm{ORC}} \mid \mathbf Z] = \Phi \left(\Phi^{-1}(\alpha) \cdot \frac{\left(\sum_{i=1}^n \mathrm V_{\mathrm{KL}, (i)}^{[0]}\right)^{1/2}}{\left(\sum_{i=1}^n \mathrm V_{\mathrm{KL}, (i)}\right)^{1/2}}+\mathrm{SNR_{ORC}} \right) + \mathrm{o}_P(1).
  \end{equation}
    Furthermore, when Assumption~\ref{assn:local_alternative} holds, 
    we have 
  \begin{equation}\label{eq:orc_conc_2}
    \mathbb E[\phi_{\mathrm{ORC}} \mid \mathbf Z] = \Phi \left(\Phi^{-1}(\alpha) + \mathrm{SNR_{ORC}} \right) + \mathrm{o}_P(1).
  \end{equation}
\end{theorem}
 
\begin{proof}
  Assumption~\ref{assn:bdd_skewness} gives
  $\eta_0=\mathrm{o}_P(1)$ and $\eta_1=\mathrm{o}_P(1)$ directly. The
  finite-sample requirement
  $\eta_0\leq\tfrac12\min\{\alpha,1-\alpha\}$ therefore holds with
  probability tending to one. The
  aggregate variance ratio in Assumption~\ref{assn:cond_var_nontrivial} is
  $\mathrm{O}_P(1)$, so continuity of $\Phi^{-1}$ shows that the full error
  $\eta$ in Theorem~\ref{thm:general_orc_power} is $\mathrm{o}_P(1)$.
  Equation~\eqref{eq:orc_conc_1} therefore follows from that theorem.

  Under Assumption~\ref{assn:local_alternative}, if
  $\mathrm{SNR_{ORC}}\to\infty$, both \eqref{eq:orc_conc_1} and
  \eqref{eq:orc_conc_2} converge to one. Otherwise the null-to-alternative
  variance ratio is $1+\mathrm{o}_P(1)$, and \eqref{eq:orc_conc_1} reduces to
  \eqref{eq:orc_conc_2}.
\end{proof}

Assumption~\ref{assn:bdd_skewness} now states the aggregate Lyapunov
conditions needed for the two Berry--Esseen approximations directly; the
uniform skewness condition displayed above is a convenient sufficient
condition. Assumption~\ref{assn:cond_var_nontrivial} additionally rules out a
single dominant log-likelihood contribution and compares the aggregate null
and alternative variances. The latter comparison is needed because the oracle
threshold is calibrated using the null variance but evaluated under the
alternative.
 
Broadly speaking, we expect the effect of the term $\frac{\sum_{i=1}^n \mathrm V_{\mathrm{KL}, (i)}^{[0]}}{\sum_{i=1}^n \mathrm V_{\mathrm{KL}, (i)}}$ in the first power expression of the above theorem to be essentially trivial. To understand this, let us consider two specific cases:
\begin{itemize}
    \item  For any fixed alternative, as $n \to \infty$, $\mathrm{SNR_{ORC}} \to \infty$ and power is asymptotically 1, washing out any effect of that term. 
    \item On the other hand, in a sequence of local alternatives which gets increasingly close to the null such that $\limsup_{n \to \infty} \mathrm{SNR_{ORC}} < \infty$, we expect the null variance and the alternative variance to be eventually indistinguishable, i.e. implying
    \[
      \lim_{n \to \infty} \frac{\sum_{i=1}^n \mathrm V_{\mathrm{KL}, (i)}^{[0]}}{\sum_{i=1}^n \mathrm V_{\mathrm{KL}, (i)}} = 1.
    \]
\end{itemize}
Both of these cases are covered in Assumption~\ref{assn:local_alternative}, in which case we get the more stringent second power guarantee.

\subsection{Proof of Theorem \ref{thm:power_of_nnpt}}\label{sec:power_of_nnpt_proof}
We will follow the same structure as in Appendix~\ref{sec:power_of_orc_proof}.
We start by defining some notation. Throughout this section, suppose that for each $k\in [K]$,
$\sigma_k\sim \textnormal{Unif}(\mathcal S_{m_k})$, and accordingly we write $\sigma_k(T_k) = T_k((\mathbf X_k)_{\sigma}, \mathbf Y_k, \mathbf Z)$. 

Similar to the analysis in Appendix~\ref{sec:power_of_orc_proof}, we define 
\begin{align*}
    \eta_0 = \frac{\sum_{k=1}^K\E[|\sigma_k(T_k)|^3 \mid \mathbf Z]}{\delta\bigl(\sum_{k=1}^K\Var(\sigma_k(T_k) \mid \mathbf Z)-\sqrt{\delta^{-1}\sum_{k=1}^K\E[\sigma_k(T_k)^4\mid \mathbf Z]}\,\bigr)^{3/2}} \text{ and }\\
     \eta_1 = \frac{\sum_{k=1}^K \E[|T_k-\E[T_k \mid \mathbf Z]|^3 \mid \mathbf Z]}{\left(\sum_{k=1}^K \mathrm{Var}(T_k \mid \mathbf Z)\right)^{3/2}},
\end{align*}
which appear as Berry--Esseen approximation terms in our analysis.
Now, we state a general finite-sample bound on the power of $\phi_{\rm LPT}$ conditional on $\bZ$.
\begin{theorem}\label{thm:general_lpt_power}
Fix $\alpha\in(0,1/2)$ and $\delta>0$. Suppose that for all $k \in [K]$, $\EEst{\sigma_k(T_k)^4}{\bZ} < \infty$, that the denominator defining $\eta_0$ is positive, and that $0\leq\eta_0\leq\tfrac12\min\{\alpha,1-\alpha\}$ and $\eta_1<\infty$. Under the setting of Section~\ref{sec:oracle_power_orc} and the notation defined above,
    \[
        \left|\E\left[\phi_{\textnormal{LPT}}\mid \bZ\right] - \Phi\left(\Phi^{-1}(\alpha) \cdot \frac{\bigl(\sum_{k=1}^K \mathrm{Var}(\sigma_k(T_k) \mid \mathbf Z)\bigr)^{1/2}}{\bigl(\sum_{k=1}^K\mathrm{Var}(T_k \mid \mathbf Z)\bigr)^{1/2}} + \mathrm{SNR_{LPT}}\right)\right|\le \eta_\star,
    \]
    where 
    \begin{multline*}
      \eta_\star:=\frac{\bigl(\sum_{k=1}^K \Var(\sigma_k(T_k) \mid \mathbf Z) +  \sqrt{\delta^{-1}\sum_{k=1}^K\E[\sigma_k(T_k)^4 \mid \mathbf Z]}\bigr)^{1/2}}{\sqrt{2\pi}\,\bigl(\sum_{k=1}^K\mathrm{Var}(T_k \mid \mathbf Z)\bigr)^{1/2}}\left(\Phi^{-1}(\alpha+\eta_0)-\Phi^{-1}(\alpha-\eta_0)\right)\\
       +\frac{\Phi^{-1}(1-\alpha)}{\sqrt{2\pi}\delta^{1/4}} \frac{\bigl(\sum_{k=1}^K\E[\sigma_k(T_k)^4 \mid \mathbf Z]\bigr)^{1/4}}{\bigl(\sum_{k=1}^K\mathrm{Var}(T_k \mid \mathbf Z)\bigr)^{1/2}}+\eta_1 + 2\delta.
    \end{multline*}
\end{theorem}

\begin{proof}
The proof is split into three key steps. In Step~1, we obtain a data-dependent approximation to $p$, where the dominant terms involve the bin-specific $T_k$ and its conditional moments with respect to the resampling distribution. Next, in Step~2, we derive concentration of these conditional moments and in Step~3, we approximate the distribution of $T_k$ given $\bZ$ to complete the proof.

\paragraph{Step~1: a data-dependent approximation to p-value of $\Phi_{\rm LPT}$.}
Given the notation above, observe that the p-value of $\phi_{\rm LPT}$ can be expressed as
\[
p=\EEst{\sum_{k=1}^K\sigma_k(T_k)\ge \sum_{k=1}^K T_k}{\bX,\bY,\bZ},
\]
where the expectation is taken over the randomness of $(\sigma_1,\ldots,\sigma_{K})$ with $\sigma_k\sim \textnormal{Unif}(\mathcal{S}_{m_k})$ and $(\sigma_1,\ldots,\sigma_{K})$ are mutually independent. Note that the summands are independent, and by definition of $T_k$ in~Section \ref{sec:power_NNPT}, we have
\begin{equation}\label{eq:symmetry_sigma_kT_k}
    \EEst{\sigma_k(T_k)}{\bX,\bY,\bZ} = \frac{1}{m_k!}\sum_{\sigma\in \mathcal S_{m_k}} \sigma(T_k) = 0,
\end{equation}
and the conditional second and third absolute moments of $\sigma_k(T_k)$ given $\mathbf X, \mathbf Y, \mathbf Z$ are:
\begin{align*}
    \EEst{(\sigma_k(T_k))^2}{\bX,\bY,\bZ} = \frac{1}{m_k!}\sum_{\sigma\in \mathcal S_{m_k}} \sigma(T_k)^2,~\text{and}\\
     \EEst{|\sigma_k(T_k)|^3}{\bX,\bY,\bZ} = \frac{1}{m_k!}\sum_{\sigma\in \mathcal S_{m_k}} |\sigma(T_k)|^3,
\end{align*}
 where the expectation is taken over the randomness of $\sigma_k\sim \textnormal{Unif}(S_{m_k})$. Further, note that $\mathbb E[(\sigma_k(T_k))^2 \mid \mathbf Z] = \Var(\sigma_k(T_k) \mid \mathbf Z)$, and we will predominantly use the latter.
  The Berry--Esseen theorem implies that 
  \[
    \left|p-\overline\Phi\!\left(\frac{\sum_{k=1}^K T_k}{\bigl(\sum_{k=1}^K \EEst{(\sigma_k(T_k))^2}{\bX,\bY,\bZ}\bigr)^{1/2}}\right)\right|
    \le \frac{\sum_{k=1}^K  \EEst{|\sigma_k(T_k)|^3}{\bX,\bY,\bZ}}{\bigl(\sum_{k=1}^K \EEst{(\sigma_k(T_k))^2}{\bX,\bY,\bZ}\bigr)^{3/2}}.
  \]
  \paragraph{Step~2: concentration of the conditional moments.}
  Now, we define the event
  \begin{multline*}
    \mathcal{A}:=\Biggl\{\sum_{k=1}^K  \EEst{|\sigma_k(T_k)|^3}{\bX,\bY,\bZ} \le \delta^{-1}\sum_{k=1}^K  \EEst{|\sigma_k(T_k)|^3}{\bZ},\\
    \left|\sum_{k=1}^K  \EEst{(\sigma_k(T_k))^2}{\bX,\bY,\bZ} -\sum_{k=1}^K\Var(\sigma_k(T_k)\mid \bZ)\right|\le \sqrt{\delta^{-1}\sum_{k=1}^K  \EEst{|\sigma_k(T_k)|^4}{\bZ}}\Biggr\}.
  \end{multline*}
  By Markov's inequality, 
  \[
  \PPst{\sum_{k=1}^K  \EEst{|\sigma_k(T_k)|^3}{\bX,\bY,\bZ} \ge \delta^{-1}\sum_{k=1}^K  \EEst{|\sigma_k(T_k)|^3}{\bZ}}{\bZ} \le \delta.
  \]
  Further, by~\eqref{eq:symmetry_sigma_kT_k}, noting that $\EEst{\EEst{\sigma_k(T_k)}{\bX,\bY,\bZ}}{\bZ}=0$, and by Chebyshev's inequality, we have
  \begin{multline*}
      \P\left(\left|\sum_{k=1}^K  \EEst{(\sigma_k(T_k))^2}{\bX,\bY,\bZ} -\sum_{k=1}^K\Var(\sigma_k(T_k)\mid \bZ)\right| \right. \\
      \left. \ge \sqrt{\delta^{-1}\sum_{k=1}^K  \Var\bigl(\EEst{(\sigma_k(T_k))^2}{\bX,\bY,\bZ}\bigm|\bZ\bigr)}\biggm|\bZ\right) \le \delta.
  \end{multline*}

  Finally, by Jensen's inequality and the tower law,
  \begin{align*}
      \sum_{k=1}^K\Var\bigl(\EEst{(\sigma_k(T_k))^2}{\bX,\bY,\bZ}\bigm|\bZ\bigr)&\le \sum_{k=1}^K\EEst{\bigl(\EEst{(\sigma_k(T_k))^2}{\bX,\bY,\bZ}\bigm|\bZ\bigr)^2}{\bZ}\\
      &\le \sum_k\EEst{\EEst{(\sigma_k(T_k))^4}{\bX,\bY,\bZ}}{\bZ}\\
      &= \sum_k\EEst{(\sigma_k(T_k))^4}{\bZ}.
  \end{align*}
Hence, it follows that
\begin{equation}\label{eq:concentr_sec_cond_moment_LPT}
    \PPst{\left|\sum_{k=1}^K  \EEst{(\sigma_k(T_k))^2}{\bX,\bY,\bZ} -\sum_{k=1}^K\Var(\sigma_k(T_k)\mid \bZ)\right|\ge \sqrt{\delta^{-1}\sum_{k=1}^K  \EEst{(\sigma_k(T_k))^4}{\bZ}}}{\bZ} \le \delta,
\end{equation}

and by a union bound,
  \[
    \P(\mathcal{A}) \geq 1-2\delta.
  \]
  \paragraph{Step~3: approximating conditional distribution of $\sum_{k=1}^K T_k$, and completing the proof.}
  Note that, on $\mathcal{A}$,
  \[
    \left|p-\overline\Phi\!\left(\frac{\sum_{k=1}^K T_k}{\bigl(\sum_{k=1}^K \EEst{(\sigma_k(T_k))^2}{\bX,\bY,\bZ}\bigr)^{1/2}}\right)\right|\le \eta_0,
  \]
  and consequently,
  \begin{equation}\label{eq:nested_event}
    \left\{\frac{\sum_{k=1}^K T_k}{\bigl(\sum_{k=1}^K \EEst{(\sigma_k(T_k))^2}{\bX,\bY,\bZ}\bigr)^{1/2}}\ge \Phi^{-1}(1-\alpha+\eta_0)\right\} \cap \mathcal{A} \subset \{p\le \alpha\} \cap \mathcal{A}.
  \end{equation}
  
  Therefore, it follows that by definition of $\mathcal{A}$,
  \begin{align*}
      \P(\{p\le \alpha\}\cap \mathcal{A}\mid \bZ)&\ge \P\!\left(\biggl\{\frac{\sum_{k=1}^K T_k}{\bigl(\sum_{k=1}^K \EEst{(\sigma_k(T_k))^2}{\bX,\bY,\bZ}\bigr)^{1/2}}\ge \Phi^{-1}(1-\alpha+\eta_0)\biggr\}\cap \mathcal{A}\right)\\
      &\ge \P\!\left(\biggl\{\frac{\sum_{k=1}^K T_k}{U_\delta^{1/2}}\ge \Phi^{-1}(1-\alpha+\eta_0)\biggr\}\cap \mathcal{A}\right)\\
      &\ge \P\!\left(\frac{\sum_{k=1}^K T_k}{U_\delta^{1/2}}\ge \Phi^{-1}(1-\alpha+\eta_0)\right)-2\delta.
  \end{align*}
  where we write
    \[
    U_\delta:= \sum_{k=1}^K\Var(\sigma_k(T_k)\mid \bZ) +  \sqrt{\delta^{-1}\sum_{k=1}^K  \EEst{(\sigma_k(T_k))^4}{\bZ}}
    \]
  for compactness. Next, observe that conditional on $\bZ$, $\sum_{k=1}^K T_k$ has independent summands, with the $k$-th summand $T_k$ having mean $\EEst{T_k}{\bZ}$, variance $\Var(T_k\mid \bZ)$ and third moment $\EEst{|T_k|^3}{\bZ}$.
  Therefore, another application of Berry--Esseen Theorem, and~\eqref{eq:concentr_sec_cond_moment_LPT} gives that
  \begin{multline*}
      \mathbb P(p \leq \alpha \mid \bZ) \geq \P(\{p\le \alpha\}\cap \mathcal{A}\mid \bZ)\\ \ge
  \Phi\!\left(\frac{\sum_{k=1}^K\EEst{T_k}{\bZ}}{\bigl(\sum_{k=1}^K\Var(T_k\mid \bZ)\bigr)^{1/2}}-\Phi^{-1}(1-\alpha+\eta_0)\frac{U_\delta^{1/2}}{\bigl(\sum_{k=1}^K\Var(T_k\mid \bZ)\bigr)^{1/2}}\right)- 2\delta - \eta_1.
  \end{multline*}
  Now, note that by triangle inequality and the fact that $\Phi^{-1}(1-t)=-\Phi^{-1}(t)$, we have
  \begin{multline*}
      \left|\Phi^{-1}(1-\alpha+\eta_0)\frac{U_\delta^{1/2}}{\bigl(\sum_{k=1}^K\Var(T_k\mid \bZ)\bigr)^{1/2}}-\Phi^{-1}(1-\alpha)\frac{ \bigl(\sum_{k=1}^K\Var(\sigma_k(T_k)\mid \bZ)\bigr)^{1/2}}{\bigl(\sum_{k=1}^K\Var(T_k\mid \bZ)\bigr)^{1/2}}\right|\\
      \le (\Phi^{-1}(\alpha+\eta_0) - \Phi^{-1}(\alpha-\eta_0)) \frac{U_\delta^{1/2}}{\bigl(\sum_{k=1}^K\Var(T_k\mid \bZ)\bigr)^{1/2}}\\
      + \Phi^{-1}(1-\alpha)\frac{U_\delta^{1/2}-\bigl(\sum_{k=1}^K\Var(\sigma_k(T_k)\mid \bZ)\bigr)^{1/2}}{\bigl(\sum_{k=1}^K\Var(T_k\mid \bZ)\bigr)^{1/2}}.
      \end{multline*}
    Moreover, since
  \[
  U_\delta^{1/2}\le \left(\sum_{k=1}^K\Var(\sigma_k(T_k)\mid \bZ)\right)^{1/2} +  \delta^{-1/4}\left(\sum_{k=1}^K  \EEst{(\sigma_k(T_k))^4}{\bZ}\right)^{1/4},
  \]
  and that $\Phi(\cdot)$ is $1/\sqrt{2\pi}$-Lipschitz,
                                we obtain
  \begin{align*}
    \P(p\le \alpha\mid \bZ)\ge \Phi\!\left(\frac{\sum_{k=1}^K\EEst{T_k}{\bZ}}{\bigl(\sum_{k=1}^K\mathrm{Var}(T_k \mid \mathbf Z)\bigr)^{1/2}}-\Phi^{-1}(1-\alpha)\frac{\bigl(\sum_{k=1}^K \Var(\sigma_k(T_k) \mid \mathbf Z)\bigr)^{1/2}}{\bigl(\sum_{k=1}^K\mathrm{Var}(T_k \mid \mathbf Z)\bigr)^{1/2}}\right)-\eta_\star.
  \end{align*}
    
  The proof of the upper bound follows by an analogous argument and noting that
  \[
    \{p \leq \alpha \} \cap \mathcal A \subset \left\{\frac{\sum_{k=1}^K T_k}{\bigl(\sum_{k=1}^K \E[\sigma_k(T_k)^2 \mid \mathbf X, \mathbf Y, \mathbf Z]\bigr)^{1/2}}\ge \Phi^{-1}(1-\alpha-\eta_0)\right\} \cap \mathcal{A}
  \]
  holds instead of \eqref{eq:nested_event}.
\end{proof}

While the finite-sample characterization of power holds in general, some mild model assumptions give us an interpretable asymptotic characterization of the same, that is a precise version of Theorem \ref{thm:power_of_nnpt}. We first state our assumptions explicitly.

\begin{assumption}[Aggregate Lyapunov condition]\label{assn:T_bdd_kurtosis}
  All fourth moments below are finite, and the observed and permutation
  triangular arrays satisfy
  \begin{align*}
    \frac{\sum_{k=1}^K
      \mathbb E[(T_k-\mathbb E[T_k\mid\mathbf Z])^4\mid\mathbf Z]}
         {\bigl(\sum_{k=1}^K\mathrm{Var}(T_k\mid\mathbf Z)\bigr)^2}
      &=\mathrm{o}_P(1),\\
    \frac{\sum_{k=1}^K\mathbb E[\sigma_k(T_k)^4\mid\mathbf Z]}
         {\bigl(\sum_{k=1}^K\mathrm{Var}(\sigma_k(T_k)\mid\mathbf Z)\bigr)^2}
      &=\mathrm{o}_P(1).
  \end{align*}
\end{assumption}

A corresponding uniform kurtosis bound implies
Assumption~\ref{assn:T_bdd_kurtosis} whenever no single variance dominates.
For either the observed or permutation array, if the centered summands $W_k$
have variances $v_k$ and satisfy
\[
  \mathbb E[W_k^4\mid\mathbf Z]\leq Mv_k^2
  \quad\text{and}\quad
  \frac{\max_k v_k}{\sum_k v_k}=\mathrm{o}_P(1),
\]
then
\[
  \frac{\sum_k\mathbb E[W_k^4\mid\mathbf Z]}
       {(\sum_kv_k)^2}
  \leq M\frac{\max_kv_k}{\sum_kv_k}
  =\mathrm{o}_P(1).
\]

\begin{assumption}[Variance comparability]\label{assn:T_cond_var_nontrivial}
  The aggregate permutation variance is not asymptotically larger than the
  aggregate observed variance:
  \[
    \frac{\sum_{k=1}^K \mathrm{Var}(\sigma_k(T_k) \mid \mathbf Z)}
         {\sum_{k=1}^K \mathrm{Var}(T_k \mid \mathbf Z)}
    = \mathrm{O}_P(1).
  \]
\end{assumption}

\begin{assumption}\label{assn:T_local_alternative}
  Either $\mathrm{SNR_{LPT}} \to \infty$ in probability or 
  \[
       \frac{\sum_{k=1}^K \mathrm{Var}(\sigma_k(T_k) \mid \mathbf Z)}{\sum_{k=1}^K \mathrm{Var}(T_k \mid \mathbf Z)} = 1 + \mathrm{o}_P(1).
  \]
\end{assumption}

\begin{theorem}[full version of Theorem \ref{thm:power_of_nnpt}] \label{cor:nnpt_power}
  Fix $\alpha\in(0,1/2)$. Under the setting of Section~\ref{sec:power_NNPT} and assumptions~\ref{assn:T_bdd_kurtosis} and \ref{assn:T_cond_var_nontrivial}, it holds that
  \begin{equation*}
      \mathbb E[\phi_{\mathrm{LPT}} \mid \mathbf Z] = \Phi \left(  \Phi^{-1}(\alpha)\cdot \frac{\left(\sum_{k=1}^K \mathrm{Var}(\sigma_k(T_k) \mid \mathbf Z) \right)^{1/2}}{\left(\sum_{k=1}^K \mathrm{Var}(T_k \mid \mathbf Z)\right)^{1/2}} + \mathrm{SNR_{LPT}}\right) + \mathrm{o}_P(1).
  \end{equation*}
    Furthermore, when Assumption \ref{assn:T_local_alternative} holds, we have
    \begin{equation*}\label{eq:lpt_conc_2}
      \mathbb E[\phi_{\mathrm{LPT}} \mid \mathbf Z] = \Phi(\Phi^{-1}(\alpha) + \mathrm{SNR_{LPT}}) + \mathrm{o}_P(1).
    \end{equation*}
\end{theorem}
\begin{proof}
  We verify the hypotheses and error terms in
  Theorem~\ref{thm:general_lpt_power} directly. Write
  \begin{align*}
    V^\pi&:=\sum_{k=1}^K \Var(\sigma_k(T_k)\mid\mathbf Z),
    &M_4^\pi&:=\sum_{k=1}^K\E[\sigma_k(T_k)^4\mid\mathbf Z],\\
    V&:=\sum_{k=1}^K\Var(T_k\mid\mathbf Z),
    &M_4&:=\sum_{k=1}^K
      \E[(T_k-\E[T_k\mid\mathbf Z])^4\mid\mathbf Z],
  \end{align*}
  and set
  \[
    r_n:=\frac{M_4}{V^2},\qquad
    r_n^\pi:=\frac{M_4^\pi}{(V^\pi)^2}.
  \]
  Assumption~\ref{assn:T_bdd_kurtosis} is exactly the assertion that
  $r_n=\mathrm{o}_P(1)$ and $r_n^\pi=\mathrm{o}_P(1)$. Choose a
  deterministic sequence $\delta_n\downarrow0$ sufficiently slowly that
  \begin{equation}\label{eq:delta_cond}
    \sqrt{r_n^\pi}=\mathrm{o}_P(\delta_n).
  \end{equation}
  Such a diagonal choice exists because $r_n^\pi=\mathrm{o}_P(1)$.

  We now check the finite-sample theorem's conditions. First,
  \[
    \frac{\sqrt{\delta_n^{-1}M_4^\pi}}{V^\pi}
    =\sqrt{\frac{r_n^\pi}{\delta_n}}=\mathrm{o}_P(1),
  \]
  so the denominator defining $\eta_0$ is positive with probability tending
  to one. Cauchy--Schwarz, first over each summand and then over $k$, gives
  \[
    \sum_{k=1}^K\E[|\sigma_k(T_k)|^3\mid\mathbf Z]
    \leq (V^\pi M_4^\pi)^{1/2}.
  \]
  Consequently,
  \[
    \eta_0
    \leq\frac{\sqrt{r_n^\pi}}
                   {\delta_n(1-\sqrt{r_n^\pi/\delta_n})^{3/2}}
    =\mathrm{o}_P(1),
  \]
  and hence the required bound
  $0\leq\eta_0\leq\tfrac12\min\{\alpha,1-\alpha\}$ also holds with
  probability tending to one. The same Cauchy--Schwarz argument for the
  observed array gives
  \[
    \eta_1
    \leq\frac{(VM_4)^{1/2}}{V^{3/2}}
    =\sqrt{r_n}=\mathrm{o}_P(1).
  \]

  It remains to check the terms in $\eta_\star$. Assumption
  \ref{assn:T_cond_var_nontrivial}, together with
  $\sqrt{\delta_n^{-1}M_4^\pi}=\mathrm{o}_P(V^\pi)$, shows that the
  prefactor multiplying the $\Phi^{-1}$ difference is $\mathrm{O}_P(1)$.
  That difference is $\mathrm{o}_P(1)$ because $\eta_0=\mathrm{o}_P(1)$.
  The remaining fourth-moment term satisfies
  \[
    \frac{(M_4^\pi)^{1/4}}{\delta_n^{1/4}V^{1/2}}
    =\left(\frac{V^\pi}{V}\right)^{1/2}
      \left(\frac{r_n^\pi}{\delta_n}\right)^{1/4}
    =\mathrm{o}_P(1).
  \]
  Together with $\eta_1=\mathrm{o}_P(1)$ and $\delta_n\to0$, this proves
  $\eta_\star=\mathrm{o}_P(1)$ and hence the first assertion.

  Under Assumption~\ref{assn:T_local_alternative}, either $V^\pi/V=1+
  \mathrm{o}_P(1)$, in which case the first assertion immediately reduces to
  the second, or $\mathrm{SNR}_{\mathrm{LPT}}\to\infty$, in which case both
  displayed power expressions converge to one. This proves the second
  assertion.
\end{proof}

Assumptions~\ref{assn:T_bdd_kurtosis}--\ref{assn:T_local_alternative}
serve the same role as assumptions~\ref{assn:bdd_skewness}--\ref{assn:local_alternative}
in the oracle analysis. The distinction is that the LPT null distribution is
generated by local permutations rather than supplied by oracle knowledge.

Assumption~\ref{assn:T_bdd_kurtosis} is an aggregate Lyapunov condition for
the observed and permutation arrays. It is weaker than requiring a uniform
per-bin kurtosis bound together with a no-dominant-bin condition: those more
familiar conditions imply the two displayed aggregate ratios, but are not
needed by the proof. Assumption~\ref{assn:T_cond_var_nontrivial} only compares
the two aggregate variances, ensuring that the simulated null is not
asymptotically more dispersed than the observed statistic.

Assumption~\ref{assn:T_local_alternative} identifies the two regimes in which
$\mathrm{SNR_{LPT}}$ alone governs power. Under fixed alternatives one expects
$\mathrm{SNR_{LPT}}\to\infty$. Under local alternatives, mild regularity
instead gives
\[
  \frac{\sum_{k=1}^K \Var(\sigma_k(T_k)\mid \mathbf Z)}
       {\sum_{k=1}^K \Var(T_k\mid \mathbf Z)}
  =1+\mathrm{o}_P(1).
\]
The model-specific calculations in Section~\ref{sec:linear_model} verify
these alternatives explicitly.

Finally, the proof shows exactly how the approximation error depends on the
two aggregate Lyapunov ratios and the aggregate variance ratio. Thus the first
power guarantee can also be made uniform over classes on which these three
quantities obey uniform rates.

\subsection{Proof of Theorem \ref{thm:orc_vs_nnpt_oracle}}\label{sec:orc_vs_nnpt_proof}
Now, in this section, we prove Theorem \ref{thm:orc_vs_nnpt_oracle}, thereby establishing that $\mathrm{SNR_{LPT}}$ and $\mathrm{SNR_{ORC}}$ are closely related, and with appropriate design choices LPT attains near-optimal power.

To study this, we first define the test-statistic suitably incorporating the oracle knowledge. Take bins $B_1, \dots, B_K$, each of size $2$, and write $B_k = \{i_k, j_k\}$. Then, define for each $\ell \in [n]$,
\[
  f^{(\ell)}(x_i, y_i, x_j, y_j) = \mathrm{LLR}(x_i, y_i \mid Z_\ell) + \mathrm{LLR}(x_j, y_j \mid Z_\ell)
  - \mathrm{LLR}(x_i, y_j \mid Z_\ell) - \mathrm{LLR}(x_j, y_i \mid Z_\ell).
\]
Thereby, for each $k \in [K] $, let
\begin{align*}
  T_k^{(i)}(X_{i_k}, Y_{i_k}, X_{j_k}, Y_{j_k})
    &= f^{(i_k)}(X_{i_k}, Y_{i_k}, X_{j_k}, Y_{j_k}),\\
  T_k^{(j)}(X_{i_k}, Y_{i_k}, X_{j_k}, Y_{j_k})
    &= f^{(j_k)}(X_{i_k}, Y_{i_k}, X_{j_k}, Y_{j_k}).
\end{align*}
and form $T^{(i)} = \sum_{k=1}^K T_k^{(i_k)}$ and $T^{(j)} = \sum_{k=1}^K T_k^{(j_k)}$. To put in words, we consider the statistic formed by (arbitrarily) choosing one of $\{i_{k},j_{k}\}$ to determine the relevant conditional model and evaluate its log-likelihood ratio on all possible pairings of $\{X_{i_k}, Y_{i_k}, X_{j_k}, Y_{j_k}\}$. This construction ensures the desired symmetry and antisymmetry properties, laid out in Section~\ref{sec:power_NNPT}:
\begin{align*}
    &T_k^{(i)}(X_{i_k}, Y_{i_k}, X_{j_k}, Y_{j_k}) = T_k^{(i)}(X_{j_k}, Y_{j_k}, X_{i_k}, Y_{i_k}),\\
    &T_k^{(i)}(X_{i_k}, Y_{i_k}, X_{j_k}, Y_{j_k}) = -T_k^{(i)}(X_{j_k}, Y_{i_k}, X_{i_k}, Y_{j_k})
\end{align*}
and similarly for $T_k^{(j)}$. Now, observe that working directly with either of $T^{(i)}$ or $T^{(j)}$ comes with technical barriers. In particular, within the same bin $B_k$, conditional on $\mathbf{Z}$, $(X_{i_k},Y_{i_k})$ and $(X_{j_k},Y_{j_k})$ come from different conditional distributions. Intuitively, however, the ideal bin choices should be such that $Z_{i_k}\approx Z_{j_k}$, so that approximately $(X_{i_k},Y_{i_k})$ and $(X_{j_k},Y_{j_k})$ are sampled i.i.d.\ from the same conditional distribution. To capture this ideal scenario,
we introduce phantom data points $(X_{i_k}', Y_{i_k}', Z_{i_k})$ and $(X_{j_k}', Y_{j_k}', Z_{j_k})$ where $(X_{i_k}', Y_{i_k}') \sim P_{X, Y \mid Z_{i_k}}$ and $(X_{j_k}', Y_{j_k}') \sim P_{X, Y \mid Z_{j_k}}$ independently of everything else. Then, we define
\[
  (T_k^{(i)})' = f^{(i_k)}(X_{i_k}, Y_{i_k}, X_{i_k}', Y_{i_k}') \text{ and }
  (T_k^{(j)})' = f^{(j_k)}(X_{j_k}', Y_{j_k}', X_{j_k}, Y_{j_k}),
\]
and finally capture the departure of $T^{(i)}$ and $T^{(j)}$ from the ideal scenario by letting
\[
  \delta_{k}^{(1)} := \max \left \{ \left|\mathbb E\left[T_k^{(i)}\right] - \mathbb E\left[(T_k^{(i)})'\right]\right|, \left|\mathbb E\left[T_k^{(j)}\right] - \mathbb E\left[(T_k^{(j)})'\right]\right| \right\}
\]
and 
\[
  \delta_{k}^{(2)} := \max \left\{ \left|\mathrm{Var}\left(T_k^{(i)}\right) - \mathrm{Var}\left((T_k^{(i)})'\right)\right|, \left|\mathrm{Var}\left(T_k^{(j)}\right) - \mathrm{Var}\left((T_k^{(j)})'\right)\right| \right\}.
\]
These quantities measure how the distribution of $T_k$ varies from the ideal setting where the conditional distributions inside each bin are identical. Note that as sample size increases, both of these quantities vanish, as $P_{X, Y \mid Z_{i_k}}\approx P_{X, Y \mid Z_{j_k}}$. Whenever these vanish sufficiently fast, and the variances under null and alternative distributions are comparable, $\mathrm{SNR_{LPT}}$ is comparable to $\mathrm{SNR_{ORC}}$. We formally write down the assumptions below.

\begin{assumption}\label{assn:delta_vanishing}
  The error terms from above satisfy
  \[
    \sum_{k=1}^K \delta_k^{(1)} = \mathrm{o}_P\left(\sqrt{\sum_{k=1}^K \Var(T_k\mid \mathbf Z)}\right) \text{ and } \sum_{k=1}^K \delta_k^{(2)} = \mathrm{o}_P\left(\sum_{k=1}^K \Var(T_k\mid \mathbf Z)\right).
  \]
\end{assumption}

\begin{assumption}\label{assn:var_ratio_bdded}
  The ratio between the average variances of log-likelihood ratio statistics under null and alternative models are bounded, i.e. there is a constant $M$ such that
  \[
    \frac{\sum_{i=1}^n \mathrm V_{\mathrm{KL}, (i)}^{[0]}}{\sum_{i=1}^n \mathrm V_{\mathrm{KL}, (i)}} < M\qquad \text{a.e.}.
  \]

\end{assumption}

\begin{theorem}[full version of Theorem \ref{thm:orc_vs_nnpt_oracle}]\label{thm:orc_vs_nnpt_oracle_precise}
  Under the above setting, and Assumptions \ref{assn:delta_vanishing} and \ref{assn:var_ratio_bdded}, at least one of the LPT tests based on $T^{(i)}$ or $T^{(j)}$ satisfies
  \[
    \mathrm{SNR_{LPT}} \geq \frac{\mathrm{SNR_{ORC}}}{\sqrt{4 + 4M + 8 \sqrt{M}}} - \mathrm{o}_P(1).
  \]
\end{theorem}

\begin{proof}
  First, observe that for each $i \in [n]$,
  \begin{align*}
    \E \left[ \mathrm{LLR}(X_i, Y_i \mid Z_i) \mid \mathbf Z \right] = \E \left[ \mathrm{LLR}(X_i', Y_i' \mid Z_i) \mid \mathbf Z \right] = \mathrm{KL}(P_{X, Y \mid Z_i} \| P_{X \mid Z_i} \times P_{Y \mid Z_i}),\quad \text{and}\\
    \E \left[ \mathrm{LLR}(X_i, Y_i' \mid Z_i)\mid \mathbf Z \right] = \E \left[ \mathrm{LLR}(X_i', Y_i \mid Z_i) \mid \mathbf Z\right] = -\mathrm{KL}( P_{X \mid Z_i} \times P_{Y \mid Z_i} \| P_{X, Y \mid Z_i})
  \end{align*}
  as $X_i, Y_i'$ (and $X_i', Y_i$) are independent from one another.
  Moreover, we can compute
  \begin{align*}
    \Var(\mathrm{LLR}(X_i, Y_i) \mid Z_i)\mid \mathbf Z) =
    \Var(\mathrm{LLR}(X_i', Y_i') \mid Z_i)\mid \mathbf Z) = \mathrm{V}_{\mathrm{KL}, (i)},\quad \text{and}\\
    \Var(\mathrm{LLR}(X_i, Y_i') \mid Z_i)\mid \mathbf Z) =
    \Var(\mathrm{LLR}(X_i', Y_i) \mid Z_i)\mid \mathbf Z) = \mathrm{V}^{[0]}_{\mathrm{KL}, (i)}.
  \end{align*}
Note that with the idealized statistic $(T_k^{(i)})'$, we have $\mathbb E[(T_k^{(i)})'\mid \mathbf Z] = 2\overline{\mathrm{KL}}_{(i_k)}\geq\overline{\mathrm{KL}}_{(i_k)}$.
Moreover, the Cauchy--Schwarz inequality gives the two aggregate bounds
\begin{align*}
  &\sum_{k=1}^K\Var((T_k^{(i)})'\mid\mathbf Z)+\sum_{k=1}^K\Var((T_k^{(j)})'\mid\mathbf Z)\\
  &\quad\leq 2\left(\sum_{i=1}^n\mathrm V_{\mathrm{KL},(i)}+\sum_{i=1}^n\mathrm V^{[0]}_{\mathrm{KL},(i)}+2\sum_{i=1}^n\sqrt{\mathrm V_{\mathrm{KL},(i)}\mathrm V^{[0]}_{\mathrm{KL},(i)}}\right)\\
  &\quad\leq 2(1+\sqrt M)^2\sum_{i=1}^n\mathrm V_{\mathrm{KL},(i)},
\end{align*}
where the last step uses Assumption~\ref{assn:var_ratio_bdded} only after summing, together with Cauchy--Schwarz.

Therefore, Titu's lemma (for instance, see \cite[Chapter~8]{sedrakyan2018algebraic}) gives
\begin{align*}
  &\frac{\left(\sum_{k=1}^K \mathbb E\left[(T_k^{(i)})'\mid \mathbf Z\right]\right)^{2}}{\sum_{k=1}^K \Var\bigl((T_k^{(i)})'\mid \mathbf Z\bigr)}  + \frac{\left(\sum_{k=1}^K \mathbb E\left[(T_k^{(j)})'\mid \mathbf Z\right]\right)^{2}}{\sum_{k=1}^K \Var\bigl((T_k^{(j)})'\mid \mathbf Z\bigr)} \\
  &\qquad\quad \ge \frac{\left(\sum_{k=1}^K \mathbb E\left[(T_k^{(i)})'\mid \mathbf Z\right]+\sum_{k=1}^K \mathbb E\left[(T_k^{(j)})'\mid \mathbf Z\right]\right)^{2}}{\sum_{k=1}^K \Var\bigl((T_k^{(i)})'\mid \mathbf Z\bigr)+\sum_{k=1}^K \Var\bigl((T_k^{(j)})'\mid \mathbf Z\bigr)} \\
  &\qquad\quad\geq \frac{\left(\sum_{k=1}^K \overline{\mathrm{KL}}_{(i_k)} + \overline{\mathrm{KL}}_{(j_k)}\right)^2}{\left( 2 + 2M + 4\sqrt{M} \right)\left( \sum_{k=1}^K \mathrm{V}_{\mathrm{KL}, (i_k)} + \mathrm{V}_{\mathrm{KL}, (j_k)} \right)}
  = \frac{\mathrm{SNR}_{\mathrm{ORC}}^2}{2 + 2M + 4\sqrt{M}},
\end{align*}
which implies that the larger of the two summands on the left must be at least $\frac{\mathrm{SNR}_{\mathrm{ORC}}^2}{4 + 4M + 8\sqrt{M}}$. Without loss of generality, suppose that maximum is attained by $\{(T_k^{(i)})\}_{k=1}^K$. Then, we have
\[
  \frac{\sum_{k=1}^K \mathbb E\left[(T_k^{(i)})'\mid \mathbf Z\right]}{\left( \sum_{k=1}^K \Var\left((T_k^{(i)})'\mid \mathbf Z\right) \right)^{1/2}} \geq \frac{\mathrm{SNR}_{\mathrm{ORC}}}{\sqrt{4 + 4M + 8\sqrt{M}}}.
\]
Relating this to the original statistic $T_k^{(i)}$, we have
\begin{multline*}
  \mathrm{SNR_{LPT}}
  = \frac{\sum_{k=1}^K \E \left[ T_k^{(i)} \mid \mathbf Z \right]}{\left( \sum_{k=1}^K \Var\left(T_k^{(i)}\mid \mathbf Z\right) \right)^{1/2}}
  \geq \frac{\sum_{k=1}^K \mathbb E\left[(T_k^{(i)})'  \mid \mathbf Z\right] - \sum_{k=1}^K \delta^{(1)}_k}{\left( \sum_{k=1}^K \Var\left((T_k^{(i)})'\mid \mathbf Z \right) + \sum_{k=1}^K \delta^{(2)}_k \right)^{1/2}} \\
  \geq \frac{\sum_{k=1}^K \mathbb E\left[(T_k^{(i)})' \mid \mathbf Z \right] }{\left( \sum_{k=1}^K \Var\left((T_k^{(i)})' \mid \mathbf Z\right) \right)^{1/2}}-\mathrm{o}_P(1)
  \geq \frac{\mathrm{SNR_{ORC}}}{\sqrt{4 + 4M + 8\sqrt{M}}}-\mathrm{o}_P(1)
\end{multline*}
  where the penultimate step follows from Assumption \ref{assn:delta_vanishing}.
\end{proof}

Note that if we are under local alternatives such that (for instance, as stated in  Assumption \ref{assn:local_alternative} -- cf. the discussion after Theorem \ref{thm:general_orc_power}) such that
\[
    \frac{\sum_{i=1}^n \mathrm V_{\mathrm{KL}, (i)}^{[0]}}{\sum_{i=1}^n \mathrm V_{\mathrm{KL}, (i)}} = 1 + \mathrm{o}_P(1)
\]
then we may take $M=1+\mathrm{o}_P(1)$, and the constant in the denominator converges to 4, giving the statement of Theorem \ref{thm:orc_vs_nnpt_oracle}.

We remark that Assumption~\ref{assn:delta_vanishing} which posits
that the distribution of $T_k$ is not too different from the ideal setting is crucial: otherwise, the fact that local permutation tests are constrained to simulate the null distribution only via appropriate permutations, whereas the oracle test has no such constraint (note that that this is only possible as the oracle test is designed with true knowledge of the likelihood ratio in mind) is insurmountable. However, practically speaking, this is not too stringent: usually both $\delta_k^{(1)}$ and $\delta_k^{(2)}$ are vanishing as $n \to \infty$, as $P_{X, Y \mid Z_{i_k}}$ and $P_{X, Y \mid Z_{j_k}}$ converge with a good binning strategy, and for non-trivial choices of statistics, $\Var(T_k\mid \mathbf{Z})$ is often of constant order against any fixed alternative.

\subsection{Proof of Theorem \ref{thm:gaussian_lrt_power}}\label{sec:gaussian_lrt_power_proof}

This proof is a straightforward application of Theorem \ref{thm:power_of_orc}. However, given the additional knowledge of the underlying linear confounder model, we can further simplify $\mathrm{SNR_{ORC}}$ as follows.

Under the model class \eqref{model:linear_gaussian_confounder}, we have
\[
  P_{X\mid Z}\times P_{Y\mid Z}= N \left(\mu ,\Sigma_0 \right),
  \quad\text{and}\quad
  P_{X,Y\mid Z}= N \left( \mu,\Sigma_1 \right),
\]
where we write
\[
  \mu :=\begin{bmatrix} f_1(Z)\\ f_2(Z) \end{bmatrix},
  \quad
  \Sigma_0:=\begin{bmatrix} \beta_{1,n}^2+1 & 0\\[2pt] 0 & \beta_{2,n}^2+1 \end{bmatrix}~\text{and}
  \quad
  \Sigma_1:=\begin{bmatrix} \beta_{1,n}^2+1 & \beta_{1,n}\beta_{2,n}\\[2pt] \beta_{1,n}\beta_{2,n} & \beta_{2,n}^2+1 \end{bmatrix}.
\]
The log-likelihood ratio of $P_{X,Y\mid Z}$ with respect to $P_{X\mid Z}\times P_{Y\mid Z}$ equals
\[
  \mathrm{LLR}(X, Y \mid Z) =\frac{1}{2}\log\!\left(\frac{|\Sigma_0|}{|\Sigma_1|}\right)
  -\frac{1}{2}(V - \mu)^\top\bigl(\Sigma_1^{-1}-\Sigma_0^{-1}\bigr)(V-\mu),
\]
where $V:=(X,Y)^\top$. Now, we recall the following standard facts on moments of quadratic forms. With $W\sim N(0,\Sigma)$, and for any symmetric matrix $A$, we have
\[
    \mathbb E[W^\top A W] =  \mathrm{tr}(A\Sigma),\quad
    \Var(W^\top A W) = 2 \mathrm{tr}((A\Sigma)^2).
\]
Therefore, we have
\begin{align*}
    \mathbb E_{P_{X\mid Z_i} \times P_{Y \mid Z_i}}[(V - \mu)^\top(\Sigma_1^{-1} - \Sigma_0^{-1})(V - \mu)] &= \mathrm{tr}((\Sigma_1^{-1} - \Sigma_0^{-1})\Sigma_0) = \mathrm{tr}(\Sigma_1^{-1}\Sigma_0) - 2~~\text{and}\\
    \mathbb E_{P_{X, Y \mid Z_i}}[(V - \mu)^\top(\Sigma_1^{-1} - \Sigma_0^{-1})(V - \mu)] &= \mathrm{tr}((\Sigma_1^{-1} - \Sigma_0^{-1})\Sigma_1) = 2 -\mathrm{tr}(\Sigma_0^{-1}\Sigma_1).
\end{align*}

Then, we can compute
\begin{align*}
  \overline{\mathrm{KL}}_{(i)} &=  \mathbb E_{P_{X, Y \mid Z_i}}[ \mathrm{LLR}^{(i)}(X, Y)] - \mathbb E_{P_{X \mid Z_i} \times P_{Y \mid Z_i}}[\mathrm{LLR}^{(i)}(X, Y) ] \\
                               &= \frac{1}{2} \left(\tr(\Sigma_0^{-1}\Sigma_1) + \tr(\Sigma_1^{-1}\Sigma_0)) - 4 \right) = \frac{\beta_{1, n}^2 \beta_{2, n}^2}{\beta_{1, n}^2 + \beta_{2,n}^2 + 1} = \frac{\rho^2_n}{1 - \rho^2_n}.
\end{align*}

Moreover, we compute
\[
  \mathrm{V}_{\textnormal{KL},(i)}=\mathrm{Var}_{P_{X, Y \mid Z_i}}\left(\mathrm{LLR}^{(i)}(X, Y) \right)
                                  =\frac{1}{2}\mathrm{tr}\left(\bigl(I - \Sigma_{0}^{-1}\Sigma_1\bigr)^2\right)=\frac{\beta_{1,n}^2\,\beta_{2,n}^2}{(\beta_{1,n}^2+1)(\beta_{2,n}^2+1)}=\rho_n^2,
\]
and that
\begin{align*}
  \mathrm{V}_{\textnormal{KL},(i)}^{[0]} &= \frac{1}{2}\mathrm{tr}\left(\bigl(I - \Sigma_{1}^{-1}\Sigma_0\bigr)^2\right)\\
                                         &=\left(\frac{\beta_{1,n}^2\beta_{2,n}^2}{(\beta_{1,n}^2+1)(\beta_{2,n}^2+1)}\right)^2+\frac{\beta_{1,n}^2\beta_{2,n}^2}{\beta_{1,n}^2+\beta_{2,n}^2+1}\cdot \frac{\beta_{1,n}^2+\beta_{2,n}^2}{\beta_{1,n}^2+\beta_{2,n}^2+1}=\rho_n^2 \cdot \frac{1 + \rho_n^2}{(1 - \rho_n^2)^2}.
\end{align*}

Thus, we obtain
\[
  \mathrm{SNR_{ORC}} = \frac{\sum_{i=1}^n \overline{\mathrm{KL}}_{(i)}}{\left( \sum_{i=1}^n \mathrm{V}_{\mathrm{KL}, (i)} \right)^{1/2}} = \frac{ \sqrt{n}\,|\rho_n|}{1 - \rho_n^2}
\]
as desired.

To conclude the proof, it remains to verify the conditions of Theorem \ref{thm:power_of_orc}, as stated formally in Appendix~\ref{sec:power_of_orc_proof}. 
We start by noting that
$\mathrm{LLR}^{(i)}(X, Y) - \mathbb E_{P_{X, Y \mid Z_i}}[\mathrm{LLR}^{(i)}(X, Y)]$ is a quadratic polynomial in $X, Y$, and thus satisfies Gaussian hypercontractivity (see e.g.\ Theorem 5.10 of \cite{Janson_1997}), i.e, we have
\[
  \mathbb E_{P_{X, Y \mid Z_i}}\left[ \left|\mathrm{LLR}^{(i)}(X, Y) - \mathbb E_{P_{X, Y \mid Z_i}}[\mathrm{LLR}^{(i)}(X, Y)] \right|^{3} \right]^{1/3} \!\!\leq 4\,\Var_{P_{X, Y \mid Z_i}}(\mathrm{LLR}^{(i)}(X, Y))^{1/2}.
\]
The same inequality holds under
$P_{X \mid Z_i} \times P_{Y \mid Z_i}$. Moreover, under the linear
confounder model, $\mathrm{V}_{\textnormal{KL},(i)}^{[0]}$ and
$\mathrm{V}_{\textnormal{KL},(i)}$ do not vary in $i$. Hence each aggregate
third-moment ratio in Assumption~\ref{assn:bdd_skewness} is at most
$4^3/\sqrt n=\mathrm{o}(1)$, verifying that assumption.

For Assumption~\ref{assn:cond_var_nontrivial}, the same invariance in $i$
implies
\[
  \frac{\max_{i} \mathrm V^{[0]}_{\mathrm{KL}, (i)}}{\sum_{i=1}^n \mathrm V^{[0]}_{\mathrm{KL}, (i)} }
  = \frac{\max_{i}\mathrm V_{\mathrm{KL}, (i)}}{\sum_{i=1}^n \mathrm V_{\mathrm{KL}, (i)}}
  = \frac{1}{n}=\mathrm{o}_P(1).
\]
This establishes the first condition. Further, we compute
\[
  \frac{\sum_{i=1}^n \mathrm V_{\mathrm{KL}, (i)}^{[0]}}{\sum_{i=1}^n \mathrm V_{\mathrm{KL}, (i)}} = \frac{1 + \rho_n^2}{(1 - \rho_n^2)^2},
\]
which is bounded since by definition, $\rho_n<1$ under the model class~\eqref{model:linear_gaussian_confounder}.
Finally, Assumption \ref{assn:local_alternative} holds as well, since when $\rho_n \to 0$ then 
\[
  \lim_{n \to \infty} \frac{\sum_{i=1}^n \mathrm V_{\mathrm{KL}, (i)}^{[0]}}{\sum_{i=1}^n \mathrm V_{\mathrm{KL}, (i)}} = \lim_{n \to \infty} \frac{1 + \rho_n^2}{(1 - \rho_n^2)^2} = 1
\]
and otherwise when $\rho_n<1$,
\[
\mathrm{SNR_{ORC}}  = \frac{ \sqrt{n} \rho_n}{1 - \rho_n^2}\to\infty.
\]
This completes the proof.
\qed

\subsection{Proof of Theorem \ref{thm:power_lpt_confounder_model_general}}\label{sec:power_linear_proof}

Similarly to the previous section, this proof is an application of Theorem \ref{thm:power_of_nnpt} (or more precisely, Theorem \ref{cor:nnpt_power}). However, the linear confounder model class~\eqref{model:linear_gaussian_confounder} enables more precise characterization of each of the underlying terms.

First, we introduce some definitions. We write
\begin{multline*}
  \overline{\mathbf X}_k = \frac{1}{m_k} \sum_{i \in B_k} X_i,~
  \overline{\mathbf Y}_k = \frac{1}{m_k} \sum_{i \in B_k} Y_i,~
  \overline{f_1(\mathbf Z)}_k = \frac{1}{m_k} \sum_{i \in B_k} f_1(Z_i)\quad\text{and}~
  \overline{f_2(\mathbf Z)}_k = \frac{1}{m_k} \sum_{i \in B_k} f_2(Z_i).
\end{multline*}
As assumed in Section \ref{sec:power}, we verify that $T_k$ is centered under permutations. We use the identity
\[
  T_k = \frac{1}{m_k} \sum_{i \in B_k} \left( X_i -  \overline{\mathbf X}_{k}\right)\left( Y_i - \overline{\mathbf Y}_k \right) = \frac{1}{m_k} \sum_{i \in B_k} X_i Y_i - \overline{\mathbf X}_k \overline{\mathbf Y}_k
\]
and note that 
\[
 \sum_{\sigma \in \mathcal S_{m_k} } \left( \frac{1}{m_k} \sum_{i \in B_k}  X_{\sigma(i)} Y_i \right) = \frac{1}{m_k} \sum_{i \in B_k} Y_i \sum_{\sigma \in \mathcal S_{m_k}} X_{\sigma(i)}   =\frac{(m_k - 1)!}{m_k} \sum_{i \in B_k} \sum_{j \in B_k} X_j Y_i = m_k ! \cdot \overline{\mathbf X}_k \overline{\mathbf Y}_k
\]
so that $T_k$ is already centered under permutation, as required:
\[
  \sum_{\sigma \in \mathcal S_{m_k}} T_k = \sum_{\sigma \in \mathcal S_{m_k} } \left( \frac{1}{m_k} \sum_{i \in B_k}  X_{\sigma(i)} Y_i \right) - m_k! \cdot \overline{\mathbf X}_k \overline{\mathbf Y}_k = 0.
\]

It will be useful to also write $W_i = X_i - f_1(Z_i)$ and $V_i = Y_i - f_2(Z_i)$ for the de-meaned versions of $X_i$ and $Y_i$, as well as \[
  \overline{\mathbf W}_k = \frac{1}{m_k}\sum_{i \in B_k} W_i
  \quad \text{and} \quad
  \overline{\mathbf V}_k = \frac{1}{m_k}\sum_{i \in B_k} V_i
\]
so that we may write
\[
  \overline{\mathbf X}_k = \overline{\mathbf W}_k + \overline{f_1(\mathbf Z)}_k
  \quad \text{and} \quad
  \overline{\mathbf Y}_k = \overline{\mathbf V}_k +  \overline{f_2(\mathbf Z)}_k.
\]
Note that
\[
  (W_i, V_i) \mid \mathbf Z \sim N \left( 0, \begin{bmatrix}
    \beta_{1, n}^2 + 1 && \beta_{1, n} \beta_{2, n} \\
    \beta_{1, n} \beta_{2, n} && \beta_{2, n}^2 + 1 \\
  \end{bmatrix} \right).
\]
Finally, we write an additional error term
\[
  S_k^{(12)} = \frac{1}{m_k} \sum_{i \in B_k} \left( f_1(Z_i) - \overline{f_1(\mathbf Z)}_k \right) \left( f_2(Z_i) -  \overline{f_2(\mathbf Z)}_k \right)
\]
with the understanding that $|S_k^{(12)}| \leq \sqrt{S_k^{(1)} S_k^{(2)}}$ by Cauchy--Schwarz.

\paragraph{Computing $\mathbb E[T_k \mid \mathbf Z]$.} We start by decomposing
\begin{align*}
  T_k &= \frac{1}{m_k} \sum_{i \in B_k} \bigl( f_1(Z_i) - \overline{f_1(\mathbf Z)}_k + W_i - \overline{\mathbf W}_k \bigr) \bigl( f_2(Z_i) - \overline{f_2(\mathbf Z)}_k + V_i - \overline{\mathbf V}_k \bigr) \\
       &= S_k^{(12)} + \underbrace{\frac{1}{m_k} \sum_{i \in B_k} \bigl( f_1(Z_i) - \overline{f_1(\mathbf Z)}_k  \bigr)\left( V_i - \overline{\mathbf V}_k \right)}_{:= A_{k, 1}} \\
        &\qquad\quad+ \underbrace{\frac{1}{m_k} \sum_{i \in B_k} \bigl( f_2(Z_i) - \overline{f_2(\mathbf Z)}_k\bigr)\left(W_i - \overline{\mathbf W}_k\right)}_{:= A_{k, 2}}
  + \underbrace{\frac{1}{m_k} \sum_{i \in B_k} \left( W_i - \overline{\mathbf W}_k \right)\left( V_i - \overline{\mathbf V}_k \right)}_{:= A_{k, 3}}.
\end{align*}
Observe that by linearity, $\EEst{A_{k, 1}}{\mathbf Z} = \EEst{A_{k, 2}}{\mathbf Z} = 0$, and further that,
\[
  \EEst{A_{k, 3}}{\mathbf Z} = \frac{m_k - 1}{m_k}\Cov \left( W_i, V_i \mid \mathbf Z \right) = \frac{m_k - 1}{m_k} \beta_{1,n}\beta_{2,n}.
\]
Consequently, it follows that
\[
  \EEst{T_k}{\mathbf Z} = S_k^{(12)} + \frac{m_k - 1}{m_k}\beta_{1,n}\beta_{2,n}.
\]

\paragraph{Computing $\Var(T_k \mid \mathbf Z)$.} Now, moving on to the conditional variance of $T_k$, we reuse the decomposition from above to write
\begin{align*}
    \Var(T_k \mid \mathbf Z) &= \Var(A_{k, 1} + A_{k, 2} + A_{k, 3} \mid \mathbf Z)\\
    &=\Var(A_{k,1}\mid \mathbf{Z})+\Var(A_{k,2}\mid \mathbf{Z})+ \Var(A_{k,3}\mid \mathbf{Z}) \\&+\mathrm{Cov}(A_{k,1},A_{k,2}\mid \mathbf{Z})+\mathrm{Cov}(A_{k,2},A_{k,3}\mid \mathbf{Z})+\mathrm{Cov}(A_{k,1},A_{k,3}\mid \mathbf{Z}).
\end{align*}
Thereby, we can compute
\begin{align*}
  \Var(A_{k, 1} \mid \mathbf Z) &= \frac{1}{m_k^2} \sum_{i, j \in B_k}  \left( f_1(Z_i) - \overline{f_1(\mathbf Z)}_k  \right) \left( f_1(Z_j) - \overline{f_1(\mathbf Z)}_k  \right) \Cov(V_i - \overline{\mathbf V}_k, V_j - \overline{\mathbf V}_{k} \mid \mathbf Z) \\
                                &= \frac{1}{m_k^2} \sum_{i, j \in B_k}  \left( f_1(Z_i) - \overline{f_1(\mathbf Z)}_k  \right) \left( f_1(Z_j) - \overline{f_1(\mathbf Z)}_k  \right) (\beta_{2, n}^2 + 1) \left( \One{i = j} - \frac{1}{m_k} \right)\\
                                &= \frac{\beta_{2, n}^2 + 1}{m_k}S_k^{(1)}
\end{align*}
where the last equality follows by noting that
\[
  \sum_{i, j \in B_k}  \left( f_1(Z_i) - \overline{f_1(\mathbf Z)}_k  \right) \left( f_1(Z_j) - \overline{f_1(\mathbf Z)}_k  \right) = \Bigl( \sum_{i \in B_k} \left( f_1(Z_i) - \overline{f_1(\mathbf Z)}_k \right) \Bigr)^2 = 0.
\]
By a symmetric argument, we can compute
\[
  \Var(A_{k, 2} \mid \mathbf Z) = \frac{\beta_{1, n}^2 + 1}{m_k} S_k^{(2)}.
\]
Next, we write $R_i=(W_i,V_i)$ for each $i\in [n]$. Writing $\mathbf R_k=(R_1,\ldots,R_{m_k})\in \mathbb{R}^{m_k\times 2}$, we observe that $\sum_{i \in B_k}(W_i - \overline{\mathbf W}_k)(V_i - \overline{\mathbf V}_k)$ is the off-diagonal elements of the $2\times 2$ matrix $(\mathbf{R}_k-\bar{\mathbf{R}})^\top (\mathbf{R}_k-\bar{\mathbf{R}})$. Further, recalling that for each $i$, $R_i\sim N ( 0, \Sigma)$ where
\[
\Sigma:=\begin{bmatrix}
    \beta_{1, n}^2 + 1 && \beta_{1, n} \beta_{2, n} \\
    \beta_{1, n} \beta_{2, n} && \beta_{2, n}^2 + 1 \\
  \end{bmatrix}
\]
As a result, $\mathbf{R}_k-\bar{\mathbf{R}}$ is a centered data-matrix, and $(\mathbf{R}_k-\bar{\mathbf{R}})^\top (\mathbf{R}_k-\bar{\mathbf{R}})$ is a Wishart matrix. In particular, 
\[
(\mathbf{R}_k-\bar{\mathbf{R}})^\top (\mathbf{R}_k-\bar{\mathbf{R}})\sim \mathrm{W}_{2}(\Sigma, m_{k}-1).
\]
Finally, we can compute
\begin{align*}
    \Var(A_{k, 3} \mid \mathbf Z)=&\frac{1}{m_k^2}\Var\left(\sum_{i \in B_k}(W_i - \overline{\mathbf W}_k)(V_i - \overline{\mathbf V}_k) \mid \mathbf Z\right)\\ =&\frac{(m_k - 1)}{m_k^2}\left(  \beta_{1, n}^2\beta_{2, n}^2 + (\beta_{1, n}^2 + 1)(\beta_{2, n}^2 + 1) \right).
\end{align*}

Next, we need to compute the covariances between $A_{k, 1}, A_{k, 2}, A_{k, 3}$. First, by an argument analogous to the one used to compute $\Var(A_{k, 1} \mid \mathbf Z)$, we observe that
\begin{align*}
  \Cov(A_{k, 1}, A_{k, 2} \mid \mathbf Z) &= \frac{1}{m_k^2}\sum_{i, j \in B_k}\bigl(f_1(Z_i) - \overline{f_1(\mathbf Z)}_k\bigr)\bigl(f_2(Z_j) - \overline{f_2(\mathbf Z)}_k\bigr)\Cov(V_i - \overline{\mathbf V}_k, W_j - \overline{\mathbf W}_k \mid \mathbf Z) \\
  &= \frac{\beta_{1,n}\beta_{2,n}}{m_k^2}\sum_{i,j \in B_k}\bigl(f_1(Z_i) - \overline{f_1(\mathbf Z)}_k\bigr)\bigl(f_2(Z_i) - \overline{f_2(\mathbf Z)}_k\bigr)\left( \One{i = j} - \frac{1}{m_k} \right) \\&= \frac{\beta_{1,n}\beta_{2,n}\, S_k^{(12)}}{m_k}.
\end{align*}
Finally, note that conditional on $\mathbf Z$, for each $k$, $\mathbf{R}_k\overset{d}{=}-\mathbf{R}_k$. Consequently, for each $i,j$, we have
\[
\Cov \left( (V_i - \overline{\mathbf V}_k), (W_j - \overline{\mathbf W}_k)(V_j - \overline{\mathbf V}_k)\right)=-\Cov \left( (V_i - \overline{\mathbf V}_k), (W_j - \overline{\mathbf W}_k)(V_j - \overline{\mathbf V}_k)\right)=0
\]
Therefore, it holds that
\[
  \Cov(A_{k, 1}, A_{k, 3} \mid \mathbf Z) = \frac{1}{m_k^2} \sum_{i, j\in B_k} \left(f_1(Z_i) - \overline{f_1(\mathbf Z)_k} \right) \Cov \left( (V_i - \overline{\mathbf V}_k), (W_j - \overline{\mathbf W}_k)(V_j - \overline{\mathbf V}_k)\right)=0.
\]
By a symmetric argument, $\Cov(A_{k, 2}, A_{k, 3} \mid \mathbf Z)=0$.
Finally, combining all the parts, we conclude that
\begin{multline}\label{eq:var_nonperm}
  \Var(T_k \mid \mathbf Z) = \Var(A_{k, 1} \mid \mathbf Z) + \Var(A_{k, 2} \mid \mathbf Z)  + 2\Cov(A_{k, 1}, A_{k, 2} \mid \mathbf Z) + \Var(A_{k, 3} \mid \mathbf Z)\\
  = \frac{(\beta_{2, n}^2 + 1)\,S_k^{(1)} + (\beta_{1, n}^2 + 1)\,S_k^{(2)} + 2\beta_{1, n}\beta_{2, n}\,S_k^{(12)}}{m_k} + \frac{(m_k - 1)}{m_k^2}(\beta_{1, n}^2\beta_{2, n}^2 + (\beta_{1, n}^2 + 1)(\beta_{2, n}^2 + 1)).
\end{multline}

\paragraph{Computing $\mathrm{SNR}_{\mathrm{LPT}}$.} Now, under the model class~\eqref{model:linear_gaussian_confounder}, observe that $\mathrm{SNR_{LPT}}$ can be simplified as follows:
\begin{align*}
  D_k:=\frac{(\beta_{2,n}^2+1)S_k^{(1)}+(\beta_{1,n}^2+1)S_k^{(2)}
                 +2\beta_{1,n}\beta_{2,n}S_k^{(12)}}{m_k}
  +\frac{(m_k-1)(\beta_{1,n}^2+1)(\beta_{2,n}^2+1)(1+\rho_n^2)}{m_k^2}.
\end{align*}
\begin{align*}
  \mathrm{SNR_{LPT}} = \frac{\sum_{k = 1}^K \mathbb E[T_k \mid \mathbf Z]}{ \left( \sum_{k=1}^K \Var(T_k \mid \mathbf Z) \right)^{1/2}}
  =\frac{\sum_{k=1}^K \left(S_k^{(12)} + \dfrac{m_k - 1}{m_k}\beta_{1,n}\beta_{2,n}\right)}{\left(\sum_{k=1}^K D_k\right)^{1/2}}.
\end{align*}

Observe that
\begin{align*}
  \left|\sum_{k=1}^K S_k^{(12)}\right| \leq \sum_{k=1}^K \sqrt{S_k^{(1)} S_k^{(2)}} &\leq K\sqrt{\frac{1}{K}\sum_{k=1}^K S_k^{(1)}}\sqrt{\frac{1}{K}\sum_{k=1}^K S_k^{(2)}} =\mathrm{o}_P(K \rho_n),
\end{align*}
where the inequalities follow by the Cauchy-Schwarz inequality and the last step is by \eqref{eq:linear_assmp_2}. Further, since $m_k \geq 2$ and $\rho_n \leq \beta_{1, n}\beta_{2, n}$, we can conclude 
\[
\sum_{k=1}^K \left(S_k^{(12)} + \dfrac{m_k - 1}{m_k}\beta_{1,n}\beta_{2,n}\right)= \left( \sum_{k=1}^K \frac{m_k - 1}{m_k} \beta_{1, n} \beta_{2, n} \right)\cdot (1+\mathrm{o}_P(1)).
\]
Next, by~\eqref{eq:bin_size_assmp} and~\eqref{eq:linear_assmp_1}, we have that
\begin{multline}\label{eq:neg_denom_1}
  \sum_{k=1}^K \frac{S_k^{(1)}}{m_k} \leq \frac{1}{\min_{k\in [K]} m_k} \sum_{k=1}^K S_k^{(1)} = \mathrm{O}_P\!\left( \frac{1}{\max_{k\in [K]} m_k} \sum_{k=1}^K S_k^{(1)}\right) \\
  = \mathrm{O}_P\!\left( \frac{K}{\max_{k\in [K]} m_k} \cdot \frac{1}{K} \sum_{k=1}^K S_k^{(1)}\right) = \mathrm{o}_P\!\left( \frac{K}{\max_{k\in [K]} m_k}\right) = \mathrm{o}_P \left( \sum_{k=1}^K \frac{1}{m_k} \right).
\end{multline}
Similarly, we obtain that
\[
  \sum_{k=1}^K \frac{S_k^{(2)}}{m_k} = \mathrm{o}_P \left( \sum_{k=1}^K \frac{1}{m_k} \right).
\]
Next, we have that by Cauchy--Schwarz inequality,
\begin{equation}\label{eq:neg_denom_2}
  \sum_{k=1}^K \frac{|S_k^{(12)}|}{m_k} \leq \sum_{k=1}^K \frac{\sqrt{S_k^{(1)} S_k^{(2)}}}{m_k} \leq \sqrt{\sum_{k=1}^K \frac{S_k^{(1)}}{m_k}} \cdot \sqrt{\sum_{k=1}^K \frac{S_k^{(2)}}{m_k}} = \mathrm{o}_P\!\left(\sum_{k=1}^K \frac{ 1}{m_k}\right).
\end{equation}
Finally, we have that
\[
\dfrac{(m_k - 1)(\beta_{1,n}^2 + 1)(\beta_{2,n}^2 + 1)(1 + \rho_n^2)}{m_k^2}\gtrsim \max\left\{\frac{(\beta_{1,n}^2 + 1)}{m_k},\frac{(\beta_{2,n}^2 + 1)}{m_k}\right\},
\]
which implies that 
\begin{multline*}
    \sum_{k=1}^K \left(\dfrac{(\beta_{2,n}^2 + 1)\,S_k^{(1)} + (\beta_{1,n}^2 + 1)\,S_k^{(2)} + 2\beta_{1,n}\beta_{2,n}\,S_k^{(12)}}{m_k} + \dfrac{(m_k - 1)(\beta_{1,n}^2 + 1)(\beta_{2,n}^2 + 1)(1 + \rho_n^2)}{m_k^2}\right)\\
    =\left(\sum_{k=1}^K \dfrac{(m_k - 1)(\beta_{1,n}^2 + 1)(\beta_{2,n}^2 + 1)(1 + \rho_n^2)}{m_k^2}\right)\cdot (1+\mathrm{o}_P(1)).
\end{multline*}
Combining all the parts, we have
\begin{align}\label{eq:SNR_LPT_final}
  \mathrm{SNR_{LPT}} &= \frac{\beta_{1,n}\beta_{2,n}\sum_{k=1}^K \frac{m_k - 1}{m_k}}{\sqrt{(\beta_{1,n}^2 + 1)(\beta_{2,n}^2 + 1)(1 + \rho_n^2)\sum_{k=1}^K \frac{m_k - 1}{m_k^2}}} \cdot \left( 1 + \mathrm{o}_P(1) \right)\notag \\  &=  
  \frac{\sum_{k=1}^K \frac{m_k - 1}{m_k}}{\sqrt{\sum_{k=1}^K \frac{m_k - 1}{m_k^2}}} \cdot \frac{\rho_n}{\sqrt{1 + \rho_n^2}} \cdot \left( 1 + \mathrm{o}_P(1) \right),
\end{align}
as desired.

\paragraph{Aggregate Lyapunov estimates and the power conclusion.}
We now verify, in order, the aggregate Lyapunov condition
(Assumption~\ref{assn:T_bdd_kurtosis}), variance comparability
(Assumption~\ref{assn:T_cond_var_nontrivial}), and the relevant branch of
Assumption~\ref{assn:T_local_alternative}. This will permit a direct
application of Theorem~\ref{cor:nnpt_power} in the bounded-signal regime.
The diverging-signal regime will be handled by a one-sided concentration
argument, for which no fourth-moment bound on the permutation statistic is
needed.

For compactness, define
\[
  V_n=\sum_{k=1}^K\Var(T_k\mid\mathbf Z),\qquad
  V_n^\pi=\sum_{k=1}^K\Var(\sigma_k(T_k)\mid\mathbf Z),
\]
and
\[
  M_{4,n}=\sum_{k=1}^K\E[(T_k-\E[T_k\mid\mathbf Z])^4\mid\mathbf Z],
  \qquad
  M_{4,n}^\pi=\sum_{k=1}^K\E[\sigma_k(T_k)^4\mid\mathbf Z].
\]
Let $\bar m=n/K$; condition~\eqref{eq:bin_size_assmp} implies that all $m_k$ are of common order $\bar m$.

For $i\in B_k$, put
\[
  a_i=X_i-\overline{\mathbf X}_k,\qquad b_i=Y_i-\overline{\mathbf Y}_k,
\]
and define
\[
  Q_{X,k}=\frac1{m_k}\sum_{i\in B_k}a_i^2,\quad
  Q_{Y,k}=\frac1{m_k}\sum_{i\in B_k}b_i^2,\quad
  R_{X,k}=\frac1{m_k}\sum_{i\in B_k}a_i^4,\quad
  R_{Y,k}=\frac1{m_k}\sum_{i\in B_k}b_i^4.
\]
Lemmas~\ref{lem:unif_perm_var} and~\ref{lem:unif_perm_fourth} give, conditionally on the data,
\begin{align}
  \E_\sigma[\sigma_k(T_k)^2\mid\mathbf X,\mathbf Y,\mathbf Z]
  &=\frac{Q_{X,k}Q_{Y,k}}{m_k-1},\label{eq:perm_second_Q}\\
  \E_\sigma[\sigma_k(T_k)^4\mid\mathbf X,\mathbf Y,\mathbf Z]
  &\lesssim \frac{Q_{X,k}^2Q_{Y,k}^2}{m_k^2}
  +\frac{R_{X,k}R_{Y,k}}{m_k^3}.\label{eq:perm_fourth_QR}
\end{align}
These identities are the starting point for checking the permutation half of
Assumption~\ref{assn:T_bdd_kurtosis}.
Evaluating the expectation in \eqref{eq:perm_second_Q} gives the exact identity
\begin{multline}\label{eq:var_perm}
  \Var(\sigma_k(T_k)\mid\mathbf Z)
  =\frac{S_k^{(1)}S_k^{(2)}}{m_k-1}
   +\frac{(\beta_{2,n}^2+1)S_k^{(1)}+(\beta_{1,n}^2+1)S_k^{(2)}}{m_k}\\
   +\frac{4\beta_{1,n}\beta_{2,n}S_k^{(12)}}{m_k(m_k-1)}
   +\frac{(m_k-1)(\beta_{1,n}^2+1)(\beta_{2,n}^2+1)}{m_k^2}
   +\frac{2\beta_{1,n}^2\beta_{2,n}^2}{m_k^2}.
\end{multline}
The Gaussian moment bounds used above, Cauchy--Schwarz, and the deterministic inequality
\[
  \frac1{m_k}\sum_{i\in B_k}
  \bigl(f_j(Z_i)-\overline{f_j(\mathbf Z)}_k\bigr)^4
  \leq m_k(S_k^{(j)})^2
\]
imply
\begin{align}
  \E[Q_{X,k}^2Q_{Y,k}^2\mid\mathbf Z]
  &\lesssim \bigl((1+S_k^{(1)})(1+S_k^{(2)})\bigr)^2,\label{eq:Q_bound}\\
  \E[R_{X,k}R_{Y,k}\mid\mathbf Z]
  &\lesssim \bigl(1+m_k(S_k^{(1)})^2\bigr)
                    \bigl(1+m_k(S_k^{(2)})^2\bigr).\label{eq:R_bound}
\end{align}
For \eqref{eq:R_bound}, expand each centered observation into its deterministic within-bin drift and centered Gaussian part, use $(x+y)^4\lesssim x^4+y^4$, and then apply the displayed deterministic inequality. This argument controls the mixed product directly and does not require a uniform bound on deterministic coordinate leverage.

The exact representation~\eqref{eq:perm_second_Q} also supplies a noise floor without discarding the possibly negative $S_k^{(12)}$ term in~\eqref{eq:var_perm}. Indeed, conditionally on the shared Gaussian signals $\{U_i:i\in B_k\}$, the two sample variances are independent over the noise variables and each has conditional expectation at least $(m_k-1)/m_k$. Consequently,
\begin{equation}\label{eq:perm_variance_floor}
  V_n^\pi=\sum_{k=1}^K\frac{\E[Q_{X,k}Q_{Y,k}\mid\mathbf Z]}{m_k-1}
  \gtrsim\sum_{k=1}^K\frac1{m_k}\asymp\frac K{\bar m}.
\end{equation}

Let
\[
  a_k=(1+S_k^{(1)})(1+S_k^{(2)}),
  \qquad c_k=S_k^{(1)}S_k^{(2)}.
\]
Combining \eqref{eq:perm_fourth_QR}--\eqref{eq:perm_variance_floor} and using comparable bin sizes yields
\begin{equation}\label{eq:aggregate_lyapunov_reduction}
  \frac{M_{4,n}^\pi}{(V_n^\pi)^2}
  \lesssim
  \frac{\sum_k a_k^2}{K^2}
  +\frac{\sum_k(S_k^{(1)})^2+\sum_k(S_k^{(2)})^2}{K^2}
  +\frac{\bar m\sum_kc_k^2}{K^2}
  +\frac1{K\bar m}.
\end{equation}
To read this bound in terms of the stated assumptions, note first that
$a_k^2\lesssim1+(S_k^{(1)})^2+(S_k^{(2)})^2+c_k^2$. By nonnegativity and
\eqref{eq:linear_assmp_1},
\[
  \frac{1}{K^2}\sum_k(S_k^{(j)})^2
  \leq\left(\frac1K\sum_kS_k^{(j)}\right)^2
  =\mathrm{o}_P(1),\qquad j=1,2.
\]
Because $K\to\infty$ and $\bar m\geq2$, every term on the right-hand side
of~\eqref{eq:aggregate_lyapunov_reduction}, except possibly the term involving
$\bar m\sum_kc_k^2/K^2$, is therefore either $\mathrm{o}_P(1)$ or bounded by
that remaining term. For that term,
\[
  \sum_kc_k=\mathrm{o}_P(K),\qquad
  \max_kc_k
  \leq\left(\sum_kS_k^{(1)}\right)\left(\sum_kS_k^{(2)}\right)
  =\mathrm{o}_P(K^2\rho_n^2),
\]
by \eqref{eq:linear_assmp_1} and \eqref{eq:linear_assmp_2}. Hence
\[
  \frac{\bar m\sum_kc_k^2}{K^2}
  \leq\frac{\bar m(\max_kc_k)\sum_kc_k}{K^2}
  =\mathrm{o}_P(K\bar m\rho_n^2).
\]
On every subsequence on which $\mathrm{SNR}_{\mathrm{LPT}}=\mathrm{O}_P(1)$, equation~\eqref{eq:SNR_LPT_final} gives $K\bar m\rho_n^2=\mathrm{O}_P(1)$, so \eqref{eq:aggregate_lyapunov_reduction} is $\mathrm{o}_P(1)$.
Thus
\[
  \frac{M_{4,n}^\pi}{(V_n^\pi)^2}=\mathrm{o}_P(1),
\]
which is exactly the permutation half of
Assumption~\ref{assn:T_bdd_kurtosis}.

For the observed half of that assumption, Gaussian hypercontractivity gives
\[
  \E[(T_k-\E[T_k\mid\mathbf Z])^4\mid\mathbf Z]
  \lesssim\Var(T_k\mid\mathbf Z)^2.
\]
Equation~\eqref{eq:var_nonperm}, comparable bin sizes, and \eqref{eq:linear_assmp_1} imply
\[
  \frac{\max_k\Var(T_k\mid\mathbf Z)}{V_n}=\mathrm{o}_P(1),
\]
and therefore $M_{4,n}/V_n^2=\mathrm{o}_P(1)$.
This verifies the observed half of Assumption~\ref{assn:T_bdd_kurtosis}.

We next check Assumption~\ref{assn:T_cond_var_nontrivial}. Summing
\eqref{eq:var_nonperm} and \eqref{eq:var_perm}, with the negligible terms
controlled by \eqref{eq:neg_denom_1} and \eqref{eq:neg_denom_2}, gives
\begin{align*}
  V_n
  &=(\beta_{1,n}^2+1)(\beta_{2,n}^2+1)(1+\rho_n^2)
    \sum_{k=1}^K\frac{m_k-1}{m_k^2}\,(1+\mathrm{o}_P(1)),\\
  V_n^\pi
  &=(\beta_{1,n}^2+1)(\beta_{2,n}^2+1)
    \left(\sum_{k=1}^K\frac{m_k-1}{m_k^2}
          +2\rho_n^2\sum_{k=1}^K\frac1{m_k^2}\right)(1+\mathrm{o}_P(1)).
\end{align*}
Thus $V_n^\pi/V_n=\mathrm{O}_P(1)$, which is precisely
Assumption~\ref{assn:T_cond_var_nontrivial}. Moreover, this ratio is
$1+\mathrm{o}_P(1)$ whenever $\rho_n\to0$.

It remains to check Assumption~\ref{assn:T_local_alternative} and conclude the
power calculation. First suppose
$\mathrm{SNR}_{\mathrm{LPT}}=\mathrm{O}_P(1)$. Equation
\eqref{eq:SNR_LPT_final} and $K\bar m\to\infty$ imply $\rho_n\to0$, so the
variance ratio above is $1+\mathrm{o}_P(1)$. This is the second branch of
Assumption~\ref{assn:T_local_alternative}. All assumptions of
Theorem~\ref{cor:nnpt_power} have now been verified, and that theorem gives
\[
  \mathbb E[\phi_{\mathrm{LPT}}\mid\mathbf Z]
  =\Phi\!\left(\Phi^{-1}(\alpha)+\mathrm{SNR}_{\mathrm{LPT}}\right)
   +\mathrm{o}_P(1).
\]

Finally, suppose $\mathrm{SNR}_{\mathrm{LPT}}\to\infty$, which is the first
branch of Assumption~\ref{assn:T_local_alternative}. The permutation Lyapunov
ratio need not be controlled in this regime, because power can be proved
directly. Write
$\mu_n=\sum_k\E[T_k\mid\mathbf Z]$. Chebyshev's inequality gives
$\sum_kT_k=\mu_n(1+\mathrm{o}_P(1))$. Conditional on the data, the permutation statistic is centered. On the event $\sum_kT_k>0$, its p-value therefore satisfies
\[
  p\leq
  \frac{\widehat V_n^\pi}{(\sum_kT_k)^2},
  \qquad
  \widehat V_n^\pi
  =\sum_{k=1}^K\E_\sigma[\sigma_k(T_k)^2\mid\mathbf X,\mathbf Y,\mathbf Z].
\]
Since $\E[\widehat V_n^\pi\mid\mathbf Z]=V_n^\pi$ and $V_n^\pi/V_n=\mathrm{O}_P(1)$, Markov's inequality gives $p\to0$ in probability. Hence power tends to one, which agrees with the asserted Gaussian-CDF expression as its argument tends to $+\infty$.

\qed

\section{Proofs from Appendix \ref{sec:additional} and other technical lemmas}\label{sec:app_tech}

\subsection{Proof of Theorem \ref{thm:validity_hellinger}}

We just need to prove a corresponding version of Lemmas \ref{lem:initial_tv_bound} and \ref{lem:bin_tv_bound} for the square Hellinger distance; these will be Lemmas \ref{lem:initial_h2_bound} and \ref{lem:bin_h2_bound} respectively.

\begin{lemma}\label{lem:initial_h2_bound}
  Under the notation of Lemmas \ref{lem:initial_tv_bound} and \ref{lem:bin_tv_bound}, the Type I error of adaptive-LPT, conditional on $\bZ$, is bounded as
  \[
  \mathbb P \left( p\leq \alpha \mid \mathbf Z \right) \leq \alpha + \left( 2\sum_{k=1}^K \mathrm d_{\mathrm H}^2(P_k(\bZ), P_k^*(\bZ)) \right)^{1/2}.
  \]
\end{lemma}

\begin{proof}
  Examining the final display in the proof of Lemma \ref{lem:initial_tv_bound}, using $\mathrm d_{\mathrm{TV}}(P,Q)\leq\sqrt{2}\,\mathrm d_{\mathrm H}(P,Q)$ under our normalization, gives
  \begin{align*}
    \P(p \leq \alpha \mid \mathbf Z)
    &\leq \alpha + \dtv \left( \bigotimes_{k=1}^K P_k(\mathbf Z), \bigotimes_{k=1}^K P_k^*(\mathbf Z) \right) \\
    &\leq \alpha + \sqrt{2}\,\mathrm d_{\mathrm H} \left( \bigotimes_{k=1}^K P_k(\mathbf Z), \bigotimes_{k=1}^K P_k^*(\mathbf Z) \right) \\
    &\leq \alpha + \left(2\sum_{k=1}^K \mathrm d_{\mathrm H}^2(P_k(\bZ), P^*_k(\bZ)) \right)^{1/2}.
  \end{align*}
\end{proof}

\begin{lemma}\label{lem:bin_h2_bound}
  Under the notation of Lemmas \ref{lem:initial_tv_bound} and \ref{lem:bin_tv_bound}, it holds that for each $k\in[K]$,
  \[
    \mathrm d_{\mathrm H}^2(P_k(\bZ), P_k^*(\bZ)) \leq 4(m_k-1)\cdot \left( \max_{i, j \in B_k} \mathrm d_{\mathrm H}^2(P_{X \mid Z_i}, P_{X \mid Z_j}) \right) \left( \max_{i, j \in B_k} \mathrm d_{\mathrm H}^2(P_{Y \mid Z_i}, P_{Y \mid Z_j}) \right).
  \]
\end{lemma}

\begin{proof}
  Fix a bin of size $m=m_k$. All laws and expectations below are conditional on $\mathbf Z$; expectations indexed by permutations are only over the displayed auxiliary permutation randomness.
  Put
  \[
    \Delta_X=\max_{i,j\in B_k}\mathrm d_{\mathrm H}^2(P_{X\mid Z_i},P_{X\mid Z_j}),
    \qquad
    \Delta_Y=\max_{i,j\in B_k}\mathrm d_{\mathrm H}^2(P_{Y\mid Z_i},P_{Y\mid Z_j}).
  \]
  For each relative permutation $\pi\in\mathcal S_m$, define
  \[
    R_\pi:=\mathcal L(X_\sigma,Y_{\sigma\circ\pi}\mid\mathbf Z),
    \qquad \sigma\sim\mathrm{Unif}(\mathcal S_m)
  \]
  where $\mathcal L(X_\sigma, Y_{\sigma \circ \pi} \mid \mathbf Z)$ notates the law of $(X_\sigma, Y_{\sigma \circ \pi})$ conditional on $\mathbf Z$, so that  $R_{\mathrm{Id}}=P_k(\mathbf Z)$ and $P_k^*(\mathbf Z)=\frac{1}{m!}\sum_{\pi}R_\pi$. Choose a permutation-invariant dominating measure $\mu$ (which may depend on the fixed $\mathbf Z$), write $f_\pi=\sqrt{\mathsf{d}R_\pi/\mathsf{d}\mu}\in L^2(\mu)$, and set $\bar f=\mathbb E_\pi[f_\pi]$. Notice that
  \[
    \frac{\mathsf{d}P_k^*(\mathbf Z)}{\mathsf{d}\mu}=\mathbb E_\pi[f_\pi^2],
    \qquad
    \sqrt{\frac{\mathsf{d}P_k^*(\mathbf Z)}{\mathsf{d}\mu}}
    =\sqrt{\mathbb E_\pi[f_\pi^2]},
  \]
  which gives
    \[
    \mathrm d_{\mathrm H}^2(P_k(\mathbf Z),P_k^*(\mathbf Z))
    = \mathrm d_{\mathrm H}^2(R_{\mathrm{Id}}, \mathbb E_\pi [R_\pi])
    \leq \mathbb E_\pi\left[
      \mathrm d_{\mathrm H}^2(R_{\mathrm{Id}},R_\pi)
    \right].
  \]
    Moreover, we have
  \[
    \mathrm d_{\mathrm H}^2(R_{\mathrm{Id}}, R_\pi) = \frac{1}{2} \left\| f_{\mathrm Id} - f_\pi \right\|_{L_2(\mu)}^2 
  \]
  so we can compute
  \begin{equation}\label{eq:dH_l2_eq}
   \mathbb E_\pi \left[ \mathrm d_{\mathrm H}^2(R_{\mathrm{Id}}, R_\pi) \right] = \frac{1}{2} \mathbb E_\pi \left[ \left\| f_{\mathrm Id} - f_\pi \right\|_{L_2(\mu)}^2 \right]
    =\frac{1}{2}\mathbb E_{\pi,\pi'}\left[\left\|f_\pi-f_{\pi'}\right\|_{L_2(\mu)}^2 \right]
    = \mathbb E_\pi\left[\left\|f_\pi-\bar f\right\|_{L_2(\mu)}^2\right]
  \end{equation}
  The last equality is a basic variance identity. The second equality holds because the dominating measure $\mu$ is symmetric under permutations of $Y$ by construction, so that performing a common permutation on the integrand leaves the $L_2(\mu)$ norm unchanged:
  \[
    \left\|f_\pi-f_{\pi'}\right\|_{L_2(\mu)}^2 = 
    \left\|f_{\pi \circ \pi^{-1}} -f_{\pi' \circ \pi^{-1}}\right\|_{L_2(\mu)}^2 =
    \left\|f_{\mathrm{Id}} - f_{\pi' \circ \pi^{-1}}\right\|_{L_2(\mu)}^2 .
  \]
  But since $\pi$ and $\pi'$ are independent and uniform on $\mathcal S_m$, $\pi' \circ \pi^{-1}$ is also uniform on $\mathcal S_m$, giving us the equality in expectation in the second equality.

  Combining this identity with the previous bound gives
  \[
    \mathrm d_{\mathrm H}^2(P_k(\mathbf Z),P_k^*(\mathbf Z))
    \leq\mathbb E_\pi\left[
      \left\|f_\pi-\bar f\right\|_{L_2(\mu)}^2
    \right].
  \]
  Interchanging the order of integration yields
  \[
    \mathbb E_\pi\left[
      \left\|f_\pi-\bar f\right\|_{L_2(\mu)}^2
    \right]
    = \mathbb E_\pi \left[ \int  (f_\pi(x, y) - \bar f(x, y))^2  \mathsf{d}\mu(x, y) \right]
    = \int \mathbb E_\pi \left[ (f_\pi(x, y) - \bar f(x, y))^2 \right] \mathsf{d}\mu(x, y).
  \]
  Now, Lemma~\ref{lem:permutation_variance}, which we defer as a technical lemma, establishes the following inequality controlling the inner expectation:
  \[
    \mathbb E_\pi\left[
      \left(f_\pi(w)-\bar f(w)\right)^2
    \right]
    \leq \frac{m-1}{4}\,
    \mathbb E_{\pi,\tau}\left[
      \left(f_\pi(w)-f_{\pi\circ\tau}(w)\right)^2
    \right].
  \]
  where $\tau$ is a uniformly random transposition. This inequality has two desiderata: both introducing the desired relation to the bin size $m$ and simplifying the expectation. After swapping back the order of integration,
  \[
    \mathbb E_\pi\left[\left\|f_\pi-\bar f\right\|_{L_2(\mu)}^2\right]
    \leq \frac{m-1}{4}\,
    \mathbb E_{\pi,\tau}\left[\left\|f_\pi-f_{\pi\circ\tau}\right\|_{L_2(\mu)}^2\right] =\frac{m-1}{2}\,
    \mathbb E_\tau\left[\mathrm d_{\mathrm H}^2(R_{\mathrm{Id}},R_\tau)\right],
  \]
  where the final equality uses $\|f_\pi-f_{\pi\circ\tau}\|_{L_2(\mu)}^2 =2\mathrm d_{\mathrm H}^2(R_\pi,R_{\pi\circ\tau})$ by the same argument used to establish \eqref{eq:dH_l2_eq}. Consequently,
  \[
    \mathrm d_{\mathrm H}^2(P_k(\mathbf Z),P_k^*(\mathbf Z))
    \leq \frac{m-1}{2}\,\mathbb E_\tau\left[\mathrm d_{\mathrm H}^2(R_{\mathrm{Id}},R_\tau)\right].
  \]

  It remains to bound the effect of one transposition, which proceeds similarly as in Lemma \ref{lem:dtv_transposition}. For $\eta\in\mathcal S_m$, write
  \[
    P^X_\eta:=\mathcal L(X_\eta\mid\mathbf Z),
    \qquad
    P^Y_\eta:=\mathcal L(Y_\eta\mid\mathbf Z).
  \]
  Fix a transposition $\tau$, let $\kappa$ be uniform on $\mathcal S_m$, and independently let $\nu\sim\tfrac12\delta_{\mathrm{Id}}+\tfrac12\delta_\tau$. The auxiliary variable $\nu$ simply chooses (as in Lemma \ref{lem:dtv_transposition}), with equal probability, between the two permutations $\kappa$ and $\kappa\circ\tau$. Since $\kappa\circ\nu$ is uniform on $\mathcal S_m$, averaging first over $\nu$ gives the exact representations
  \begin{align*}
    R_{\mathrm{Id}}
    &=\mathbb E_\kappa\left[
      \tfrac12\left(
        P^X_\kappa\times P^Y_\kappa
        +P^X_{\kappa\circ\tau}\times P^Y_{\kappa\circ\tau}
      \right)
    \right],\\
    R_\tau
    &=\mathbb E_\kappa\left[
      \tfrac12\left(
        P^X_\kappa\times P^Y_{\kappa\circ\tau}
        +P^X_{\kappa\circ\tau}\times P^Y_\kappa
      \right)
    \right].
  \end{align*}
  Joint convexity and Lemma~\ref{lem:dh2_product_swap} now apply to these full-vector laws and give
  \[
    \mathrm d_{\mathrm H}^2(R_{\mathrm{Id}},R_\tau)
    \leq 2\mathbb E_\kappa\left[
      \mathrm d_{\mathrm H}^2(P^X_\kappa,P^X_{\kappa\circ\tau})
      \mathrm d_{\mathrm H}^2(P^Y_\kappa,P^Y_{\kappa\circ\tau})
    \right].
  \]
  Suppose $\tau$ exchanges positions $a$ and $b$, and set
  $i=\kappa(a)$ and $j=\kappa(b)$. The two $X$-laws in the last
  display have identical factors outside positions $a$ and $b$; at
  those positions they have factors $P_{X\mid Z_i}\times P_{X\mid Z_j}$
  and $P_{X\mid Z_j}\times P_{X\mid Z_i}$. Multiplicativity of
  Hellinger affinity therefore gives
  \begin{align*}
    \mathrm d_{\mathrm H}^2(P^X_\kappa,P^X_{\kappa\circ\tau})
    &=1-\left(
      1-\mathrm d_{\mathrm H}^2(P_{X\mid Z_i},P_{X\mid Z_j})
    \right)^2\\
    &\leq 2\Delta_X.
  \end{align*}
  The same argument gives
  $\mathrm d_{\mathrm H}^2(P^Y_\kappa,P^Y_{\kappa\circ\tau})\leq2\Delta_Y$.
  Hence
  \[
    \mathrm d_{\mathrm H}^2(R_{\mathrm{Id}},R_\tau)
    \leq 8\Delta_X\Delta_Y.
  \]
  Combining this with the preceding random-transposition bound yields
  \[
    \mathrm d_{\mathrm H}^2(P_k(\mathbf Z),P_k^*(\mathbf Z))
    \leq4(m-1)\Delta_X\Delta_Y,
  \]
  as claimed. Together with Lemma~\ref{lem:initial_h2_bound} and
  $m_k-1\leq m_k$, this also proves Theorem~\ref{thm:validity_hellinger}.
\end{proof}

\begin{lemma}\label{lem:dh2_product_swap}
  Under the notation of Lemma \ref{lem:dtv_product_swap}, we have
  \[
    \mathrm d_{\mathrm H}^2\left(\frac{1}{2}(P\times Q + P'\times Q'), \frac{1}{2}(P\times Q' + P'\times Q)\right) \leq 2 \mathrm d_{\mathrm H}^2(P,P')\,\mathrm d_{\mathrm H}^2(Q,Q').
  \]
\end{lemma}

\begin{proof}
  Let the corresponding densities with respect to a dominating measure $\mu_P$ of $P, P'$ be $p, p'$, and similarly let densities of $Q, Q'$ with respect to a dominating measure $\mu_Q$ be denoted as $q, q'$. Then we compute:
  \begin{align*}
    & \mathrm d_{\mathrm H}^2 \left(\frac{1}{2}(P\times Q + P'\times Q'), \frac{1}{2}(P\times Q' + P'\times Q)\right) \\
    &= \frac{1}{2} \int \left( \sqrt{\frac{pq + p'q'}{2}} - \sqrt{\frac{pq' + p'q}{2}} \right)^2 (\mathsf{d}\mu_P \otimes \mathsf{d}\mu_Q) \\
    &= \frac{1}{4} \int \left( \frac{(pq + p'q') - (pq' + p'q)}{\sqrt{pq + p'q'} + \sqrt{pq' + p'q}} \right)^2 (\mathsf{d}\mu_P \otimes \mathsf{d}\mu_Q) \\
    &= \frac{1}{4} \int \left( \frac{(p - p')(q - q')}{\sqrt{pq + p'q'} + \sqrt{pq' + p'q}} \right)^2 (\mathsf{d}\mu_P \otimes \mathsf{d}\mu_Q) \\
    &\leq \frac{1}{2} \int \left( \frac{(p - p')(q - q')}{(\sqrt{p} + \sqrt{p'})(\sqrt{q} + \sqrt{q'})} \right)^2 (\mathsf{d}\mu_P \otimes \mathsf{d}\mu_Q) \\
    &= 2\left( \frac{1}{2} \int \left( \frac{p - p'}{\sqrt{p} + \sqrt{p'}}\right)^2 \mathsf{d}\mu_P \right)\left( \frac{1}{2} \int \left( \frac{q - q'}{\sqrt{q} + \sqrt{q'}}\right)^2 \mathsf{d}\mu_Q \right) \\
    &= 2\mathrm d_{\mathrm H}^2(P, P') \mathrm d_{\mathrm H}^2(Q, Q')
  \end{align*}
  The only inequality follows from the elementary bound $\sqrt{x^2 + y^2} \geq (x + y)/\sqrt{2}$ for $x, y \geq 0$, applied to each square root:
  \[
    \sqrt{pq + p'q'} + \sqrt{pq' + p'q} \geq \frac{1}{\sqrt{2}}\left( \sqrt{pq} + \sqrt{p'q'} + \sqrt{pq'} + \sqrt{p'q} \right) = \frac{1}{\sqrt{2}} (\sqrt{p} + \sqrt{p'})(\sqrt{q} + \sqrt{q'}),
  \]
  so that $(\sqrt{pq + p'q'} + \sqrt{pq' + p'q})^2 \geq \frac{1}{2}(\sqrt{p} + \sqrt{p'})^2 (\sqrt{q} + \sqrt{q'})^2$.
\end{proof}

\subsection{Proof of Theorem \ref{thm:power_nongaussian}}\label{sec:proof_nongaussian}

The proof of Theorem \ref{thm:power_lpt_confounder_model_general} used the
Gaussianity of $U, \epsilon_1, \epsilon_2$ only in two places:

\begin{itemize}
\item Computing $\mathrm{SNR_{LPT}}$.
\item Establishing the observed and permutation aggregate Lyapunov ratios in
  Assumption~\ref{assn:T_bdd_kurtosis}.
\end{itemize}

Both steps can instead be carried out under the stated moment conditions. The
computation of $\mathbb E[T_k \mid \mathbf Z]$ uses only mutual independence
and the first two moments of $U, \epsilon_1, \epsilon_2$, so it carries over
verbatim:
\[
  \EEst{T_k}{\mathbf Z} = S_k^{(12)} + \frac{m_k - 1}{m_k}\beta_{1,n}\beta_{2,n}.
\]
The conditional variance, on the other hand, picks up a dependence on the fourth moment $\mu_4 = \mathbb E[U^4]$; it is recomputed without Gaussianity in Lemma \ref{lem:VarT_k}.

\paragraph{Computing $\mathrm{SNR}_{\mathrm{LPT}}$.} Write
\[
  \nu_k := 1 + \rho_n^2(\mu_4 - 2) - \frac{\rho_n^2(\mu_4 - 3)}{m_k},
\]
  Since
  \[
    \beta_{1,n}^2\beta_{2,n}^2
    =\rho_n^2(\beta_{1,n}^2+1)(\beta_{2,n}^2+1),
  \]
  the second term in Lemma~\ref{lem:VarT_k} equals
  \[
    \frac{(m_k-1)(\beta_{1,n}^2+1)(\beta_{2,n}^2+1)}{m_k^2}\,\nu_k.
  \]
  Moreover, $\nu_k\asymp1$ uniformly in $k$ and $n$. Indeed, $\mu_4\geq\mathbb E[U^2]^2=1$ and $m_k\geq2$ imply $\mu_4-2-(\mu_4-3)/m_k\geq-1$, whereas $\mu_4<\infty$ gives a uniform upper bound.
  Define
  \[
    D_k^{\mathrm{NG}}:=
    \frac{(\beta_{2,n}^2+1)S_k^{(1)}+(\beta_{1,n}^2+1)S_k^{(2)}
      +2\beta_{1,n}\beta_{2,n}S_k^{(12)}}{m_k}
    +\frac{(m_k-1)(\beta_{1,n}^2+1)(\beta_{2,n}^2+1)\nu_k}{m_k^2}.
  \]
  We can now compute that
  \[
  \mathrm{SNR_{LPT}} = \frac{\sum_{k = 1}^K \mathbb E[T_k \mid \mathbf Z]}{ \left( \sum_{k=1}^K \Var(T_k \mid \mathbf Z) \right)^{1/2}} = 
  \frac{\sum_{k=1}^K \left(S_k^{(12)} + \dfrac{m_k - 1}{m_k}\beta_{1,n}\beta_{2,n}\right)}{\left(\sum_{k=1}^K D_k^{\mathrm{NG}}\right)^{1/2}}.
  \]

Now recall \eqref{eq:linear_assmp_1} and \eqref{eq:linear_assmp_2}. In the numerator, we have
\[
  \begin{aligned}
  \left|\sum_{k=1}^K S_k^{(12)}\right|
  &\leq K\sqrt{\frac{1}{K}\sum_{k=1}^K S_k^{(1)}}
             \sqrt{\frac{1}{K}\sum_{k=1}^K S_k^{(2)}}\\
  &=\mathrm{o}_P(K\rho_n)
   =\mathrm{o}_P\!\left(\sum_{k=1}^K\frac{m_k-1}{m_k}\beta_{1,n}\beta_{2,n}\right).
  \end{aligned}
\]
where the first two inequalities are from Cauchy-Schwarz and the first equality is by \eqref{eq:linear_assmp_2}. The last equality follows as $m_k \geq 2$ and $\rho_n \leq \beta_{1, n}\beta_{2, n}$. In the denominator, we have
by that
\begin{multline}\label{eq:neg_denom_1_general}
  \sum_{k=1}^K \frac{S_k^{(1)}}{m_k} \leq \frac{1}{\min_{k=1, \dots, K} m_k} \sum_{k=1}^K S_k^{(1)} \overset{\eqref{eq:bin_size_assmp}}= \mathrm{O}_P\!\left( \frac{1}{\max_{k=1, \dots, K} m_k} \sum_{k=1}^K S_k^{(1)}\right) \\
  = \mathrm{O}_P\!\left( \frac{K}{\max_{k=1, \dots, K} m_k} \cdot \frac{1}{K} \sum_{k=1}^K S_k^{(1)}\right) \overset{\eqref{eq:linear_assmp_1}}= \mathrm{o}_P\!\left( \frac{K}{\max_{k=1, \dots, K} m_k}\right) = \mathrm{o}_P \left( \sum_{k=1}^K \frac{1}{m_k} \right).
\end{multline}
Similarly, we get
\[
  \sum_{k=1}^K \frac{S_k^{(2)}}{m_k} = \mathrm{o}_P \left( \sum_{k=1}^K \frac{1}{m_k} \right).
\]
Finally,
\begin{equation}\label{eq:neg_denom_2_general}
  \sum_{k=1}^K \frac{|S_k^{(12)}|}{m_k} \leq \sum_{k=1}^K \frac{\sqrt{S_k^{(1)} S_k^{(2)}}}{m_k} \leq \sqrt{\sum_{k=1}^K \frac{S_k^{(1)}}{m_k}} \cdot \sqrt{\sum_{k=1}^K \frac{S_k^{(2)}}{m_k}} = \mathrm{o}_P\!\left(\sum_{k=1}^K \frac{ 1}{m_k}\right)
\end{equation}
so, using $\nu_k \asymp 1$ to absorb the negligible terms, we have
\begin{multline}\label{eq:SNR_LPT_final_general}
  \mathrm{SNR_{LPT}} = \frac{\beta_{1,n}\beta_{2,n}\sum_{k=1}^K \frac{m_k - 1}{m_k}}{\sqrt{(\beta_{1,n}^2 + 1)(\beta_{2,n}^2 + 1)\sum_{k=1}^K \frac{m_k - 1}{m_k^2}\,\nu_k}} \cdot \left( 1 + \mathrm{o}_P(1) \right) =  \\
  \frac{\sum_{k=1}^K \frac{m_k - 1}{m_k}\,\rho_n}{\sqrt{\sum_{k=1}^K \frac{m_k - 1}{m_k^2} \left(1 + \rho_n^2(\mu_4 - 2) - \frac{\rho_n^2(\mu_4 - 3)}{m_k}\right)}} \cdot \left( 1 + \mathrm{o}_P(1) \right)
\end{multline}
as desired.

\paragraph{Checking assumptions and concluding power.}
As in the Gaussian proof, we verify Assumptions~\ref{assn:T_bdd_kurtosis}--
\ref{assn:T_local_alternative} explicitly. The required distributional input
for the two aggregate Lyapunov ratios is the mixed moment
\begin{equation}\label{eq:nongaussian_mixed_moment}
  \mathbb E\!\left[|\beta_{1,n}U+\epsilon_1|^4
                         |\beta_{2,n}U+\epsilon_2|^4\right]<\infty
\end{equation}
uniformly in $n$. This follows directly from $\mathbb E[U^8]<\infty$ and
$\mathbb E[\epsilon_1^4],\mathbb E[\epsilon_2^4]<\infty$ by simply expanding the two fourth powers. In particular, no separate eighth moment of either noise variable is asserted or needed.

As before, write $\bar m=n/K$; condition~\eqref{eq:bin_size_assmp} makes every $m_k$ comparable to $\bar m$.

Define $Q_{X,k},Q_{Y,k},R_{X,k},R_{Y,k}$ as in the Gaussian proof. Lemmas~\ref{lem:unif_perm_var} and~\ref{lem:unif_perm_fourth} again give
\[
  \E_\sigma[\sigma_k(T_k)^2\mid\mathbf X,\mathbf Y,\mathbf Z]
  =\frac{Q_{X,k}Q_{Y,k}}{m_k-1}
\]
and
\[
  \E_\sigma[\sigma_k(T_k)^4\mid\mathbf X,\mathbf Y,\mathbf Z]
  \lesssim\frac{Q_{X,k}^2Q_{Y,k}^2}{m_k^2}
          +\frac{R_{X,k}R_{Y,k}}{m_k^3}.
\]
Using \eqref{eq:nongaussian_mixed_moment} to control the mixed random terms directly gives the same coarse estimates
\begin{align*}
  \E[Q_{X,k}^2Q_{Y,k}^2\mid\mathbf Z]
  &\lesssim\bigl((1+S_k^{(1)})(1+S_k^{(2)})\bigr)^2,\\
  \E[R_{X,k}R_{Y,k}\mid\mathbf Z]
  &\lesssim\bigl(1+m_k(S_k^{(1)})^2\bigr)
             \bigl(1+m_k(S_k^{(2)})^2\bigr).
\end{align*}
The second estimate uses
$m_k^{-1}\sum_i(f_j(Z_i)-\overline{f_j(\mathbf Z)}_k)^4
 \leq m_k(S_k^{(j)})^2$ for the deterministic profiles. Thus the aggregate
reduction~\eqref{eq:aggregate_lyapunov_reduction} remains valid. On
subsequences where $\mathrm{SNR}_{\mathrm{LPT}}=\mathrm{O}_P(1)$,
equation~\eqref{eq:SNR_LPT_final_general} again implies
$K\bar m\rho_n^2=\mathrm{O}_P(1)$, and the argument following
\eqref{eq:aggregate_lyapunov_reduction} yields
\[
  \frac{M_{4,n}^\pi}{(V_n^\pi)^2}=\mathrm{o}_P(1).
\]
This verifies the permutation half of
Assumption~\ref{assn:T_bdd_kurtosis}.

For the observed half of the same assumption, use the decomposition
\[
  T_k-\E[T_k\mid\mathbf Z]
  =A_{k,1}+A_{k,2}+
    \bigl(A_{k,3}-\E[A_{k,3}\mid\mathbf Z]\bigr),
\]
Rosenthal's inequality applied to the two linear terms and to
$m_k^{-1}\sum_i(W_iV_i-\E[W_iV_i])$, together with
\eqref{eq:nongaussian_mixed_moment}, gives
\[
  \E[(T_k-\E[T_k\mid\mathbf Z])^4\mid\mathbf Z]
  \lesssim \frac{1+(S_k^{(1)})^2+(S_k^{(2)})^2}{m_k^2}.
\]
Lemma~\ref{lem:VarT_k}, $\sup_n|\rho_n|<1$, and comparable bin sizes give
$V_n\gtrsim K/\bar m$. Hence
\[
  \frac{M_{4,n}}{V_n^2}
  \lesssim\frac{K+\sum_k(S_k^{(1)})^2+\sum_k(S_k^{(2)})^2}{K^2}
  =\mathrm{o}_P(1).
\]
Thus the observed half of Assumption~\ref{assn:T_bdd_kurtosis} also holds.
The exact second-moment representation provides the same permutation noise floor as in~\eqref{eq:perm_variance_floor}. Direct expansion using only second and fourth moments further gives
\[
  \frac{V_n^\pi}{V_n}=\mathrm{O}_P(1),
  \qquad
  \frac{V_n^\pi}{V_n}=1+\mathrm{o}_P(1)\quad\text{if }\rho_n\to0.
\]
The first relation is Assumption~\ref{assn:T_cond_var_nontrivial}. In the
bounded-signal regime, \eqref{eq:SNR_LPT_final_general} and
$K\bar m\to\infty$ imply $\rho_n\to0$, so the second relation verifies the
variance-ratio branch of Assumption~\ref{assn:T_local_alternative}.

All assumptions of Theorem~\ref{cor:nnpt_power} have now been checked on
bounded-signal subsequences. Applying that theorem gives
\[
  \mathbb E[\phi_{\mathrm{LPT}}\mid\mathbf Z]
  =\Phi\!\left(\Phi^{-1}(\alpha)+\mathrm{SNR}_{\mathrm{LPT}}\right)
   +\mathrm{o}_P(1).
\]
If instead $\mathrm{SNR}_{\mathrm{LPT}}\to\infty$, the first branch of
Assumption~\ref{assn:T_local_alternative} holds, but the direct
Chebyshev--Markov argument from the Gaussian proof avoids any need to verify
the permutation Lyapunov ratio: $\sum_kT_k=\mu_n(1+\mathrm{o}_P(1))$, while
the conditional permutation p-value is at most
$\widehat V_n^\pi/(\sum_kT_k)^2$ on the event that the observed statistic is positive. Since
$\E[\widehat V_n^\pi\mid\mathbf Z]=V_n^\pi$ and $V_n^\pi/V_n=\mathrm{O}_P(1)$, the p-value converges to zero in probability. This proves the stated power conclusion without invoking permutation fourth moments in the strong-signal regime. \qed

\subsection{Proof of Lemma \ref{lem:binary_power_lb_tv}}

For this proof, we use the following notation: in a single bin $B_k = \{i, j\}$, suppose that $P_{X, Y \mid Z_i}$ and $P_{X, Y \mid Z_j}$ have probability densities $f_{i}(x, y)$ and $f_{j}(x, y)$ respectively; similarly let $P_{X \mid Z_i} \times P_{Y \mid Z_i}$ and $P_{X \mid Z_j} \times P_{Y \mid Z_j}$ have densities $g_i(x, y)$ and $g_j(x, y)$ respectively. As in Section \ref{sec:binary_power}, we ignore ties throughout. Now define the region
\[
  E_k = \{(x_i, x_j, y_i, y_j) \mid f_i(x_i, y_i)f_i(x_j, y_j) \geq f_i(x_i, y_j)f_i(x_j, y_i) \} \subset \mathcal X^2 \times \mathcal Y^2
\]
and set
\[
  T_k(X_i, X_j, Y_i, Y_j) = 2 \cdot \mathbbm 1 \{(X_i, X_j, Y_i, Y_j) \in E_k\} - 1;
\]
this is the bet that would be optimal if both points in the bin shared the conditional distribution at $Z_i$, and it is available to a bettor with access to the true likelihood ratios.
Now consider the joint densities
\[
  r(x_i, x_j, y_i, y_j) = f_i(x_i, y_i)f_i(x_j, y_j)
\]
and
\[
  s(x_i, x_j, y_i, y_j) = f_i(x_i, y_j)f_i(x_j, y_i).
\]
Then,
\[
  \mathrm{d_{TV}}(r, s) = \int \mathbbm 1 \{E_k\} (r - s) = \int \mathbbm 1 \{E_k\} r - \int \mathbbm 1\{E_k\} s.
\]
However, we see that
\begin{align*}
  \int \mathbbm 1\{E_k\}s &= \int \mathbbm 1\{(x_i, x_j, y_i, y_j) \in E_k\} s(x_i, x_j, y_i, y_j) \mathsf{d}x_i \mathsf{d}x_j \mathsf{d}y_i  \mathsf{d}y_j \\
  \intertext{and by just changing the order of integration and renaming the variables,}
                          &= \int \mathbbm 1\{(x_i, x_j, y_j, y_i) \in E_k\} s(x_i, x_j, y_j, y_i) \mathsf{d}x_i \mathsf{d}x_j \mathsf{d}y_i \mathsf{d}y_j \\
  \intertext{and since, ignoring ties, $(x_i, x_j, y_i, y_j) \in E_k \iff (x_i, x_j, y_j, y_i) \notin E_k$ by construction of $E_k$,}
                          &= \int (1 - \mathbbm 1\{(x_i, x_j, y_i, y_j) \in E_k\}) s(x_i, x_j, y_j, y_i) \mathsf{d}x_i \mathsf{d}x_j \mathsf{d}y_i \mathsf{d}y_j \\
  \intertext{and simply by substituting the definition of $s$, we get}
                          &= \int (1 - \mathbbm 1\{(x_i, x_j, y_i, y_j) \in E_k\}) r(x_i, x_j, y_i, y_j) \mathsf{d}x_i \mathsf{d}x_j \mathsf{d}y_i \mathsf{d}y_j \\
                          &= 1 - \int \mathbbm 1\{E_k\}r.
\end{align*}
Thus we may conclude that
\begin{equation}\label{eq:dtv_relation}
  \dtv(r, s) = 2 \int \One{E_k} r - 1.
\end{equation}
To control this total variation, notice that marginalizing $r$ and $s$ over the coordinates $x_j, y_j$ cannot increase the total variation, by the data processing inequality for total variation; but marginalizing over those coordinates gives
\[
  \int r(x_i, x_j, y_i, y_j) \mathsf{d}x_j \mathsf{d}y_j = f_i(x_i, y_i)
\]
and
\begin{align*}
  \int s(x_i, x_j, y_i, y_j) \mathsf{d}x_j\mathsf{d}y_j &= \int f_i(x_i, y_j)f_i(x_j, y_i)\mathsf{d}x_j \mathsf{d}y_j \\
                                      &= \left( \int f_i(x_i, y_j)\mathsf{d}y_j \right) \left( \int f_i(x_j, y_i)\mathsf{d}x_j \right) \\
                                      &=  g_i(x_i, y_i),
\end{align*}
so
\begin{equation}\label{eq:dtv_bound_lower}
  \dtv(r, s) \geq \dtv(P_{X, Y \mid Z_i}, P_{X \mid Z_i} \times P_{Y \mid Z_i}).
\end{equation}
Combining \eqref{eq:dtv_relation} and \eqref{eq:dtv_bound_lower} gives
\[
  \int \One{E_k}r - \frac{1}{2} \geq \frac{1}{2}\dtv(P_{X, Y \mid Z_i}, P_{X \mid Z_i} \times P_{Y \mid Z_i}).
\]
On the other hand, $\int \One{E_k}r$ is the probability that $T_k = 1$ when the data is drawn from $P_{X, Y \mid Z_i} \otimes P_{X, Y \mid Z_i}$ rather than its true distribution $P_{X, Y \mid Z_i} \otimes P_{X, Y \mid Z_j}$. Since the total variation distance between two product measures sharing a common factor is at most the total variation distance between the differing factors, we have that
\[
  \left|\int \One{E_k} r - p_k \right| \leq \dtv\left(P_{X, Y \mid Z_i} \otimes P_{X, Y \mid Z_i},\ P_{X, Y \mid Z_i} \otimes P_{X, Y \mid Z_j}\right) \leq \dtv(P_{X, Y \mid Z_i}, P_{X, Y \mid Z_j}) = \delta_k
\]
so the above bound gives
\[
  p_k - \frac{1}{2} \geq \frac{1}{2}\dtv(P_{X, Y \mid Z_i}, P_{X \mid Z_i} \times P_{Y \mid Z_i}) - \delta_k.
\]
An identical argument with the roles of $i$ and $j$ exchanged (i.e., betting with the region built from $f_j$ instead of $f_i$) also gives
\[
  p_k - \frac{1}{2} \geq \frac{1}{2}\dtv(P_{X, Y \mid Z_j}, P_{X \mid Z_j} \times P_{Y \mid Z_j}) - \delta_k.
\]
The two displays bound the success probabilities of two different bets, each available with access to the true likelihood ratios; since the optimal strategy guesses correctly with probability no smaller than either, $p_k$ satisfies both bounds, and taking the maximum gives the desired result. \qed

\subsection{Proof of Lemma \ref{lem:binary_power_lb_orc}}

The key technical ingredient is Lemma \ref{lem:dtv_kl_comparison} below, which lower bounds the total variation separation appearing in Lemma \ref{lem:binary_power_lb_tv} by the KL-type separation that drives the oracle test.

Set $C = \sqrt{2 \log K}$. Using $p_k(1 - p_k) \leq \frac14$, Lemma \ref{lem:binary_power_lb_tv}, and $\max\{u, v\} \geq \frac{u + v}{2}$,
\[
  \mathrm{SNR_{LPT}} \geq \frac{2}{\sqrt K}\sum_{k=1}^K \left( p_k - \frac12 \right) \geq \frac{1}{2\sqrt K}\sum_{i=1}^n \dtv(P_{X, Y \mid Z_i}, P_{X \mid Z_i} \times P_{Y \mid Z_i}) - \frac{2}{\sqrt K}\sum_{k=1}^K \delta_k.
\]
Next, we claim that $\overline{\mathrm{KL}}_{(i)} \leq \sigma^2$ for every $i$. Indeed, writing $\mu_{0,i}$ and $\mu_{1,i}$ for the means of $\mathrm{LLR}^{(i)}(X, Y)$ under the null and the alternative respectively, since $e^{\mathrm{LLR}^{(i)}}$ is the density of $P_{X, Y \mid Z_i}$ with respect to $P_{X \mid Z_i} \times P_{Y \mid Z_i}$, we have
\[
  1 = \E_{P_{X \mid Z_i} \times P_{Y \mid Z_i}}\left[e^{\mathrm{LLR}^{(i)}(X, Y)}\right] \leq e^{\mu_{0,i} + \sigma^2/2}
  \quad \text{and} \quad
  1 = \E_{P_{X, Y \mid Z_i}}\left[e^{-\mathrm{LLR}^{(i)}(X, Y)}\right] \leq e^{-\mu_{1,i} + \sigma^2/2},
\]
applying the sub-Gaussian assumption with $t = 1$ and $t = -1$ respectively; rearranging gives $\mu_{0,i} \geq -\sigma^2/2$ and $\mu_{1,i} \leq \sigma^2/2$, so $\overline{\mathrm{KL}}_{(i)} = \mu_{1,i} - \mu_{0,i} \leq \sigma^2$. Since $C = \sqrt{2 \log K} \to \infty$, for all $K \geq e^{\sigma^2/8}$ we have $\sigma \leq 2C$ and hence $\overline{\mathrm{KL}}_{(i)} \leq 2C\sigma$, so the minimum in the second display of Lemma \ref{lem:dtv_kl_comparison} is attained by its first argument, and combining with $\sigma^2 \leq M\,\mathrm{V}_{\mathrm{KL}, (i)}$,
\[
  \dtv(P_{X, Y \mid Z_i}, P_{X \mid Z_i} \times P_{Y \mid Z_i}) \geq \frac{\overline{\mathrm{KL}}_{(i)}}{4C\sqrt{M\, \mathrm{V}_{\mathrm{KL}, (i)}}} - \frac{e^{-C^2/2}}{C^2}.
\]
Note also that the assumption $\sigma^2 \leq M\, \mathrm{V}_{\mathrm{KL}, (i)}$ for all $i$ automatically makes the varentropies comparable:
\[
  \max_i \mathrm{V}_{\mathrm{KL}, (i)} \leq \sigma^2 \leq M \min_i \mathrm{V}_{\mathrm{KL}, (i)} \leq \frac{M}{n} \sum_{i=1}^n \mathrm{V}_{\mathrm{KL}, (i)},
\]
where the first inequality holds since the variance is always bounded by the sub-Gaussian parameter. Summing the total variation bound over $i \in [n]$ and using $\sqrt{\mathrm{V}_{\mathrm{KL}, (i)}} \leq \max_i \sqrt{\mathrm{V}_{\mathrm{KL}, (i)}} \leq \sqrt{M \sum_{i=1}^n \mathrm{V}_{\mathrm{KL}, (i)} / n}$,
\[
  \sum_{i=1}^n \dtv(P_{X, Y \mid Z_i}, P_{X \mid Z_i} \times P_{Y \mid Z_i}) \geq \frac{\sqrt{n}}{4CM} \cdot \mathrm{SNR_{ORC}} - \frac{n\, e^{-C^2/2}}{C^2},
\]
and substituting into the first display (recalling $n = 2K$) gives
\[
  \mathrm{SNR_{LPT}} \geq \frac{\mathrm{SNR_{ORC}}}{4\sqrt 2\, C M} - \frac{\sqrt K\, e^{-C^2/2}}{C^2} - \frac{2 \sum_{k=1}^K \delta_k}{\sqrt K}.
\]
Finally, with $C = \sqrt{2 \log K}$ we have $e^{-C^2/2} = 1/K$, so the middle term equals $\frac{1}{2\sqrt K \log K} = \mathrm o(1)$, while $4\sqrt 2\, C = 8 \sqrt{\log K}$, which yields
\[
  \mathrm{SNR_{LPT}} \geq \frac{\mathrm{SNR_{ORC}}}{8 M \sqrt{\log K}} - \frac{2 \sum_{k=1}^K \delta_k}{\sqrt K} - \mathrm o(1)
\]
as desired. \qed

Relative to Theorem \ref{thm:orc_vs_nnpt_oracle}, the comparison carries an additional $\sqrt{\log K}$ factor: this is the worst-case price of converting the KL-type separation exploited by the oracle into the total variation separation to which binary bets are sensitive. It disappears in settings where $\dtv \asymp \overline{\mathrm{KL}}/\sqrt{\mathrm V_{\mathrm{KL}}}$ holds directly, as in the Gaussian linear model of Section \ref{sec:linear_model}.

\subsection{Proof of Theorem \ref{thm:gaussian_binary_power}}\label{sec:prf_gaussian_binary_power}

Let $K = \lfloor n / 2 \rfloor$ denote the number of bins, and assume for notational convenience that $n = 2K$ is even. Fix a bin $B_k = \{i, j\}$ and write $D_1 = X_i - X_j$ and $D_2 = Y_i - Y_j$, so that $T_k = 2 \cdot \One{D_1 D_2 \geq 0} - 1$.

\paragraph{Conditional distribution of $(D_1, D_2)$.} Under the model \eqref{model:linear_gaussian_confounder}, conditional on $\mathbf Z$, the pair $(D_1, D_2)$ is bivariate Gaussian with means
\[
  \Delta_{1, k} = f_1(Z_i) - f_1(Z_j), \qquad \Delta_{2, k} = f_2(Z_i) - f_2(Z_j),
\]
variances $\sigma_1^2 = 2(\beta_{1, n}^2 + 1)$ and $\sigma_2^2 = 2(\beta_{2, n}^2 + 1)$, and correlation
\[
  \frac{2\beta_{1,n}\beta_{2,n}}{\sqrt{2(\beta_{1,n}^2+1) \cdot 2(\beta_{2,n}^2+1)}} = \rho_n.
\]
Since $m_k = 2$, a direct computation gives $\Delta_{l, k}^2 = 4 S_k^{(l)}$ for $l \in \{1, 2\}$; we write $\delta_{l, k} = \Delta_{l, k}/\sigma_l$ for the standardized drifts, so that $\delta_{l, k}^2 = 2 S_k^{(l)} / (\beta_{l, n}^2 + 1) \leq 2 S_k^{(l)}$. Note also that since $\limsup_n |\beta_{1, n}|$ and $\limsup_n |\beta_{2, n}| < \infty$, we have $\bar\rho := \sup_n |\rho_n| < 1$.

\paragraph{Idealized $\mathrm{SNR_{\mathrm{LPT}}}$.}

First we compute $\mathrm{SNR_{LPT}}$ in the idealized setting where $\Delta_{1, k} = \Delta_{2, k} = 0$, with the aim of eventually showing that the non-idealized case is asymptotically equivalent. In this case, we have that by Sheppard's classical quadrant-probability formula (see \cite{sheppard1899application}), 
\[
  \P\bigl( D_1 D_2 > 0 \mid \mathbf Z \bigr)
      = 2 \left( \frac{1}{4} + \frac{\arcsin \rho_n}{2\pi} \right)
      = \frac{1}{2} + \frac{\arcsin \rho_n}{\pi}
      = \frac{1}{2} + q_n
\]
exactly. Plugging into \eqref{eq:binary_SNR_formula}, we get that in this setting, 
\[
  \mathrm{SNR_{LPT}}
  = \frac{\sum_{k=1}^K (p_k - 1/2)}{\left(\sum_{k=1}^K p_k(1 - p_k)\right)^{1/2}}
  = \frac{K q_n}{\sqrt{K \left( \frac{1}{4} - q_n^2 \right)}}
  = \frac{\sqrt{2n}\, q_n}{\sqrt{1 - 4 q_n^2}}
\]
as expected.

\paragraph{Deviation from the idealized case.}

Write $p_k = \P \left( D_1 D_2 > 0 \mid \mathbf Z \right)$ and define
\[
  R_k := \E[T_k \mid \mathbf Z] - 2 q_n = 2 \left( p_k - \frac{1}{2} - q_n \right),
\]
the per-bin deviation from the idealized value computed above. This quantity satisfies
\[
  |R_k|
  \leq \frac{2}{\pi\sqrt{1-\bar\rho^2}}
  \left( |\delta_{1,k}\,\delta_{2,k}|
    + |\rho_n| \bigl( \delta_{1,k}^2 + \delta_{2,k}^2 \bigr) \right)
  \leq \frac{4}{\pi\sqrt{1-\bar\rho^2}}
  \left( \sqrt{S_k^{(1)} S_k^{(2)}}
    + |\rho_n| \bigl( S_k^{(1)} + S_k^{(2)} \bigr) \right).
\]

The key estimate is that $R_k$ is \emph{quadratic}, not linear, in the standardized drifts $\delta_{l,k}$: a naive bound (based on the fact that the drift flips the sign of coordinate $l$ with probability of order $|\delta_{l,k}|$) would give only $|R_k| = \mathrm{O}(|\delta_{1,k}| + |\delta_{2,k}|)$, which is too weak to be dominated by assumption~\eqref{eq:linear_assmp_2}. The gain comes from the fact that, on the event that the drift flips one coordinate's sign, the sign of the other coordinate is nearly unbiased. We separate the proof of the above bound to Lemma \ref{lem:binary_drift_perturbation}.

\paragraph{Averaging over the bins.} By Cauchy--Schwarz across bins and then
assumptions~\eqref{eq:linear_assmp_2} and~\eqref{eq:linear_assmp_1},
\begin{multline*}
  \frac{1}{K} \sum_{k=1}^K |R_k|
  \leq \frac{4}{\pi\sqrt{1-\bar\rho^2}}
    \left( \left( \frac{1}{K}\sum_{k=1}^K S_k^{(1)} \right)^{1/2}
           \left( \frac{1}{K}\sum_{k=1}^K S_k^{(2)} \right)^{1/2}
           + \rho_n \cdot \frac{1}{K} \sum_{k=1}^K
             \bigl( S_k^{(1)} + S_k^{(2)} \bigr) \right) \\
  = \mathrm{o}_P(\rho_n) + \rho_n \cdot \mathrm{o}_P(1),
\end{multline*}
and since $\rho_n \leq \arcsin(\rho_n) = \pi q_n$ for $\rho_n \in [0,1]$, we
conclude
\begin{equation}\label{eq:Rk_average}
  \frac{1}{K} \sum_{k=1}^K |R_k| = \mathrm{o}_P(q_n).
\end{equation}

\paragraph{Combining.} Since $p_k - \frac{1}{2} = q_n + \frac{R_k}{2}$,
\[
  \sum_{k=1}^K \left( p_k - \frac{1}{2} \right)
  = K q_n + \frac{1}{2} \sum_{k=1}^K R_k
  = K q_n \left( 1 + \mathrm{o}_P(1) \right),
\]
where the second equality holds because, by~\eqref{eq:Rk_average}, $\bigl| \frac{1}{2} \sum_{k=1}^K R_k \bigr| \leq \frac{K}{2} \cdot \mathrm{o}_P(q_n) = K q_n \cdot \mathrm{o}_P(1)$. Next, using
$p_k (1 - p_k) = \frac{1}{4} - \bigl( p_k - \frac{1}{2} \bigr)^2$ and
\[
  \left| \left( p_k - \frac{1}{2} \right)^2 - q_n^2 \right|
  = \left| p_k - \frac{1}{2} - q_n \right|
    \cdot \left| p_k - \frac{1}{2} + q_n \right|
  \leq \frac{|R_k|}{2} \cdot 1
  \leq |R_k|,
\]
where the middle factor is at most $1$ since $\bigl| p_k - \frac{1}{2} \bigr| \leq \frac{1}{2}$ and $|q_n| \leq \frac{1}{2}$, we obtain
\[
  \left| \sum_{k=1}^K p_k (1 - p_k) - K \left( \frac{1}{4} - q_n^2 \right)
  \right|
  \leq \sum_{k=1}^K |R_k|
  = K \cdot \mathrm{o}_P(q_n)
  = K \cdot \mathrm{o}_P(1).
\]
Since $\frac{1}{4} - q_n^2 \geq \frac{1}{4} - \bigl( \frac{\arcsin \bar\rho}{\pi} \bigr)^2 > 0$ is bounded away from zero (as $\bar\rho < 1$), this yields
\[
  \sum_{k=1}^K p_k (1 - p_k)
  = K \left( \frac{1}{4} - q_n^2 \right) \left( 1 + \mathrm{o}_P(1) \right).
\]

Substituting the two above computations into \eqref{eq:binary_SNR_formula}, we get
\[
  \mathrm{SNR_{LPT}}
  = \frac{K q_n}{\sqrt{K \left( \frac{1}{4} - q_n^2 \right)}}
    \cdot \left( 1 + \mathrm{o}_P(1) \right)
  = \frac{2 \sqrt{K}\, q_n}{\sqrt{1 - 4 q_n^2}}
    \cdot \left( 1 + \mathrm{o}_P(1) \right)
  = \frac{\sqrt{2n}\, q_n}{\sqrt{1 - 4 q_n^2}}
    \cdot \left( 1 + \mathrm{o}_P(1) \right),
\]
the final equality using $n = 2K$. For $n$ odd, the
single unbinned point is invariant under within-bin permutation and may be
dropped, so the argument applies with $K = (n-1)/2$, and
$2\sqrt{K} = \sqrt{2n} \bigl( 1 + \mathrm{o}(1) \bigr)$ is absorbed into the
$1 + \mathrm{o}_P(1)$ factor. \qed

\subsection{Technical lemmas}

\begin{lemma}\label{lem:normal_ks_dw}
    Let $A,B\in \mathbb R$ be random variables.
    Let $\mu\in \mathbb R$ and $\sigma^2>0$. Then
    \[\dks(A,B) \leq 2\dks(B,\mathcal{N}(\mu,\sigma^2)) + 2\sqrt{\frac{\dw(A,B)}{\sigma\sqrt{2\pi}}},\]
    where $\dw$ denotes the $1$-Wasserstein distance.
\end{lemma}
\begin{proof}[Proof of Lemma~\ref{lem:normal_ks_dw}]
    First, by replacing $A$ and $B$ with $\frac{A-\mu}{\sigma}$ and $\frac{B-\mu}{\sigma}$, respectively, we can assume $\mu=0$ and $\sigma=1$ without loss of generality. Moreover, by definition of the $1$-Wasserstein distance, without loss of generality we can assume that $A,B$ are defined on the same probability space, with $\EE{|A-B|}=\dw(A,B)$.
    
    Fix any $x\in\mathbb R$. Let $\Phi$ denote the CDF of the $\mathcal{N}(0,1)$ distribution. Then, for any $\Delta>0$,
    \begin{align*}
        \PP{A\leq x}
        &= \PP{B + (A-B) \leq x}\\
        &\leq \PP{B\leq x + \Delta} + \PP{|A-B|>\Delta}\\
        &\leq \PP{B\leq x+ \Delta} + \frac{\EE{|A-B|}}{\Delta}\textnormal{ by Markov's inequality}\\
        &= \Phi(x+\Delta) + \dks(B,\mathcal{N}(0,1)) +  \frac{\dw(A,B)}{\Delta}\\
        &\leq \Phi(x) + \frac{\Delta}{\sqrt{2\pi}} +  \dks(B,\mathcal{N}(0,1)) +  \frac{\dw(A,B)}{\Delta},
    \end{align*}
    where the last step holds since $\Phi$ is $\frac{1}{\sqrt{2\pi}}$-Lipschitz. By choosing $\Delta = \sqrt[4]{2\pi}\sqrt{\dw(A,B)}$, we see that
    \[ \PP{A\leq x}- \Phi(x) \leq \dks(B,\mathcal{N}(0,1)) + 2 \frac{\sqrt{\dw(A,B)}}{\sqrt[4]{2\pi}}. \]
    An identical argument provides a lower bound on $\Phi(x) - \PP{A\leq x}$, which means that we have showed
    \[\dks(A,\mathcal{N}(0,1)) \leq \dks(B,\mathcal{N}(0,1)) + 2 \frac{\sqrt{\dw(A,B)}}{\sqrt[4]{2\pi}}.\]
    The claim then holds by the triangle inequality.
\end{proof}

\begin{lemma}\label{lem:cond_tv_ineq}
  Let $P, Q$ be two probability measures on the same space $(\Omega, \mathcal F)$, and let $P_x$ and $Q_x$ be probability measures indexed by $x$ such that for a random variable $X$ and any event $A \in \mathcal F$, $\mathbb E[P_X(A)] = P(A)$ and $\mathbb E[Q_X(A)] = Q(A)$. Then
  \[
    \dtv(P, Q) \leq \mathbb E[ \dtv(P_X, Q_X) ].
  \]
\end{lemma}

\begin{proof}
  We have that
  \begin{align*}
    \dtv(P, Q) &= \sup_{A \in \mathcal F} P(A) - Q(A) \\
               &= \sup_{A \in \mathcal F} \mathbb E[P_X(A)] - \mathbb E[Q_X(A)] \\
               &\leq \mathbb E\left[ \sup_{A \in \mathcal F} P_X(A) - Q_X(A) \right] \\
               &= \mathbb E[ \dtv(P_X, Q_X)].
  \end{align*}
\end{proof}

\begin{lemma}\label{lem:unif_perm_var}
  Let $(a_1, \dots, a_m)$ and $(b_1, \dots, b_m)$, $m \geq 2$, be fixed sequences with $\sum_{i = 1}^m a_i = 0$ and $\sum_{i=1}^m b_i = 0$. Then for $\sigma \sim \mathrm{Unif}(\mathcal S_m)$ we have
  \[
    \Var \left( \sum_{i=1}^m a_{\sigma(i)} b_i \right) = \frac{1}{m-1} \left( \sum_{i=1}^m a_i^2 \right) \left( \sum_{i=1}^m b_i^2 \right)
  \]
\end{lemma}

\begin{proof}
  First note that
  \[
    \mathbb E \left[ \sum_{i=1}^m a_{\sigma(i)}b_i \right] = \frac{1}{m!} \sum_{i, j = 1}^m a_i b_j = 0.
  \]
  Then compute
  \[
    \mathbb E \left[ \left( \sum_{i=1}^m a_{\sigma(i)}b_i\right)^2 \right] = \mathbb E \left[ \sum_{i, j = 1}^m a_{\sigma(i)} a_{\sigma(j)} b_i b_j \right] = \sum_{i, j = 1}^m  b_i b_j \mathbb E[a_{\sigma(i)}a_{\sigma(j)}].
  \]
  Now for $i \neq j$,
  \[
    \mathbb E[a_{\sigma(i)}a_{\sigma(j)}]
    = \frac{1}{m(m - 1)}\sum_{p \neq q} a_p a_q
    = \frac{1}{m(m - 1)}\left( \sum_{p, q = 1}^m a_p a_q - \sum_{p=1}^m a_p^2 \right)
    = -\frac{1}{m(m - 1)}\sum_{p=1}^m a_p^2 
  \]
  and for $i = j$ we have
  \[
    \mathbb E[a_{\sigma(i)}a_{\sigma(j)}] = \frac{1}{m} \sum_{p=1}^m a_p^2
  \]
  so
  \[
    \mathbb E[a_{\sigma(i)}a_{\sigma(j)}] = \left( \frac{1}{m-1} \sum_{p=1}^m a_p^2 \right) \left( \One{i = j} - \frac{1}{m} \right).
  \]
  Thus we can compute
  \[
    \sum_{i, j = 1}^m  b_i b_j \mathbb E[a_{\sigma(i)}a_{\sigma(j)}] = \frac{\sum_{p=1}^m a_p^2}{m-1} \sum_{i=1}^m b_ib_j \left( \One{i = j} - \frac{1}{m} \right)
  \]
  and
  \[
    \sum_{i, j =1}^m b_ib_j \left( \One{i = j} - \frac{1}{m} \right) = \sum_{i=1}^m b_i^2
  \]
  as $\sum_{i, j=1}^m b_ib_j = 0$, which gives us the desired statement.
  
\end{proof}

\begin{lemma}\label{lem:unif_perm_fourth}
  Let $(a_1, \dots, a_m)$ and $(b_1, \dots, b_m)$, $m \geq 2$, be fixed sequences with $\sum_{i = 1}^m a_i = 0$ and $\sum_{i=1}^m b_i = 0$. Then for $\sigma \sim \mathrm{Unif}(\mathcal S_m)$ we have
  \[
    \mathbb E \left[ \left( \sum_{i=1}^m a_{\sigma(i)}b_i \right)^4 \right] \leq C\left\{\frac{A_2^2B_2^2}{m^2}+\frac{A_4B_4}{m}\right\},
  \]
  for a universal constant $C$, where $A_r=\sum_i a_i^r$ and $B_r=\sum_i b_i^r$. For $m\geq4$, the exact formula is
  \begin{align*}
    \mathbb E \left[ \left( \sum_{i=1}^m a_{\sigma(i)}b_i \right)^4 \right]
    ={}&\frac{3(m^2-3m+3)}{m(m-1)(m-2)(m-3)}A_2^2B_2^2\\
    &-\frac{3}{(m-2)(m-3)}\left(A_4B_2^2+A_2^2B_4\right)\\
    &+\frac{m(m+1)}{(m-1)(m-2)(m-3)}A_4B_4.
  \end{align*}
\end{lemma}

\begin{proof}
  For convenience, assume $m \geq 4$. First expand
  \[
    \mathbb E \left[ \left( \sum_{i=1}^m a_{\sigma(i)}b_i \right)^4 \right] = \sum_{i_1, i_2, i_3, i_4=1}^m b_{i_1} b_{i_2} b_{i_3} b_{i_4} \mathbb E[a_{\sigma(i_1)}a_{\sigma(i_2)}a_{\sigma(i_3)}a_{\sigma(i_4)}].
  \]
  Write
  \[
    A_r = \sum_{i=1}^m a_i^r
    \quad \text{and} \quad
    B_r = \sum_{i=1}^m b_i^r
  \]
  for $r = 1, 2, 3, 4$. We will need to compute expectations of the form
  \[
    \mathbb E\left[\prod_{\ell = 1}^k a_{\sigma(i_\ell)}^{r_\ell}\right] = \frac{(m - k)!}{m!} \sum_{i_1 \neq \cdots \neq i_k} \prod_{\ell = 1}^k a_{i_\ell}^{r_\ell}
  \]
  for $k = 1, 2, 3, 4$, which we can do by inclusion-exclusion. More precisely, for sets of the form $S = \{(p, q) \mid 1 \leq p, q \leq k \}$ (denoting collisions between indices),
  \[
    \sum_{i_1 \neq \cdots \neq i_k} \prod_{\ell = 1}^k a_{i_\ell}^{r_\ell} = \sum_{S} (-1)^{|S|} \sum_{\substack{\ell_1, \dots, i_k \\ i_p = i_q \forall (p, q) \in S}} \prod_{\ell=1}^k a_{i_\ell}^{r_\ell}
  \]
  by the usual inclusion-exclusion principle; we give explicit computations for $k = 1, 2, 3, 4$.
  \begin{itemize}
  \item $k = 1$. This is straightforwardly
    \[
      \sum_{i_1} a_{i_1}^{r_1} = A_{r_1}.
    \]
  \item $k = 2$. Similarly
    \[
      \sum_{i_1 \neq i_2} a_{i_1}^{r_1} a_{i_2}^{r_2} = A_{r_1} A_{r_2} - A_{r_1 + r_2}.
    \]
  \item $k = 3$. We have 
    \[
      \sum_{i_1 \neq i_2 \neq i_3} a_{i_1}^{r_1} a_{i_2}^{r_2}a_{i_3}^{r_3} = A_{r_1} A_{r_2} A_{r_3} - A_{r_1 + r_2}A_{r_3} - A_{r_1 + r_3}A_{r_2} - A_{r_2 + r_3}A_{r_1} + 2A_{r_1 + r_2 + r_3}.
    \]
    Here, the coefficient $2$ comes from the fact that $S = \{(1, 2), (1, 3), (2, 3)\}$ and
    \[
      S \in \{\{(1, 2), (2, 3)\}, S = \{(1, 2), (1, 3)\}, S = \{(1, 3), (2, 3)\}\}
    \]
    both give rise to
    \[
      \sum_{\substack{\ell_1, \dots, i_k \\ i_p = i_q \forall (p, q) \in S}} \prod_{\ell=1}^k a_{i_\ell}^{r_\ell} = A_{r_1 + r_2 + r_3}
    \]
    meaning that the overall coefficient is $(-1)^3 + 3 \cdot (-1)^2 = 2$.
  \item $k = 4$. A similar combinatorial calculation shows that
    \begin{multline*}
      \sum_{i_1 \neq \cdots \neq i_4} a_{i_1} a_{i_2} a_{i_3} a_{i_4} = A_{r_1}A_{r_2}A_{r_3}A_{r_4} - \sum_{\{i_1, i_2\} \subset [4]} A_{r_{i_1} + r_{i_2}} \prod_{i \neq i_1, i_2} A_{r_i} \\
      + \sum_{\substack{\{i_1, i_2\}, \{i_3, i_4\} \subset [4] \\ \{i_1, i_2\} \cap \{i_3, i_4\} = \emptyset}} A_{r_{i_1} + r_{i_2}} A_{r_{i_3} + r_{i_4}} + 2 \sum_{\substack{\{i_1, i_2, i_3\} \subset [4] \\ i_4 \neq i_1, i_2, i_3}} A_{r_{i_1} + r_{i_2} + r_{i_3}} A_{r_{i_4}} - 6 A_{r_1 + r_2 + r_3 + r_4}.
    \end{multline*}
    Here, the sums are over all partitions of $[4] = \{1, 2, 3, 4\}$ satisfying the given condition. For $r_1 = r_2 = r_3 = r_4 = 1$, this simplifies to
    \[
      A_1^4 - 6A_2A_1^2 + 3A_2^2 + 8A_3 A_1 - 6A_4 = 3A_2^2 - 6 A_4
    \]
    when $A_1 = 0$.
  \end{itemize}

  Now to compute $\sum_{i_1, i_2, i_3, i_4=1}^m b_{i_1} b_{i_2} b_{i_3} b_{i_4} \mathbb E[a_{\sigma(i_1)}a_{\sigma(i_2)}a_{\sigma(i_3)}a_{\sigma(i_4)}]$, we separate into the above cases.
  \begin{itemize}
  \item All four indices collide: $i_1 = i_2 = i_3 = i_4$. This case contributes
    \[
      \sum_{i=1}^m b_i^4 \mathbb E[a_{\sigma(i)}^4] = \frac{1}{m} A_4 B_4
    \]
    as the inner expectation is $\mathbb E[a_{\sigma(i)}^4] = \frac{1}{m} A_4$.
  \item Three indices collide (e.g. $i_1 = i_3 = i_4, i_1 \neq i_2$). The inner expectation is $\mathbb E[a_{\sigma(i_1)}^3 a_{\sigma(i_2)}] = \frac{A_3 A_1 - A_4}{m(m-1)} = -\frac{A_4}{m(m-1)}$ when $A_1 = 0$ so this case contributes
    \[
      4 \frac{-A_4}{m(m-1)} \sum_{i_1 \neq i_j}^m b_{i_1}^3 b_{i_2} = 4 \frac{A_4 B_4}{m(m - 1)}
    \]
    where the $4$ comes from the $4$ different ways to choose which three indices are the ones which collide.
  \item Two indices collide (e.g. $i_1 = i_4$, $i_1 \neq i_2 \neq i_3$). The inner expectation is
    \[
      \mathbb E[a_{\sigma(i_1)}^2 a_{\sigma(i_2)} a_{\sigma(i_3)}] = \frac{A_2 A_1^2 - 2A_3 A_1 - A_2^2 + 2A_4}{m(m-1)(m-2)} = \frac{2A_4 - A_2^2}{m(m-1)(m-2)}
    \]
    and similarly the contribution from the summed $b$ terms is also $2B_4 - B_2^2$, so the total contribution from this case is
    \[
      6 \frac{(2A_4 - A_2^2)(2B_4 - B_2^2)}{m(m-1)(m-2)}
    \]
    as there are $6$ ways to choose the colliding pair.
  \item Two pairs of indices collide (e.g. $i_1 = i_3, i_2 = i_4$, $i_1 \neq i_2$). The inner expectation is
    \[
      \mathbb E[a_{\sigma(i_1)}^2 a_{\sigma(i_2)}^2] = \frac{A_2^2 - A_4}{m(m-1)}
    \]
    and similarly the contribution from the summed $b$ terms is $B_2^2 - B_4$; thus the total contribution is
    \[
      3 \frac{(A_2^2 - A_4)(B_2^2 - B_4)}{m(m - 1)}
    \]
    where we have $3 = 6 / 2$ ways of choosing the two pairs.
  \item No indices collide (e.g. $i_1 \neq i_2 \neq i_3 \neq i_4)$. The inner expectation is, as computed above
    \[
      \mathbb E[a_{\sigma(i_1)}a_{\sigma(i_2)}a_{\sigma(i_3)}a_{\sigma(i_4)}] = \frac{3A_2^2 - 6A_4}{m(m-1)(m-2)(m-3)}
    \]
    and the contribution from the summed $b$ terms is also $3B_2^2 - 6B_4$ so we get a total contribution of
    \[
      \frac{(3A_2^2 - 6A_4)(3B_2^2 - 6B_4)}{m(m-1)(m-2)(m-3)}.
    \]
  \end{itemize}

  After summing the above contributions and simplifying, we see that 
  \begin{align*}
    \mathbb E \left[ \left( \sum_{i=1}^m a_{\sigma(i)}b_i \right)^4 \right]
    ={}& \frac{A_4 B_4}{m} + \frac{4 A_4 B_4}{m(m-1)}
      + \frac{3(A_2^2 - A_4)(B_2^2 - B_4)}{m(m-1)}\\
      &+ \frac{6(2A_4 - A_2^2)(2B_4 - B_2^2)}{m(m-1)(m-2)}
      + \frac{9(A_2^2 - 2A_4)(B_2^2 - 2B_4)}{m(m-1)(m-2)(m-3)}\\
    ={}&\frac{3(m^2-3m+3)}{m(m-1)(m-2)(m-3)}A_2^2B_2^2\\
      &- \frac{3}{(m - 2)(m-3)} (A_4B_2^2 + A_2^2 B_4)
      + \frac{m(m+1)}{(m - 1)(m-2)(m-3)} A_4B_4.
  \end{align*}
  The middle term is nonpositive. The first and third coefficients are bounded by $C/m^2$ and $C/m$, respectively, uniformly for $m\geq4$, and hence
  \[
    \mathbb E \left[ \left( \sum_{i=1}^m a_{\sigma(i)}b_i \right)^4 \right] \leq C\left\{\frac{A_2^2B_2^2}{m^2}+\frac{A_4B_4}{m}\right\}.
  \]
  For $m=2,3$, the finitely many collision patterns are absorbed by increasing the universal constant $C$.
\end{proof}

\begin{lemma}\label{lem:neighbor_expectation}
  For any measurable $f: \RR \to [0, \infty)$, we have that
  \[
    \mathbb E[f(Z_{N(1), n})] \leq 2 \mathbb E[f(Z_1)].
  \]
\end{lemma}

\begin{proof}
  This is Lemma 9.4 of \cite{chatterjee2021new}.
\end{proof}

\begin{lemma}\label{lem:conv_in_prob}
  For any measurable $f: \RR \to \RR$, we have that $f(Z_{N(1), n}) \to f(Z_{1})$ in probability.
\end{lemma}

\begin{proof}
  This is Lemma 9.5 of \cite{chatterjee2021new}.
\end{proof}

\begin{lemma}\label{lem:ui}
  The family of random variables $\{p(x \mid Z_{N(1), n})\}_{n=1}^\infty$ is uniformly integrable for any fixed $x$.
\end{lemma}

\begin{proof}
  First, note that
  \[
    \E[p(x \mid Z_1)] = p(x) < \infty
  \]
  and so by dominated convergence
  \[
    \lim_{K \to \infty} \E[p(x \mid Z_1) \cdot \One{p(x \mid Z_1) \geq K}] = \E \left[ \lim_{K \to \infty} p(x \mid Z_1) \cdot \One{p(x \mid Z_1) \geq K} \right] = 0.
  \]
  Then, by Lemma \ref{lem:neighbor_expectation} applied to $f(z) = p(x \mid z) \cdot \One{p(x \mid z) \geq K}$, we have that
  \[
    \lim_{K \to \infty}\sup_n \E [p(x \mid Z_{N(1), n}) \cdot \One{p(x \mid Z_{N(1), n}) \geq K}] \leq \lim_{K \to \infty} 2 \E[p(x \mid Z_1) \cdot \One{p(x \mid Z_1) \geq K}] = 0
  \]
  so we may conclude.
\end{proof}

\begin{lemma}\label{lem:gen_hellinger}
  For any $1 \leq \gamma_1 \leq \gamma_2$, we have
  \[
    \mathrm d_{\mathrm H,\gamma_2}^{\gamma_2}(P, Q) \leq 
    \mathrm d_{\mathrm H,\gamma_1}^{\gamma_1}(P, Q).
  \]
  Moreover, for any $\gamma \geq 1$ and $1 \leq \alpha \leq 2$, it holds that
  \[
    \mathrm d_{\mathrm H,\gamma}(P, Q) \leq 2^{1 - 1/\gamma + 1/(\gamma \alpha)} \mathrm d_{\mathrm H,\gamma\alpha}(P, Q).
  \]
\end{lemma}

\begin{proof}
The above lemma can be found with proof as Lemma 4 of \cite{kim2022local}.
\end{proof}

\begin{lemma}\label{lem:renyi}
  For every fixed $\gamma>0$,
  \[
    \mathrm d_{\mathrm H}^2(P,Q)\leq C_\gamma\,\mathrm{d_{R,\gamma}}(P\|Q),
    \qquad
    C_\gamma=\frac12\max\left\{1,\frac{1-\gamma}{\gamma}\right\}.
  \]
\end{lemma}

\begin{proof}
  The identity
  \[
    \mathrm d_{\mathrm H}^2(P,Q)=1-\exp\{-\mathrm d_{R,1/2}(P\|Q)/2\}
  \]
  implies $\mathrm d_{\mathrm H}^2(P,Q)\leq \mathrm d_{R,1/2}(P\|Q)/2$. For $\gamma\geq1/2$, monotonicity of R\'enyi divergence in its order completes the proof. For $0<\gamma<1/2$, log-convexity of $t\mapsto\int p^tq^{1-t}$ gives
  $\mathrm d_{R,1/2}(P\|Q)\leq(1-\gamma)\mathrm d_{R,\gamma}(P\|Q)/\gamma$; see \cite{van2014renyi}.
\end{proof}

\begin{lemma}\label{lem:VarT_k}
  Under the model \eqref{model:linear_nongaussian_confounder} and in the notation of Section \ref{sec:power_linear_proof},
  \begin{multline*}
    \Var(T_k \mid \mathbf Z) = \Var(A_{k, 1} \mid \mathbf Z) + \Var(A_{k, 2} \mid \mathbf Z) + \Var(A_{k, 3} \mid \mathbf Z) + 2\Cov(A_{k, 1}, A_{k, 2} \mid \mathbf Z) \\
    = \frac{(\beta_{2, n}^2 + 1)\,S_k^{(1)} + (\beta_{1, n}^2 + 1)\,S_k^{(2)} + 2\beta_{1, n}\beta_{2, n}\,S_k^{(12)}}{m_k} \\
    + \frac{m_k - 1}{m_k^2}\left((\beta_{1, n}^2 + 1)(\beta_{2, n}^2 + 1) + \left(\mu_4 - 2 - \frac{\mu_4 - 3}{m_k}\right)\beta_{1, n}^2\beta_{2, n}^2\right).
  \end{multline*}
  In particular, when $U$ is Gaussian we have $\mu_4 = 3$ and the second term reduces to
  \[
    \frac{m_k - 1}{m_k^2}\left(\beta_{1, n}^2\beta_{2, n}^2 + (\beta_{1, n}^2 + 1)(\beta_{2, n}^2 + 1)\right),
  \]
  recovering \eqref{eq:var_nonperm}.
\end{lemma}

\begin{proof}
  The variance is more complicated, but still straightforward to compute. We use the decomposition $T_k = S_k^{(12)} + A_{k, 1} + A_{k, 2} + A_{k, 3}$ from the proof of Theorem \ref{thm:power_lpt_confounder_model_general} in Section \ref{sec:power_linear_proof}, and again decompose
  \[
    \Var(T_k \mid \mathbf Z) = \Var(A_{k, 1} + A_{k, 2} + A_{k, 3} \mid \mathbf Z)
  \]
  and compute all the corresponding variances and covariances. First,
  \begin{align*}
    \Var(A_{k, 1} \mid \mathbf Z) &= \frac{1}{m_k^2} \sum_{i, j \in B_k}  \left( f_1(Z_i) - \overline{f_1(\mathbf Z)}_k  \right) \left( f_1(Z_j) - \overline{f_1(\mathbf Z)}_k  \right) \Cov(V_i - \overline{\mathbf V}_k, V_j - \overline{\mathbf V}_{k} \mid \mathbf Z) \\
                                  &= \frac{1}{m_k^2} \sum_{i, j \in B_k}  \left( f_1(Z_i) - \overline{f_1(\mathbf Z)}_k  \right) \left( f_1(Z_j) - \overline{f_1(\mathbf Z)}_k  \right) (\beta_{2, n}^2 + 1) \left( \One{i = j} - \frac{1}{m_k} \right)\\
                                  &= \frac{\beta_{2, n}^2 + 1}{m_k}S_k^{(1)}
  \end{align*}
  where the last equality follows by noting that
  \[
    \sum_{i, j \in B_k}  \left( f_1(Z_i) - \overline{f_1(\mathbf Z)}_k  \right) \left( f_1(Z_j) - \overline{f_1(\mathbf Z)}_k  \right) = \left( \sum_{i \in B_k} \left( f_1(Z_i) - \overline{f_1(\mathbf Z)}_k \right) \right)^2 = 0.
  \]
  By symmetry
  \[
    \Var(A_{k, 2} \mid \mathbf Z) = \frac{\beta_{1, n}^2 + 1}{m_k} S_k^{(2)}.
  \]
  Then,
  \[
    \Var(A_{k, 3} \mid \mathbf Z)
    = \frac{1}{m_k^2}\Var\left(\sum_{i \in B_k}(W_i - \overline{\mathbf W}_k)(V_i - \overline{\mathbf V}_k) \mid \mathbf Z\right)
    = \frac{1}{m_k^2}\Var\left(\sum_{i \in B_k} W_iV_i - m_k \overline{\mathbf W}_k \overline{\mathbf V}_k \mid \mathbf Z\right).
  \]
  It will be useful to first compute expectations of various products of $W$ and $V$. First we compute
  \begin{align*}
    \E[W_i V_j \mid \mathbf Z] &= \beta_{1, n}\beta_{2, n}\One{i = j},\\
    \E[W_i W_j \mid \mathbf Z] &= (\beta_{1, n}^2 + 1)\One{i = j},\\
    \E[V_i V_j \mid \mathbf Z] &= (\beta_{2, n}^2 + 1)\One{i = j}.
  \end{align*}
  Now expanding
  \[
    W_i V_j W_p V_q = (\beta_{1, n} U_i + \epsilon_{1, i})(\beta_{2, n} U_j + \epsilon_{2, j})(\beta_{1, n} U_p + \epsilon_{1, p})(\beta_{2, n} U_q + \epsilon_{2, q}).
  \]
  gives $16$ cross-products; since $U$ and $\epsilon_1, \epsilon_2$ are mean-zero and independent, any term containing a single $U$, $\epsilon_1$, or $\epsilon_2$ at a single unique index vanishes in expectation. Four index-types survive:
  \begin{itemize}
  \item $i = j, p = q$: in this case,
    \[
      \E[ W_i V_j W_p V_q \mid \mathbf Z] = \E[W_iV_i\mid \mathbf Z] \E[W_pV_p\mid \mathbf Z] = \beta_{1,n}^2 \beta_{2, n}^2.
    \]
  \item $i = p, j = q$: in this case,
    \[
      \E[ W_i V_j W_p V_q \mid \mathbf Z] = \E[W_i^2\mid \mathbf Z] \E[V_j^2\mid \mathbf Z] = (\beta_{1, n}^2 + 1)(\beta_{2, n}^2 + 1).
    \]
  \item $i = q, j = p$: in this case,
    \[
      \E[ W_i V_j W_p V_q \mid \mathbf Z] = \E[W_iV_i\mid \mathbf Z] \E[W_jV_j\mid \mathbf Z] = \beta_{1,n}^2 \beta_{2, n}^2.
    \]
  \item $i = j = p = q$: in this case, expanding $W_i^2 V_i^2$ and noting that every term containing an odd power of $\epsilon_{1, 1}$ or $\epsilon_{2, 1}$ vanishes in expectation (in particular, the terms proportional to $\E[U_1^3]$ carry a factor of $\E[\epsilon_{1,1}]$ or $\E[\epsilon_{2,1}]$),
    \begin{multline*}
      \E[ W_i V_j W_p V_q \mid \mathbf Z] = \E[W_i^2V_i^2 \mid \mathbf Z] \\
      = \beta_{1, n}^2\beta_{2, n}^2 \E[U_1^4] + \beta_{1, n}^2\E[U_1^2]\E[\epsilon_{2, 1}^2] + \beta_{2, n}^2\E[U_1^2]\E[\epsilon_{1, 1}^2] + \E[\epsilon_{1, 1}^2]\E[\epsilon_{2, 1}^2] \\
      = \mu_4 \beta_{1,n}^2 \beta_{2, n}^2 + \beta_{1, n}^2 + \beta_{2, n}^2 + 1
      = (\mu_4 - 1)\beta_{1,n}^2 \beta_{2, n}^2 + (\beta_{1, n}^2 + 1)(\beta_{2, n}^2 + 1).
    \end{multline*}
  \end{itemize}
  For brevity, in the remainder of the proof we write
  \[
    a := \Var(W_1 V_1 \mid \mathbf Z) = (\mu_4 - 2)\beta_{1,n}^2 \beta_{2, n}^2 + (\beta_{1, n}^2 + 1)(\beta_{2, n}^2 + 1), \quad
    b := \beta_{1,n}^2 \beta_{2, n}^2 + (\beta_{1, n}^2 + 1)(\beta_{2, n}^2 + 1),
  \]
  so that $a - b = (\mu_4 - 3)\beta_{1,n}^2 \beta_{2, n}^2$ and $a = b$ exactly in the Gaussian case.

  Write $C_k = \sum_{i \in B_k} W_iV_i$ and $D_k = m_k \overline{\mathbf W}_k \overline{\mathbf V}_k = \frac{1}{m_k}\sum_{i,j \in B_k} W_i V_j$; then we first compute
  \[
    \Var(C_k \mid \mathbf Z) = m_k \Var(W_1 V_1 \mid \mathbf Z) = m_k a.
  \]
  Now develop
  \begin{align*}
    \Var(D_k \mid \mathbf Z) &= \frac{1}{m_k^2}\sum_{i, j, p, q \in B_k}\Cov(W_i V_j, W_p V_q \mid \mathbf Z) \\
                             &= \frac{1}{m_k^2} \sum_{i,j,p,q \in B_k} \mathbb E[W_i V_j W_p V_q \mid \mathbf Z] - \mathbb E[ W_iV_j \mid \mathbf Z]\mathbb E[ W_pV_q \mid \mathbf Z] \\
                             &= \frac{1}{m_k^2} \sum_{i, j, p, q \in B_k}(\beta_{1,n}^2+1)(\beta_{2,n}^2+1)\One{i = p}\One{j = q} + \beta_{1,n}^2\beta_{2,n}^2\One{i = q}\One{j = p} \\
                             &\qquad\qquad\qquad\qquad + (\mu_4 - 3)\beta_{1,n}^2\beta_{2,n}^2\One{i = j = p = q} \\
                             &= b + \frac{a - b}{m_k},
  \end{align*}
  where the third indicator corrects the diagonal: on $i = j = p = q$ the first two indicators contribute only $b$, whereas the centered fourth moment computed above equals $a = b + (\mu_4 - 3)\beta_{1,n}^2\beta_{2,n}^2$. (In the Gaussian case this correction vanishes and the first two terms happen to cover the diagonal exactly.)
  Finally, in
  \[
    \Cov(C_k, D_k \mid \mathbf Z) = \frac{1}{m_k}\sum_{i, p, q \in B_k}\Cov(W_i V_i, W_p V_q \mid \mathbf Z)
  \]
  only the terms with $p = q = i$ contribute: if $p = q \neq i$ the two products are independent, and if $p \neq q$ then either $W_p$ or $V_q$ appears at a unique index and decouples as a mean-zero factor. Hence
  \[
    \Cov(C_k, D_k \mid \mathbf Z) = \frac{1}{m_k} \cdot m_k \Var(W_1 V_1 \mid \mathbf Z) = a.
  \]
  So we get
  \begin{align*}
    \Var(A_{k, 3} \mid \mathbf Z) &= \frac{1}{m_k^2}\left( \Var(C_k \mid \mathbf Z) - 2\Cov(C_k, D_k \mid \mathbf Z) + \Var(D_k \mid \mathbf Z) \right) \\
    &= \frac{1}{m_k^2}\left( m_k a - 2a + b + \frac{a - b}{m_k} \right)
    = \frac{m_k - 1}{m_k^3}\left( (m_k - 1) a + b \right) \\
    &= \frac{m_k - 1}{m_k^2}\left( (\beta_{1, n}^2 + 1)(\beta_{2, n}^2 + 1) + \left( \mu_4 - 2 - \frac{\mu_4 - 3}{m_k} \right)\beta_{1, n}^2\beta_{2, n}^2 \right).
  \end{align*}
  Only the covariances between $A_{k, 1}, A_{k, 2}, A_{k, 3}$ remain. First,
  \begin{multline*}
    \Cov(A_{k, 1}, A_{k, 2} \mid \mathbf Z) = \frac{1}{m_k^2}\sum_{i, j \in B_k}\bigl(f_1(Z_i) - \overline{f_1(\mathbf Z)}_k\bigr)\bigl(f_2(Z_j) - \overline{f_2(\mathbf Z)}_k\bigr)\Cov(V_i - \overline{\mathbf V}_k, W_j - \overline{\mathbf W}_k \mid \mathbf Z) \\
    = \frac{\beta_{1,n}\beta_{2,n}}{m_k^2}\sum_{i \in B_k}\bigl(f_1(Z_i) - \overline{f_1(\mathbf Z)}_k\bigr)\bigl(f_2(Z_i) - \overline{f_2(\mathbf Z)}_k\bigr) = \frac{\beta_{1,n}\beta_{2,n}\, S_k^{(12)}}{m_k}
  \end{multline*}
  by an analogous expansion to the one used to compute $\Var(A_{k, 1} \mid \mathbf Z)$.
  Finally,
  \[
    \Cov(A_{k, 1}, A_{k, 3} \mid \mathbf Z) = \frac{1}{m_k^2} \sum_{i, j\in B_k} \left(f_1(Z_i) - \overline{f_1(\mathbf Z)_k} \right) \Cov \left( (V_i - \overline{\mathbf V}_k), (W_j - \overline{\mathbf W}_k)(V_j - \overline{\mathbf V}_k)\right).
  \]
  However, examining the covariance on the right shows that it takes on only two possible values depending on if $i = j$ (as it is clearly independent of the particular value of $j$ when $i \neq j$): write
  \[
    \Cov \left( (V_i - \overline{\mathbf V}_k), (W_j - \overline{\mathbf W}_k)(V_p - \overline{\mathbf V}_k)\right) =
    \begin{cases}
      \alpha_1 & i = j \\
      \alpha_2 & i \neq j.
    \end{cases}
  \]
  Then we can express
  \begin{multline*}
    \sum_{i, j\in B_k} \left(f_1(Z_i) - \overline{f_1(\mathbf Z)_k} \right) \Cov \left( (V_i - \overline{\mathbf V}_k), (W_j - \overline{\mathbf W}_k)(V_j - \overline{\mathbf V}_k)\right) \\
    = \alpha_1 \sum_{i \in B_k} \left(f_1(Z_i) - \overline{f_1(\mathbf Z)_k} \right) + 
    \alpha_2 \sum_{i, j \in B_k} \left(f_1(Z_i) - \overline{f_1(\mathbf Z)_k} \right) -
    \alpha_2 \sum_{i\in B_k} \left(f_1(Z_i) - \overline{f_1(\mathbf Z)_k} \right) = 0.
  \end{multline*}
  as $\sum_{i \in B_k} \left(f_1(Z_i) - \overline{f_1(\mathbf Z)_k} \right) = 0$ by construction. Note that this argument uses no distributional features of $(W_i, V_i)$ beyond their i.i.d.-ness across $i$; in particular, the nonzero third moments $\E[U^3], \E[\epsilon_1^3], \E[\epsilon_2^3]$ entering $\alpha_1$ and $\alpha_2$ are eliminated by the centering. So $\Cov(A_{k, 1}, A_{k, 3} \mid \mathbf Z) = 0$ (and by symmetry, $\Cov(A_{k, 2}, A_{k, 3} \mid \mathbf Z) = 0$ as well).
  Finally,
  \begin{multline}\label{eq:var_nonperm_general}
    \Var(T_k \mid \mathbf Z) = \Var(A_{k, 1} \mid \mathbf Z) + \Var(A_{k, 2} \mid \mathbf Z) + \Var(A_{k, 3} \mid \mathbf Z) + 2\Cov(A_{k, 1}, A_{k, 2} \mid \mathbf Z) \\
    = \frac{(\beta_{2, n}^2 + 1)\,S_k^{(1)} + (\beta_{1, n}^2 + 1)\,S_k^{(2)} + 2\beta_{1, n}\beta_{2, n}\,S_k^{(12)}}{m_k} \\
    + \frac{m_k - 1}{m_k^2}\left((\beta_{1, n}^2 + 1)(\beta_{2, n}^2 + 1) + \left(\mu_4 - 2 - \frac{\mu_4 - 3}{m_k}\right)\beta_{1, n}^2\beta_{2, n}^2\right).
  \end{multline}
\end{proof}

\begin{lemma}\label{lem:dtv_kl_comparison}
    In the setting of Lemma \ref{lem:binary_power_lb_orc}, for any $a > 0$, it holds that
    \[
    \dtv(P_{X, Y \mid Z_i}, P_{X \mid Z_i} \times P_{Y \mid Z_i}) \geq \frac{\overline{\mathrm{KL}}_{(i)} - \frac{2\sigma^2}{a} \exp\left(-\frac{a^2}{2\sigma^2}\right)}{\overline{\mathrm{KL}}_{(i)} + 2a}.
    \]
    In particular, taking $a = C \sigma$ for a constant $C > 0$,
    \[
    \dtv(P_{X, Y \mid Z_i}, P_{X \mid Z_i} \times P_{Y \mid Z_i}) \geq \min\left\{ \frac{\overline{\mathrm{KL}}_{(i)}}{4C\sigma},\ \frac{1}{2} \right\} - \frac{e^{-C^2/2}}{C^2},
    \]
    and if additionally $\sigma^2 \leq M\, \mathrm{V}_{\mathrm{KL}, (i)}$ for some constant $M \geq 1$, the same bound holds with $\sigma$ replaced by $\sqrt{M\, \mathrm{V}_{\mathrm{KL}, (i)}}$.
\end{lemma}

\begin{proof}
Fix $i \in [n]$ and write
\[
  \mu_{0,i} = \E_{P_{X\mid Z_i}\times P_{Y\mid Z_i}}\left[\mathrm{LLR}^{(i)}(X, Y)\right]
  \quad \text{and} \quad
  \mu_{1,i} = \E_{P_{X,Y\mid Z_i}}\left[\mathrm{LLR}^{(i)}(X, Y)\right]
\]
for the means of the log-likelihood ratio under the null and the alternative, so that $\overline{\mathrm{KL}}_{(i)} = \mu_{1,i} - \mu_{0,i} \geq 0$. For $a > 0$, define the truncated log-likelihood ratio
\[
  \psi_a(x, y) = \min\left\{ \max\left\{ \mathrm{LLR}^{(i)}(x, y),\ \mu_{0,i} - a \right\},\ \mu_{1,i} + a \right\},
\]
that is, $\mathrm{LLR}^{(i)}$ clipped to the interval $[\mu_{0,i} - a,\ \mu_{1,i} + a]$. The proof consists of two observations: the separation $\E_{P_{X,Y\mid Z_i}}[\psi_a(X, Y)] - \E_{P_{X\mid Z_i}\times P_{Y\mid Z_i}}[\psi_a(X, Y)]$ of this bounded test function (1) is close to $\overline{\mathrm{KL}}_{(i)}$, by sub-Gaussianity of $\mathrm{LLR}^{(i)}$, and (2) lower bounds the total variation, since $\psi_a$ is bounded.

\paragraph{The separation of $\psi_a$ lower bounds total variation.} Since $\psi_a$ takes values in an interval of length $\overline{\mathrm{KL}}_{(i)} + 2a$, writing $c$ for the midpoint of that interval,
\begin{multline*}
  \E_{P_{X,Y\mid Z_i}}[\psi_a(X, Y)] - \E_{P_{X\mid Z_i}\times P_{Y\mid Z_i}}[\psi_a(X, Y)] = \int (\psi_a - c)\,\bigl(\mathsf{d}P_{X,Y\mid Z_i} - \mathsf{d}(P_{X\mid Z_i}\times P_{Y\mid Z_i})\bigr) \\
  \leq \frac{\overline{\mathrm{KL}}_{(i)} + 2a}{2} \int \bigl|\mathsf{d}P_{X,Y\mid Z_i} - \mathsf{d}(P_{X\mid Z_i}\times P_{Y\mid Z_i})\bigr| = \left(\overline{\mathrm{KL}}_{(i)} + 2a\right) \dtv(P_{X, Y \mid Z_i}, P_{X \mid Z_i} \times P_{Y \mid Z_i}).
\end{multline*}

\paragraph{The separation of $\psi_a$ is close to $\overline{\mathrm{KL}}_{(i)}$.} Pointwise, the truncation satisfies
\[
  \mathrm{LLR}^{(i)} - \left( \mathrm{LLR}^{(i)} - \mu_{1,i} - a \right)_+ \;\leq\; \psi_a \;\leq\; \mathrm{LLR}^{(i)} + \left( \mu_{0,i} - a - \mathrm{LLR}^{(i)} \right)_+.
\]
By the Chernoff bound, sub-Gaussianity of $\mathrm{LLR}^{(i)}$ under the alternative gives
\[
  \P_{P_{X,Y\mid Z_i}}\left(\mathrm{LLR}^{(i)}(X, Y) - \mu_{1,i} > t\right) \leq e^{-t^2/(2\sigma^2)}
\]
so
\begin{multline*}
  \E_{P_{X,Y\mid Z_i}}\left[ \left(\mathrm{LLR}^{(i)}(X, Y) - \mu_{1,i} - a\right)_+ \right] = \int_a^\infty \P_{P_{X,Y\mid Z_i}}\left(\mathrm{LLR}^{(i)}(X, Y) - \mu_{1,i} > t\right) \mathsf{d}t \\
  \leq \int_a^\infty e^{-\frac{t^2}{2\sigma^2}}\,\mathsf{d}t \leq \int_a^\infty \frac{t}{a}\, e^{-\frac{t^2}{2\sigma^2}}\, \mathsf{d}t = \frac{\sigma^2}{a} e^{-\frac{a^2}{2\sigma^2}},
\end{multline*}
and, symmetrically,
\[
  \E_{P_{X\mid Z_i}\times P_{Y\mid Z_i}}
  \left[\left(\mu_{0,i}-a-\mathrm{LLR}^{(i)}(X,Y)\right)_+\right]
  \leq \frac{\sigma^2}{a}e^{-\frac{a^2}{2\sigma^2}},
\]
using the sub-Gaussianity of $-(\mathrm{LLR}^{(i)}(X, Y) - \mu_{0,i})$ under the null. Hence
\[
  \E_{P_{X,Y\mid Z_i}}[\psi_a(X, Y)] - \E_{P_{X\mid Z_i}\times P_{Y\mid Z_i}}[\psi_a(X, Y)] \geq \mu_{1,i} - \mu_{0,i} - \frac{2\sigma^2}{a}e^{-\frac{a^2}{2\sigma^2}} = \overline{\mathrm{KL}}_{(i)} - \frac{2\sigma^2}{a}e^{-\frac{a^2}{2\sigma^2}}.
\]
Combining the two bounds on the separation and rearranging gives the first display of the lemma.

For the second display, set $a = C\sigma$. For the leading term, since $\overline{\mathrm{KL}}_{(i)} + 2C\sigma \leq 2\max\{\overline{\mathrm{KL}}_{(i)},\ 2C\sigma\}$,
\[
  \frac{\overline{\mathrm{KL}}_{(i)}}{\overline{\mathrm{KL}}_{(i)} + 2C\sigma} \geq \min\left\{ \frac{\overline{\mathrm{KL}}_{(i)}}{4C\sigma},\ \frac{1}{2} \right\},
\]
while the error term satisfies
\[
  \frac{\frac{2\sigma^2}{C\sigma}e^{-C^2/2}}{\overline{\mathrm{KL}}_{(i)} + 2C\sigma} \leq \frac{2\sigma e^{-C^2/2}}{C \cdot 2C\sigma} = \frac{e^{-C^2/2}}{C^2}.
\]
Finally, if $\sigma^2 \leq M \,\mathrm{V}_{\mathrm{KL}, (i)}$, then $\overline{\mathrm{KL}}_{(i)}/(4C\sigma) \geq \overline{\mathrm{KL}}_{(i)}/(4C\sqrt{M\, \mathrm{V}_{\mathrm{KL}, (i)}})$, so the same lower bound holds with $\sigma$ replaced by $\sqrt{M\, \mathrm{V}_{\mathrm{KL}, (i)}}$.
\end{proof}

\begin{lemma}\label{lem:binary_drift_perturbation}
  In the notation of Section \ref{sec:prf_gaussian_binary_power}, for every bin $k$,
  \[
    |R_k|
    \leq \frac{2}{\pi\sqrt{1-\bar\rho^2}}
    \left( |\delta_{1,k}\,\delta_{2,k}|
      + |\rho_n| \bigl( \delta_{1,k}^2 + \delta_{2,k}^2 \bigr) \right)
    \leq \frac{4}{\pi\sqrt{1-\bar\rho^2}}
    \left( \sqrt{S_k^{(1)} S_k^{(2)}}
      + |\rho_n| \bigl( S_k^{(1)} + S_k^{(2)} \bigr) \right).
  \]
\end{lemma}

\begin{proof}
Write $G_l = D_l - \Delta_{l,k}$ for the driftless differences, so that
$(G_1, G_2)$ is centered bivariate Gaussian with variances
$\sigma_1^2, \sigma_2^2$ and correlation $\rho_n$, and define
$T_k^0 = \mathrm{sign}(G_1)\mathrm{sign}(G_2)$ on the same probability space.
By the computation of Step 1 applied to $(G_1, G_2)$, we have
$\E[T_k^0 \mid \mathbf Z] = 2 q_n$, and hence
$R_k = \E[T_k - T_k^0 \mid \mathbf Z]$. Telescoping,
\begin{equation}\label{eq:binary_telescope}
  T_k - T_k^0
  = \bigl( \mathrm{sign}(D_1) - \mathrm{sign}(G_1) \bigr)\mathrm{sign}(D_2)
  + \mathrm{sign}(G_1)\bigl( \mathrm{sign}(D_2) - \mathrm{sign}(G_2) \bigr).
\end{equation}

\paragraph{The flip event.} For any $\Delta \in \mathbb R$ and any random
variable $g$ with a continuous distribution, almost surely
\[
  \mathrm{sign}(\Delta + g) - \mathrm{sign}(g)
  = 2\mathrm{sign}(\Delta) \cdot \One{g \in J},
  \qquad
  J := \bigl( -\Delta \wedge 0,\ -\Delta \vee 0 \bigr),
\]
an interval of length $|\Delta|$ (check the cases $\Delta > 0$ and
$\Delta < 0$ directly; $\Delta = 0$ gives $J = \emptyset$). Since the density
of $G_l$ is bounded by $1/(\sigma_l \sqrt{2\pi})$, the flip probability
satisfies
\begin{equation}\label{eq:flip_prob}
  \P(G_l \in J_l \mid \mathbf Z)
  \leq \frac{|\Delta_{l,k}|}{\sigma_l \sqrt{2\pi}}
  = \frac{|\delta_{l,k}|}{\sqrt{2\pi}}.
\end{equation}

\paragraph{Near-unbiasedness of the other coordinate.} Conditional on $G_1 = g$, we have
\[
  G_2 \sim N\bigl( \rho_n (\sigma_2/\sigma_1) g, \sigma_2^2 (1-\rho_n^2) \bigr)
\]
so, writing $\Phi$ for the standard normal CDF,
\[
  \E\bigl[ \mathrm{sign}(D_2) \mid G_1 = g, \mathbf Z \bigr]
  = 2 \Phi\!\left( \frac{\Delta_{2,k} + \rho_n (\sigma_2/\sigma_1) g}
                        {\sigma_2 \sqrt{1-\rho_n^2}} \right) - 1,
\]
and the Lipschitz bound $|2\Phi(x) - 1| \leq \sqrt{2/\pi}|x|$ gives
\begin{equation}\label{eq:unbiased_other_coord}
  \bigl| \E\bigl[ \mathrm{sign}(D_2) \mid G_1 = g, \mathbf Z \bigr] \bigr|
  \leq \sqrt{\frac{2}{\pi}} \cdot
       \frac{|\delta_{2,k}| + |\rho_n| |g| / \sigma_1}{\sqrt{1-\rho_n^2}}.
\end{equation}

\paragraph{First term of \eqref{eq:binary_telescope}.} By (i), this term
equals $2\mathrm{sign}(\Delta_{1,k}) \One{G_1 \in J_1}\mathrm{sign}(D_2)$
almost surely. Conditioning on $G_1$, noting that $|g| \leq |\Delta_{1,k}|$ on
$J_1$ (so that $|g|/\sigma_1 \leq |\delta_{1,k}|$
in~\eqref{eq:unbiased_other_coord}), and applying~\eqref{eq:flip_prob}
and~\eqref{eq:unbiased_other_coord},
\[
\begin{aligned}
  &\bigl| \E\bigl[ \bigl( \mathrm{sign}(D_1) - \mathrm{sign}(G_1) \bigr)
                  \mathrm{sign}(D_2) \mid \mathbf Z \bigr] \bigr|\\
  &\qquad\leq 2 \cdot \frac{|\delta_{1,k}|}{\sqrt{2\pi}}
       \cdot \sqrt{\frac{2}{\pi}}
       \cdot \frac{|\delta_{2,k}| + |\rho_n| |\delta_{1,k}|}{\sqrt{1-\rho_n^2}}\\
  &\qquad= \frac{2}{\pi} \cdot
    \frac{|\delta_{1,k}| \bigl( |\delta_{2,k}| + |\rho_n| |\delta_{1,k}| \bigr)}
         {\sqrt{1-\rho_n^2}}.
\end{aligned}
\]

\paragraph{Second term of \eqref{eq:binary_telescope}.} Symmetrically, it
equals $2\mathrm{sign}(\Delta_{2,k}) \One{G_2 \in J_2}\mathrm{sign}(G_1)$
almost surely; now condition on $G_2 = g$, under which
$G_1 \sim N\bigl( \rho_n (\sigma_1/\sigma_2) g,\ \sigma_1^2(1-\rho_n^2) \bigr)$.
Since only the \emph{driftless} $\mathrm{sign}(G_1)$ appears here, the
analogue of~\eqref{eq:unbiased_other_coord} carries no $\delta_{1,k}$ term:
\[
  \bigl| \E[\mathrm{sign}(G_1) \mid G_2 = g, \mathbf Z] \bigr|
  \leq \sqrt{\frac{2}{\pi}} \cdot
       \frac{|\rho_n| |g| / \sigma_2}{\sqrt{1-\rho_n^2}}
  \leq \sqrt{\frac{2}{\pi}} \cdot
       \frac{|\rho_n| |\delta_{2,k}|}{\sqrt{1-\rho_n^2}}
  \quad \text{on } J_2,
\]
which gives
\[
  \bigl| \E\bigl[ \mathrm{sign}(G_1)
                  \bigl( \mathrm{sign}(D_2) - \mathrm{sign}(G_2) \bigr)
                  \mid \mathbf Z \bigr] \bigr|
  \leq \frac{2}{\pi} \cdot \frac{|\rho_n| \delta_{2,k}^2}{\sqrt{1-\rho_n^2}}.
\]
Summing the bounds above and using $1 - \rho_n^2 \geq 1 - \bar\rho^2$ yields the first inequality of the lemma; the second follows from the fact that
\[
  \delta_{l,k}^2
  = \frac{\Delta_{l,k}^2}{\sigma_l^2}
  = \frac{4 S_k^{(l)}}{2 \bigl( \beta_{l,n}^2 + 1 \bigr)}
  = \frac{2 S_k^{(l)}}{\beta_{l,n}^2 + 1}
  \leq 2 S_k^{(l)},
  \qquad l \in \{1, 2\}.
\]
which gives $|\delta_{1,k}\delta_{2,k}| \leq 2 \sqrt{S_k^{(1)} S_k^{(2)}}$ and $\delta_{l,k}^2 \leq 2 S_k^{(l)}$.
\end{proof}

\begin{lemma}
\label{lem:permutation_variance}
Let $m\geq 2$, let $\sigma\sim\mathrm{Unif}(\mathcal S_m)$, and let $\tau$ be independent of $\sigma$ and uniformly distributed over the $\binom{m}{2}$ transpositions in $\mathcal S_m$. For any function $f:\mathcal S_m\to\mathbb R$,
\[
  \Var \left(f(\sigma)\right)
  \leq
  \frac{m-1}{4}
  \E\left[
    \left(f(\sigma)-f(\sigma\circ\tau)\right)^2
  \right].
\]
\end{lemma}

\begin{proof}
  For $1\leq a<b\leq m$, let $\tau_{ab}$ denote the transposition exchanging $a$ and $b$. Since $\tau$ is uniform over all transpositions, it is equivalent to prove
  \begin{equation} \label{eq:permutation_variance_unnormalized}
    \Var \left(f(\sigma)\right)
    \leq
    \frac{1}{2m}
    \sum_{1\leq a<b\leq m}
    \E\left[
      \left(f(\sigma)-f(\sigma\circ\tau_{ab})\right)^2
    \right].
  \end{equation}
  since
  \[
    \frac{m-1}{4}\cdot\frac{1}{\binom{m}{2}}
    =\frac{1}{2m}.
  \]

  Subtracting the constant $\E[f(\sigma)]$ leaves all differences on the right-hand side unchanged, so assume without loss of generality that $\E[f(\sigma)] = 0$. Write also $V:=\E[f(\sigma)^2]$. We prove \eqref{eq:permutation_variance_unnormalized} by induction on
  $m$.

  For the base case $m=2$, the group $\mathcal S_2$ consists of two elements and there is only one transposition. Direct computation gives
  \[
    V
    =\frac14\left(f(\mathrm{Id})-f(\tau_{12})\right)^2
    =\frac14\E\left[
      \left(f(\sigma)-f(\sigma\circ\tau_{12})\right)^2
    \right],
  \]
  which is exactly \eqref{eq:permutation_variance_unnormalized}. We next establish an auxiliary bound conditional on the image of a particular index. For $r,x\in[m]$, define
  \[
    g_r(x):=\E\big[f(\sigma)\mid \sigma(r)=x\big] \quad
    \text{ and } \quad
    W_r:=\E\left[g_r(\sigma(r))^2\right].
  \]
  We claim that
  \begin{equation}
    \label{eq:permutation_projection_bound}
    \sum_{r=1}^m W_r\leq\frac{m}{m-1}V.
  \end{equation}
  Since $\sigma(r)$ is uniform on $[m]$,
  \[
    W_r=\frac1m\sum_{x=1}^m g_r(x)^2.
  \]
  Moreover, because $f$ is centered,
  \[
    \sum_{x=1}^m g_r(x)=0
    \qquad\text{for every }r\in[m].
  \]
  Define
  \[
    G(\sigma):=\sum_{r=1}^m g_r(\sigma(r)),
    \qquad
    S:=\sum_{r=1}^m W_r.
  \]
  By the tower law,
  \begin{align*}
    \E\big[f(\sigma)G(\sigma)\big]
    &=\sum_{r=1}^m
      \E\big[f(\sigma)g_r(\sigma(r))\big]\\
    &=\sum_{r=1}^m
      \E\left[
      \E[f(\sigma)\mid\sigma(r)]g_r(\sigma(r))
      \right]\\
    &=\sum_{r=1}^m
      \E\left[g_r(\sigma(r))^2\right]
      =S.
  \end{align*}

  For $r\neq s$, the pair $(\sigma(r),\sigma(s))$ is uniform over the $m(m-1)$ ordered pairs $(x,y)$ with $x\neq y$. Hence
  \begin{align*}
    &\E\big[g_r(\sigma(r))g_s(\sigma(s))\big]\\
    &\qquad=
      \frac{1}{m(m-1)}
      \sum_{x\neq y}g_r(x)g_s(y)\\
    &\qquad=
      -\frac{1}{m(m-1)}
      \sum_{x=1}^m g_r(x)g_s(x),
  \end{align*}
  where the last equality uses
  $\sum_xg_r(x)=\sum_yg_s(y)=0$. Summing the diagonal and off-diagonal
  terms gives
  \begin{align*}
    \E\left[G(\sigma)^2\right]
    &=S-
      \frac{1}{m(m-1)}
      \sum_{x=1}^m\sum_{r\neq s}
      g_r(x)g_s(x)\\
    &=\frac{m}{m-1}S-
      \frac{1}{m(m-1)}
      \sum_{x=1}^m
      \left(\sum_{r=1}^m g_r(x)\right)^2\\
    &\leq\frac{m}{m-1}S.
  \end{align*}
  Therefore, by Cauchy--Schwarz,
  \[
    S^2
    =\big(\E[f(\sigma)G(\sigma)]\big)^2
    \leq
    \E\left[f(\sigma)^2\right]\,
    \E\left[G(\sigma)^2\right]
    \leq\frac{m}{m-1}VS.
  \]
  If $S=0$, \eqref{eq:permutation_projection_bound} is immediate; otherwise, dividing by $S$ proves \eqref{eq:permutation_projection_bound}.

  Now assume that \eqref{eq:permutation_variance_unnormalized} holds for uniform permutations of $m-1$ elements. Fix $r\in[m]$. Conditional variance decomposition gives
  \begin{equation}
    \label{eq:permutation_conditional_variance}
    V-W_r
    =\E\left[
      \left(f(\sigma)-\E[f(\sigma)\mid\sigma(r)]\right)^2
    \right].
  \end{equation}
  Fix $x\in[m]$ and condition on $\sigma(r)=x$. The value at position
  $r$ is then fixed, while the remaining $m-1$ values appear in a
  uniformly random order in the remaining $m-1$ positions. This is
  exactly the setting of the induction hypothesis with $m-1$ in place
  of $m$. Moreover, if $a,b\neq r$, then $\tau_{ab}$ leaves $r$
  unchanged, so
  \[
    (\sigma\circ\tau_{ab})(r)=\sigma(r)=x.
  \]
  Thus the induction hypothesis can be applied conditionally on
  $\sigma(r)=x$, using the transpositions $\tau_{ab}$ with
  $a,b\neq r$. Averaging the resulting inequality over $x$ yields
  \[
    V-W_r
    \leq
    \frac{1}{2(m-1)}
    \sum_{\substack{1\leq a<b\leq m\\a,b\neq r}}
    \E\left[
      \left(f(\sigma)-f(\sigma\circ\tau_{ab})\right)^2
    \right].
  \]
  Summing over $r\in[m]$, each pair $a<b$ appears exactly $m-2$ times,
  so
  \begin{equation}
    \label{eq:permutation_induction_upper}
    mV-\sum_{r=1}^mW_r
    \leq
    \frac{m-2}{2(m-1)}
    \sum_{1\leq a<b\leq m}
    \E\left[
      \left(f(\sigma)-f(\sigma\circ\tau_{ab})\right)^2
    \right].
  \end{equation}
  On the other hand, \eqref{eq:permutation_projection_bound} implies
  \begin{equation}
    \label{eq:permutation_induction_lower}
    mV-\sum_{r=1}^mW_r
    \geq
    mV-\frac{m}{m-1}V
    =\frac{m(m-2)}{m-1}V.
  \end{equation}
  Combining \eqref{eq:permutation_induction_upper} and \eqref{eq:permutation_induction_lower}, and cancelling the common
  factor $(m-2)/(m-1)$ for $m\geq3$, gives
  \[
    V
    \leq
    \frac{1}{2m}
    \sum_{1\leq a<b\leq m}
    \E\left[
      \left(f(\sigma)-f(\sigma\circ\tau_{ab})\right)^2
    \right].
  \]
  This is \eqref{eq:permutation_variance_unnormalized}, completing the induction.
\end{proof}

\end{document}